\documentclass[10pt,fleqn]{article}

\usepackage[margin=1.25in]{geometry}
\usepackage{soul}
\usepackage{booktabs}
\usepackage{microtype}
\usepackage{graphicx}
\usepackage{enumitem}
\usepackage[table]{xcolor}  
\usepackage{amsmath,amssymb,amsthm}
\usepackage{amsfonts}
\usepackage{lineno}
\usepackage{bbm}
\usepackage{nicefrac}
\usepackage{tikz}
\usepackage{bm}
\usepackage[noend]{algpseudocode}
\usepackage{algorithm}
\algnewcommand{\IfThenElse}[3]{
  \State \algorithmicif\ #1\ \algorithmicthen\ #2\ \algorithmicelse\ #3}

\renewcommand{\epsilon}{\varepsilon}
\usetikzlibrary{arrows}
\usetikzlibrary{calc}
\newtheorem{lemma}{Lemma}
\newtheorem{theorem}{Theorem}
\newtheorem{definition}{Definition}

\newtheorem{observation}{Observation}

\newcommand{\field}[1]{\mathbb{#1}}
\newcommand{\R}{\field{R}}
\newcommand{\E}{\field{E}}
\renewcommand{\Pr}{\field{P}}
\newcommand{\Var}{\operatorname{Var}}
\newcommand{\var}{\Var}
\renewcommand{\epsilon}{\varepsilon}

\newcommand{\scD}{\mathcal{D}}
\newcommand{\scE}{\mathcal{E}}
\newcommand{\scG}{\mathcal{G}}

\newcommand{\scO}{O}

\newcommand{\scU}{\mathcal{U}}
\newcommand{\scV}{\mathcal{V}}
\newcommand{\scX}{\mathcal{X}}
\newcommand{\Va}{B}
\newcommand{\ev}{\scE}

\newcommand{\poly}{\operatorname{poly}}

\newcommand{\vol}{\operatorname{vol}}
\newcommand{\tvd}[2]{\text{tvd}(#1, #2)}
\newcommand{\norm}[2]{\| {#2} \|_\textsc{#1}}

\newcommand{\ba}{\pmb{b}}
\newcommand{\dmax}{\Delta}
\newcommand{\dmin}{\delta}

\newcommand{\gph}{g}

\newcommand{\cond}{\Phi}
\newcommand{\tmix}{t}

\renewcommand{\hat}{\widehat}
\renewcommand{\bar}{\overline}

\newcommand{\gk}{\mathcal{G}}
\newcommand{\vk}{\mathcal{V}}
\newcommand{\ek}{\mathcal{E}}

\newcommand{\subG}{\mathcal{S}}
\newcommand{\subGw}{\subG'}

\newcommand{\assign}{=}
\newcommand{\GraphSort}{\textsc{Apx-DD}}
\newcommand{\Preprocess}{\textsc{DD}}
\newcommand{\epsSampler}{\textsc{Sample}}
\newcommand{\sampleS}{\textsc{Rand-Grow}}
\newcommand{\Sample}{\textsc{Sample}}
\newcommand{\EpsilonSampleS}{\textsc{Apx}-\sampleS}
\newcommand{\computeP}{\textsc{Prob}}
\newcommand{\epsProbCompute}{\textsc{Apx}-\computeP}
\newcommand{\UniformAlgo}{\textsc{Ugs}}
\newcommand{\EpsilonAlgo}{\textsc{Apx}-\UniformAlgo}
\newcommand{\ComparisonAlgo}{\UniformAlgo\textsc{-Compare}}
\newcommand{\CutEstim}{\textsc{EstimateCuts}}

\newcommand{\MCsampler}{\textsc{Rwgs}}

\newcommand{\dg}[2]{d({#1}|{#2})}
\newcommand{\cut}{c}
\newcommand{\Cut}{\operatorname{Cut}}

\newcommand{\kok}{k^{\scO(k)}}
\newcommand{\kokm}{k^{-\scO(k)}}

\renewcommand{\lg}{\log}
\newcommand{\Ind}[1]{ \field{I}\left\{{#1}\right\} }

\makeatletter
\newtheorem*{rep@theorem}{\rep@title}
\newcommand{\newreptheorem}[2]{%
\newenvironment{rep#1}[1]{%
 \def\rep@title{#2 \ref{##1}}%
 \begin{rep@theorem}}%
 {\end{rep@theorem}}}
\makeatother
\newreptheorem{theorem}{Theorem}
\newreptheorem{lemma}{Lemma}

\begin{document}
\newtheorem{claim}{Claim}

\makeatletter
\let\@fnsymbol\@arabic
\makeatother

\title{Efficient and Near-Optimal Algorithms for Sampl\-ing\\ Small Con\-nected Subgraphs\footnote{A short version of these results appeared in the Proceedings of ACM STOC 2021~\cite{Bressan21STOC}.}}
\author{Marco Bressan\\Dipartimento di Informatica\\Università degli Studi di Milano\\marco.bressan@unimi.it}

\maketitle

\begin{abstract}
We study the following problem: given an integer $k \ge 3$ and a simple graph $G$, sample a connected induced $k$-node subgraph of $G$ uniformly at random. This is a fundamental graph mining primitive with applications in social network analysis, bioinformatics, and more. Surprisingly, no efficient algorithm is known for uniform sampling; the only somewhat efficient algorithms available yield samples that are only approximately uniform, with running times that are unclear or suboptimal. In this work we provide: (i) a near-optimal mixing time bound for a well-known random walk technique, (ii) the first efficient algorithm for truly uniform graphlet sampling, and (iii) the first sublinear-time algorithm for $\epsilon$-uniform graphlet sampling.
\end{abstract}

\section{Introduction}
A $k$-graphlet of a graph $G$ is a connected and induced $k$-vertex subgraph of $G$. Starting with triangles and wedges, and the discovery of triadic closure in social graphs~\cite{easley_kleinberg_2010}, graphlets have become a central subject of study in social network analysis~\cite{bonato2014dimensionality,Ugander&2013}, clustering~\cite{li2017motif,Babis17clustering}, and bioinformatics~\cite{Alon&2008,Chen06genome,prvzulj2007biological}; and they have found application in the development of graph kernels~\cite{Shervashidze2009}, graph embeddings~\cite{Tu2019-gl2vec} and graph neural networks~\cite{Peng2020-GNN}. The underlying idea is that, in many cases, the distribution of $k$-graphlets (the relative number of cliques, stars, paths, and so on) holds fundamental information about the nature of a complex network~\cite{Milo824}. Understandably, these findings have sparked research on several basic graphlet mining problems such as finding, counting, listing, and sampling graphlets.

In this work we consider the two following problems. The \emph{uniform graphlet sampling problem} asks, given $G$ and $k$, to return a $k$-graphlet uniformly at random from the set $\scV_k$ of all $k$-graphlets of $G$. The $\varepsilon$-\emph{uniform graphlet sampling problem} asks, given $G,k$ and $\varepsilon > 0$, to return a $k$-graphlet from any distribution whose total variation distance from the uniform distribution over $\scV_k$ is at most $\epsilon$. Clearly, an efficient algorithm for one of these problems yields an efficient algorithm for estimating the $k$-graphlet distribution. For this reason uniform and $\epsilon$-uniform graphlet sampling have been investigated for almost a decade, both in theory and in practice~\cite{Agostini19mixing,Bhuiyan&2012,Bressan&2017,Bressan&2018b,Bressan2019-VLDB,Bressan21TKDD,Chen&2016,Han&2016,Gionis&2020mixing,Paramonov19-lifting,Saha&2015MCMC,Wang&2014}.

Unfortunately, although sampling a random $k$-vertex subgraph of $G=(V,E)$ uniformly at random is trivial, sampling a graphlet is considerably more challenging, due to the fact that a graphlet is \emph{connected}. Let $n=|V|$ and $m=|E|$. For uniform graphlet sampling, to date no algorithm is known that runs in less then $\Theta(n+m)$ time per sample. The only somewhat efficient algorithms known are for $\epsilon$-uniform graphlet sampling, and they can be divided into direct sampling algorithms (that do not have a preprocessing phase) and two-phase sampling algorithms (which have a preprocessing phase and a sampling phase). We now discuss those algorithms briefly. Here and in what follows we assume that $n+m \gg k$, so that a running time of $2^{\scO(k)} (n+m)$ or $\kok (n+m)$ is better than a running time of $\scO((n+m)^2)$. This reflects the fact that, today, real-world graphs can easily have billions of edges, but $k$ rarely exceeds $5$ or $10$.

For direct sampling algorithms, the state of the art is the so-called $k$-graphlet walk. To begin, consider the graph $\gk_k=(\scV_k,\scE_k)$ whose vertices are the $k$-graphlets of $G$, and where there is an edge between two graphlets if their intersection is a $(k-1)$-graphlet. The $k$-\emph{graphlet walk} is the lazy random walk over $\gk_k$. It is not hard to show that, if $G$ is connected, this walk is ergodic and so converges to a stationary distribution. Thus, to obtain $\epsilon$-uniform graphlets, one can run the walk until it comes $\epsilon$-close to its stationary distribution, and then use rejection sampling. This technique is extensively used thanks to its simplicity and elegance~\cite{Agostini19mixing,Bhuiyan&2012,Chen&2016,Han&2016,Gionis&2020mixing,Saha&2015MCMC,Wang&2014}; the drawback is that its running time depends on $\tmix_{\epsilon}(\gk_k)$, the $\epsilon$-mixing time of the walk, which can range anywhere from $\Theta(1)$ to $\Theta(n^{k-1})$~\cite{Bressan&2017,Bressan&2018b}. Indeed, the analysis of $\tmix_{\epsilon}(\gk_k)$ is nontrivial, and between the best lower and upper bounds there is still a multiplicative gap of $\Delta^{k-1}$~\cite{Agostini19mixing}. %In practice, too, the time taken by the walk to approach its stationary distribution varies heavily between graphs of similar size~\cite{Bressan&2017,Bressan&2018b}.

For two-phase algorithms, the state of the art is an extension of the color coding technique of~\cite{Alon&1995}, proposed in~\cite{Bressan&2017}. This extension allows one to sample uniformly from the graphlets of $G$ that are made \emph{colorful} by a random $k$-coloring of the vertices of $G$. The algorithm has a preprocessing phase taking time $2^{\scO(k)} (n+m)$ and space $2^{\scO(k)} n$, and expected sampling time $\kok \Delta$, and by increasing the space to $2^{\scO(k)} (n+m)$, one can reduce the expected sampling time to $\kok$. It is not hard to show that, by increasing the preprocessing time and space to $2^{\scO(k)} (n+m) \lg \frac{1}{\epsilon}$, one can take $\epsilon$-uniform graphlet samples in expected time $\kok \big(\lg \frac{1}{\epsilon}\big)^2$ per sample. This algorithm however looks like an overkill for sampling, which makes one suspect that a faster algorithm is possible. We also observe that the fastest algorithm known for uniform graphlet sampling consists in running the color-coding extension above from scratch for every sample, using $2^{\scO(k)} (n+m)$ time per sample.

In conclusion, (i) we do not have tight bounds for the $k$-graphlet walk, (ii) we do not have an efficient algorithm for uniform graphlet sampling, and (iii) we do not know if the existing algorithms for $\epsilon$-uniform graphlet sampling are optimal. The goal of our work is to reduce this gap.

\section{Results}
We give three contributions. First, we settle the mixing time of the $k$-graphlet walk up to multiplicative $\kok \lg^2 n$ factors. Second, we present the first efficient algorithm for uniform graphlet sampling, with a preprocessing linear in $n+m$ and an expected sampling time $\kok \lg \Delta$. Third, we give the first $\epsilon$-uniform graphlet sampling algorithm with sampling time independent of $G$, and preprocessing time $\scO(n \lg n)$, which is sublinear in $m$ as long as $m = \omega(n \lg n)$. The rest of this section overviews these results; later sections give the proofs.

%a detailed walk-through is given respectively in \mbox{Section~\ref{sec:walks}}, \mbox{Section~\ref{sec:ugs}}, and \mbox{Section~\ref{sec:epsugs}}.

\subsection{Near-optimal mixing time bounds for the k-graphlet walk}
\label{sub:walk}
Recall the graph $\gk_k$ defined above, and let $\tmix_{\epsilon}(\gk_k)$ denote its $\epsilon$-mixing time (see Section~\ref{sub:prelim} for a formal definition); similarly, let $\tmix_{\epsilon}(G)$ be the $\epsilon$-mixing time of $G$. Moreover, let $\tmix=\tmix_{\frac{1}{4}}$; it is well-known that $\tmix_{\epsilon}=\scO(\tmix \cdot \lg \frac{1}{\epsilon})$, hence bounds on $\tmix$ yield bounds on $\tmix_{\epsilon}$ for all $0 < \epsilon \le \frac{1}{4}$. Finally, let $\rho(G) = \frac{\dmax}{\dmin}$ be the ratio between the largest and the smallest degree of $G$. We prove:
\begin{theorem}
\label{thm:tgk}
For all graphs $G$ and all $k \ge 2$,
\begin{align}
 \tmix(\gk_k) \le \tmix(G) \cdot \kok \rho(G)^{k-1}\lg n
\end{align}
Moreover, for any function $\rho(n) \in \Omega(1) \cap \scO(n)$ there exists a family of arbitrarily large graphs $G$ on $n$ vertices that satisfy $\rho(G) = \Theta(\rho(n))$ and
\begin{align}
\tmix(\gk_k) \ge \tmix(G) \cdot \kokm \rho(G)^{k-1} / \,\lg n
\end{align}
\end{theorem}
Essentially, Theorem~\ref{thm:tgk} says that the lazy walk on $\gk_k$ behaves like the lazy walk on $G$ slowed down by a factor $\rho(G)^{k-1}$. This should be compared with the upper and lower bound of~\cite{Agostini19mixing}, which are respectively $\tmix(G) \, \tilde{\scO}\big(\rho(G)^{2(k-1)}\big)$ and $\tmix(G) \,\Omega\big(\rho(G)^{k-1}\delta^{-1}\big)$. Ignoring $\kok \poly \lg n$ factors, we improve those bounds by $\rho(G)^{k-1}$ and $\delta$ respectively.

From Theorem~\ref{thm:tgk}, we obtain the best bounds known for $\epsilon$-uniform graphlet sampling based on random walks:
\begin{theorem}
\label{thm:walk_algo}
There exists a random-walk based algorithm that, for all $G$, all $k \ge 2$, and all $\epsilon > 0$, returns an $\epsilon$-uniform $k$-graphlet from $G$ in expected time $\kok \tmix(G) \, \rho(G)^{k-2}\lg \frac{n}{\epsilon}$.
\end{theorem}
Note that, although $\tmix(\gk_k)$ grows with $\rho(G)^{k-1}$, the bound above grows with $\rho(G)^{k-2}$. The reason is that, as noted in~\cite{Gionis&2020mixing,Wang&2014}, sampling $k$-graphlets is equivalent to sampling the edges of $\gk_{k-1}$. So, we can run the walk over $\gk_{k-1}$ rather than over $\gk_{k}$, which yields a mixing time proportional to $\rho(G)^{k-2}$ rather than $\rho(G)^{k-1}$. As a sanity check, when $k=2$ our algorithm matches the natural bound $\scO(\tmix(G))$ achieved by the simple random walk over $G$.

Regarding the techniques, our proofs are very different from those of~\cite{Agostini19mixing}. There, the authors showed a mapping between the cuts of $\gk_k$ and those of $G$; this allowed them to bound the conductance of $\gk_k$ by a function of the conductance of $G$, and then bound $\tmix_{\epsilon}(\gk_k)$ via Cheeger's inequality. However, since Cheeger's inequality can be loose by a quadratic factor, their upper bound on $\tmix_{\epsilon}(\gk_k)$ grows with $\rho(G)^{2(k-1)}$ instead of $\rho(G)^{k-1}$, see above. To avoid this, here we establish a connection between the \emph{relaxation times} of $\gk_i$ and $\gk_{i+1}$, for all $i=1,\ldots,k-1$, and thus between $\gk_1=G$ and $\gk_k$. To this end we prove a technical result on the relaxation time of the lazy walk on the \emph{line graph} of $G$ (the graph encoding the adjacencies between the edges of $G$):
\begin{lemma}
\label{lem:Lg}
Any graph $G$ satisfies $\tau(L(G)) \le 20\, \rho(G)\, \tau(G)$, where $L(G)$ is the line graph of $G$ and $\tau(\cdot)$ denotes the relaxation time of the lazy random walk.
\end{lemma}

\subsection{Uniform graphlet sampling}
We describe the first efficient algorithm for \emph{uniform graphlet sampling}:
\begin{theorem}
\label{thm:uniform} 
There exists a two-phase graphlet sampling algorithm, \UniformAlgo\ $($\textsc{u}niform \textsc{g}raphlet \textsc{s}ampler$)$, with the following guarantees:
\begin{enumerate}[leftmargin=2em]
    \item the preprocessing phase runs in time $\scO(n \,k^2 \lg k + m)$ and space $\scO(n)$
    \item the sampling phase returns $k$-graphlets independently and uniformly at random in $k^{\scO(k)} \lg \Delta$ expected time per sample.
\end{enumerate}
\end{theorem}
\noindent
The technique behind \UniformAlgo\ is radically different from random walks and color coding. The key idea is to ``regularize'' $G$, that is, to sort $G$ so that each vertex $v$ has maximum degree in the subgraph $G(v)$ induced by $v$ and all vertices after it (this can be done by just repeatedly removing the maximum-degree vertex from $G$). As we show, this makes each $G(v)$ behave like a regular graph, which makes it efficient to perform rejection sampling of randomly grown spanning trees. It is worth noting that several attempts have been made to sample graphlets uniformly by growing random subsets and applying rejection sampling, see for instance~\cite{Jha&2015,Paramonov19-lifting}. All those algorithms, however, have one crucial limitation: in the worst case, the rejection probability approaches $1-\dmax^{-k+1}$, in which case roughly $\dmax^{k-1}$ rejection trials are needed to draw a single graphlet. It is somewhat surprising that the fact that just sorting $G$ solves the problem has gone unnoticed until now.

\UniformAlgo\ can also be used as a graphlet \emph{counting} algorithm:
\begin{theorem} % TODO bring m out of the \lg \Delta factor?
\label{thm:counting_1}
Choose any $\epsilon_0,\epsilon_1,\delta \in (0,1)$. There exists an algorithm that runs in time
\begin{align}
 \scO(m) + \kok \left(\frac{n}{\epsilon_0^2} \lg\frac{n}{\delta} + \frac{1}{\epsilon_1^2} \lg\frac{1}{\delta}\right)\lg \dmax 
\end{align}
and, with probability $1-\delta$, returns for every distinct (up to isomorphism) connected $k$-vertex graph $H$ an additive $(\epsilon_0 N_H + \epsilon_1 N_k)$-approximation of $N_H$, where $N_H$ is the number of graphlets of $G$ isomorphic to $H$, and $N_k=|\scV_k|$ is the total number of $k$-graphlets in $G$.
\end{theorem}

\subsection{Epsilon-uniform graphlet sampling}
We present:
\begin{theorem}
\label{thm:epsuniform}
There exists a two-phase graphlet sampling algorithm, \EpsilonAlgo, that for all $\epsilon > 0$ has the following guarantees:
\begin{enumerate}[leftmargin=2em]
    \item the preprocessing phase takes time $\scO\left(\big(\frac{1}{\epsilon}\big)^{\frac{2}{(k-1)}} k^6 \, n \lg n \right)$ and space $\scO(n)$
    \item with high probability over the preprocessing phase, the sampling phase returns $k$-graphlets independently and $\epsilon$-uniformly at random in $k^{\scO(k)} \big(\frac{1}{\epsilon}\big)^{8+\frac{4}{(k-1)}} \lg \frac{1}{\epsilon}$ expected time per sample.
\end{enumerate}
\end{theorem}
\noindent
The remarkable fact about \EpsilonAlgo\ is that its preprocessing time grows as $n \lg n$, and is therefore independent of the edge set of $G$. This should be contrasted with the color-coding algorithm, whose preprocessing time grows as $n+m$. Moreover, our preprocessing time is polynomial in both $\frac{1}{\epsilon}$ and $k$, while that of color coding is exponential in $k$. For what concerns the expected sampling time, like the one of color coding ours is independent of $G$, but it pays an extra $\poly \frac{1}{\epsilon}$ factor. However, we did not make hard attempts to optimize those factors, and they might be improved.

While \UniformAlgo\ is rather simple, \EpsilonAlgo\ is considerably more involved. The high-level idea is, unsurprisingly, to ``approximate'' \UniformAlgo\ in both phases. However, this turns out to be a delicate issue, which requires a careful combination of graph sketching, cut size estimation, and coupling arguments. The reason is that \UniformAlgo\ relies crucially on a particular topological order of $G$, whose exact computation takes time $\Omega(m)$, and which is not clear how to approximate in time $o(m)$. In fact, it is not even clear what definition of ``approximate order'' is the right one for our purposes; in the end, the definition we use turns out to be nontrivial. %Note also that, for a certain range of the parameters (e.g., when $m$ is sufficiently small), the preprocessing time of \UniformAlgo\ is lower than that of \EpsilonAlgo, hence one may want to use \UniformAlgo\ instead of \EpsilonAlgo.

To conclude, we observe that \EpsilonAlgo\ is nearly optimal in our graph access model:
\begin{theorem}
\label{thm:eps_lb}
For any $k \ge 2$ and any $\epsilon \in [0,1]$, any $\epsilon$-uniform $k$-graphlet sampling algorithm has worst-case expected running time $\Omega(n/k)$ in the graph access model of~\cite{Kaufman04-model}.
\end{theorem}
\begin{proof}
Let $G$ contain a $k$-path plus $n-k$ isolated vertices. In the worst case any algorithm must examine $\Omega(n/k)$ vertices in expectation before finding the only $k$-graphlet of $G$.
\end{proof}

\medskip
The table below summarizes our upper bounds and the state of the art.% For the bounds of~\cite{Bressan&2017} see the extended version of this abstract.
\begin{table*}[h!]
\centering\small
\caption{Our upper bounds (shaded) compared to existing work.}%\vspace*{5pt}
\renewcommand{\arraystretch}{1.5}
\scalebox{.96}{
\begin{tabular}{lllll} \toprule
& preprocessing time & preprocessing space & time per sample & output \\
\midrule
\text{\cite{Bressan&2017}} & -- & -- & $2^{\scO(k)}(n+m) + \kok$ & uniform \\
\rowcolor{lightgray!50}
\UniformAlgo\ & $\scO(n k^2 \lg k + m)$ & $\scO(n)$ & $k^{\scO(k)} \lg \dmax$ & uniform \\
%\midrule
\text{\cite{Gionis&2020mixing}}  & $\scO(n)$ & $\scO(n)$ & $\kok \big(\Delta \lg \frac{n}{\epsilon}\big)^{k-3}$  & $\epsilon$-uniform \\
\text{\cite{Bressan&2017}}  & $2^{\scO(k)}(n+m) \lg \frac{1}{\epsilon}+\kok \frac{1}{\epsilon^2}\lg \frac{1}{\epsilon}$ & $2^{\scO(k)}n \lg \frac{1}{\epsilon}$ & $\kok \Delta (\lg \frac{1}{\epsilon})^2$ & $\epsilon$-uniform \\
\text{\cite{Bressan&2017}} & $2^{\scO(k)}(n+m) \lg \frac{1}{\epsilon}+ \kok \frac{1}{\epsilon^2}\lg \frac{1}{\epsilon}$& $2^{\scO(k)}m \lg \frac{1}{\epsilon}$ & $\kok (\lg \frac{1}{\epsilon})^2$ & $\epsilon$-uniform \\
\rowcolor{lightgray!50}
\EpsilonAlgo\ & $ \scO\big(\big(\frac{1}{\epsilon}\big)^{\frac{2}{(k-1)}} k^6 n \lg n \big)$ & $\scO(n)$  & $\kok \big(\frac{1}{\epsilon}\big)^{8+\frac{4}{(k-1)}} \lg \frac{1}{\epsilon}$ &  $\epsilon$-uniform\\
\rowcolor{lightgray!50}
\MCsampler\  & -- & -- & $\kok \tmix(G) (\frac{\dmax}{\delta}\big)^{k-2}\lg \frac{n}{\epsilon}$  & $\epsilon$-uniform 
\end{tabular}
}
%}
\label{tab:summary}
\end{table*}

%\section{Manuscript Organization of the manuscript}

\section{Related work}
The $k$-graphlet walk algorithm was introduced by~\cite{Bhuiyan&2012} without formal running time bounds. The first bounds on $\tmix_{\epsilon}(\gk_k)$ were given by~\cite{Bressan&2017}, while the first bounds tying $\tmix_{\epsilon}(\gk_k)$ to $\tmix_{\epsilon}(G)$ were given by~\cite{Agostini19mixing}. Recently,~\cite{Gionis&2020mixing} developed a graphlet sampling random walk with running time $\kok (\Delta \lg \frac{n}{\epsilon})^{k-3}$. Their approach is similar to ours as they build the $k$-graphlet walk recursively from the $(k-1)$-graphlet walk. However, they assume one can sample edges uniformly at random from $G$ in time $\scO(1)$, which requires a $\scO(n)$-time preprocessing, or an additional factor of $\tmix_{\epsilon}(G)$ to sample edges via random walks. Moreover, their running bound grows like $\Delta^{k}$, while ours grows as $\big(\frac{\dmax}{\dmin}\big)^{k}$.

The color coding extension for estimating graphlet  counts was introduced by~\cite{Bressan&2017}. This extension does not allow to $\epsilon$-uniform graphlet sampling directly; however, it can be obtained by making several independent runs, for a total preprocessing time of $2^{\scO(k)}(n+m) \lg \frac{1}{\epsilon}+ \kok \frac{1}{\epsilon^2}\lg \frac{1}{\epsilon}$, a preprocessing space of $2^{\scO(k)}m \lg \frac{1}{\epsilon}$, and an expected sampling time of $\kok (\lg \frac{1}{\epsilon})^2$. See Appendix~\ref{apx:cc} for a complete proof. As said, one can also obtain uniform samples by running the entire algorithm of~\cite{Bressan&2017} from scratch, in $2^{\scO(k)} \scO(n+m)$ time per sample.

Rejection sampling is at the heart of several graphlet sampling algorithms, such as path sampling~\cite{Jha&2015} and lifting~\cite{Paramonov19-lifting}. These algorithms start by drawing a random vertex from $G$ and, then, repeatedly selecting random edges in the cut. This technique alone seems destined to fail: in the worst case, the rejection probability must be as large as $\simeq 1-\dmax^{-k+1}$, resulting in a vacuous $\scO(\dmax^{k-1})$ running time bound. The main idea behind our algorithms is to make such a rejection sampling efficient by sorting $G$ so to virtually ``bucket'' the graphlets, so that within every single bucket the sampling probabilities are roughly balanced.

There is also intense work on sampling and counting copies of \emph{a specific pattern} $H$ in sublinear time, including edges, triangles, cliques, and other patterns~\cite{Eden17triangles,Assadi18-counting,Eden20cliques,Eden21-multiedge,Eden21-decomp,Eden18_edges}. However, ``sublinear'' there is meant in the maximum possible number of copies of $H$, which can be as large as $\Theta\big(m^{\frac{k}{2}}\big)$. It is also unclear how those techniques can be applied to uniform graphlet sampling.

\section{Preliminaries and notation}
\label{sub:prelim}
Given $G=(V,E)$, we assume $V=\{1,\ldots,n\}$. We denote the degree of $v \in V$ by $d_v$. We assume the graph access model of~\cite{Kaufman04-model}, where these queries take constant time:
\begin{itemize}\itemsep0pt
    \item neighbor query: given $v \in V$ and $i \in \mathbb{N}$, return the $i$-th neighbor of $v$ in $G$, or $-1$ if $d_v < i$
    \item pair query: given $u,v \in V$, tell if $\{u,v\} \in E$
    \item degree query: given $v \in V$, return $d_v$
\end{itemize}
For any $U \subseteq V$ and $U' \subseteq V \setminus U$, the cut between $U$ and $U'$ is $\Cut(U, U') = E \cap (U \times U')$. The line graph $L(G)=(V',E')$ of a graph $G=(V,E)$ is defined by $V'={v_e : e \in E}$, and $\{v_e,v_{e'}\} \in E'$ if and only if $|e \cap e'|=1$. For $u,v \in V(G)$, we write $u \sim v$ for $\{u,v \} \in E(G)$.

A $k$-graphlet $\gph=(V(\gph),E(\gph))$ is a $k$-vertex subgraph of $G$ that is connected and induced. With a slight abuse of notation, we may use $g$ in place of $V(\gph)$, and $g \cap g'$ in place of $G[V(g) \cap V(g')]$. We denote by $\scV_k$ the set of all $k$-graphlets of $G$. 
%Unless otherwise stated, by ``graphlet'' we mean ``graphlet copy'' (that is, a subgraph of $G$).
The $k$-graphlet graph of $G$ is $\gk_k=(\vk_k,\ek_k)$, where $\{g,g'\}\in \scE_k$ if and only if $g \cap g' \in \vk_{k-1}$. We note that some works define $g$ and $g'$ to be adjacent if $|V(g) \cap V(g')|=k-1$, but our proofs do not work in that case (and so the mixing time of those walks may not respect our bounds).

In this paper, ``$X$ holds with high probability for $Y=\Theta(Z)$'' means that for any fixed $a > 0$ we can make  $\Pr X > 1-n^{-a}$ by choosing $Y \in \Theta(Z)$ sufficiently large. Similarly, ``X has probability $\poly(x)$'' means that for any fixed $a > 0$ we can make $\Pr X < x^{a}$ by adjusting the constants in our algorithms.

\section{Near-optimal mixing time bounds for the k-graphlet walk}
\label{sec:walks}
In this section we prove the results of Section~\ref{sub:walk}. Towards this end, we need to recall some additional preliminary results on Markov Chains, random walks, and mixing.

\subsection{Preliminaries}\label{sub:tgk_prelim}
We denote by $X=\{X_t\}_{t \ge 0}$ a generic Markov chain over a finite state space $\scV$. We denote by $P$ the transition matrix of the chain, and $\pi_t$ the distribution of $X_t$. We always assume that the chain is ergodic, and denote by $\pi=\lim_{t \rightarrow \infty} \pi_t$ its unique limit distribution. We also let $\pi^* = \min_{x \in \scV} \pi(x)$ be the smallest stationary probability of any state. The $\epsilon$-mixing time of $X$ is $\tmix_{\epsilon}(X) = \min\{t_0 : \forall X_0 \in \scX : \forall t \geq t_0 : \tvd{\pi_{t}}{\pi} \le \epsilon \}$. When we write $\tmix(X)$, we mean $\tmix_{\frac{1}{4}}(X)$. Here $\tvd{\sigma}{\pi} = \max_{A \subseteq \mathcal{\scV}} \{ \sigma(A) - \pi(A) \}$ is the variation distance between the distributions $\sigma$ and $\pi$; if $\tvd{\sigma}{\pi} \le \epsilon$ and $\pi$ is uniform, then we say $\sigma$ is $\epsilon$-uniform. 

A graph with non-negative edge weights is denoted by $\scG=(\scV,\scE,w)$ where $w: \scE \rightarrow \R^+_0$. For every $u \in \scV$ we let $w(u)=\sum_{e \in \scE : u \in e} w(e)$. Any such $\scG$ induces a \emph{lazy random walk} as follows. Let $P_0$ be the matrix given by $P_0(u,v) = \frac{w(u,v)}{w(u)}$. Now let $P=\frac{1}{2}(P_0+I)$ where $I$ is the identity matrix. This can be seen as adding a loop of weight $w(u)$ at each vertex of the graph. Note that $P_0$ and $P$ are both stochastic. The lazy random walk over $\scG$ is Markov chain with state space $\scV$ and transition matrix $P$. By standard Markov chain theory, if $\scG$ is connected then the lazy random walk is ergodic, and converges to the limit distribution $\pi$ given by $\pi(u)= \frac{w(u)}{\sum_{v \in V} w(v)}$. It is well-known that the chain is time-reversible with respect to $\pi$, that is, $\pi(x)P(x,y)=\pi(y)P(y,x)$ for all $x,y \in \scV$; and that every time-reversible chain on a finite state space $\scV$ can be seen as a random walk over a graph $G=(\scV,\scE,w)$ where $w(x,y)=\pi(x)P(x,y)$. Thus, we will often write $\gk$ in place of $X$, in which case $X$ is understood to be the lazy chain over $\gk$. The quantity $Q(x,y)=\pi(x)P(x,y)$ is called transition rate between $x$ and $y$.
%The \emph{relaxation time} of the Markov chain, denoted by $\tau(X)$, is defined as $\tau(X)=\nicefrac{1}{\gamma}$, where $\gamma$ is the spectral gap of the transition matrix of the chain. 
%In our results about the random walks, we use several basic relationships between $\tmix_{\epsilon}(X)$ and $\tau(X)$. In particular we use $(\tau(X) - 1) \lg(\nicefrac{1}{2\epsilon}) \le t_{\epsilon}(X) \le \tau(X) \lg(\nicefrac{1}{\epsilon \pi^*})$, where $\pi^* = \min_{u \in \scV} \pi(u)$.
%More background on Markov chains is given in Appendix~\ref{apx:markov}.

The volume of $U \subseteq \scV$ is $\vol(U) = \sum_{u \in U} w(u)$. The cut of $U \subseteq \scV$ is $\Cut(U) = \{e=\{u,u'\} \in \scE : u \in U, u' \in \scV \setminus U\}$, and its weight is $c(U) = \sum_{e \in \Cut(U)} w(e)$. The conductance of $U \subseteq \scV$ is $\cond(U) = c(U)/\vol(U)$. The conductance of $\scG$ is  $\cond(\scG) = \min\{\cond(U) \,:\, U \subset \scV, \vol(U) \le \frac{1}{2}\vol(\scV)\}$.

\subsubsection{Spectral gaps and relaxation times}
\begin{definition}
\label{def:spectral_gap}
Let $P$ be the transition matrix of $X$, and let $\lambda_{*} = \max\big\{ |\lambda| \,:\, \lambda \text{ is an eigenvalue of } P, \, \lambda \ne 1 \big\}$.
%\begin{align}
%\lambda_{*} = \max\Big\{ |\lambda| \,:\, \lambda \text{ is an eigenvalue of } P, \, \lambda \ne 1 \Big\}
%\end{align}
The spectral gap of $X$ is $\gamma = 1-\lambda_*$.
The relaxation time of $X$ is $\tau(X) = \frac{1}{\gamma}$.
\end{definition}
Classic mixing time theory (see e.g.\cite{Levin&2009}) gives the following relationships:
\begin{align}
&\frac{1}{4\Phi} \le \tau(X),\tmix_{\epsilon}(X) \le \frac{2}{\Phi^2} \lg\frac{1}{\epsilon \pi^*} \label{eq:tau_fact_1}
%\\
%\label{eqn:gamma_t}
%&\frac{1}{2\Phi} \le \tau(X) \le \frac{2}{\Phi^2}
\\
&(\tau(X) - 1) \lg \frac{1}{2\epsilon}  \le \tmix_{\epsilon}(X) \le \tau(X) \lg\frac{1}{\epsilon \pi^*} \label{eq:tau_fact_2}
\end{align}
One can show that the last inequality implies $\tau(X) \le c \, \tmix(X)$ for some (small) constant $c \ge 1$.

\subsubsection{Dirichlet forms}
For any function $f : \scV \rightarrow \R$ let $\var_{\pi} f = \E_{\pi}(f - \E_{\pi}{f})^2$.
\begin{definition}[Dirichlet form; see \cite{Levin&2009}, \S 13.2.1]
\label{def:dirichlet}
%Let $P$ be a reversible transition matrix on a state space $\scV$ with stationary distribution $\pi$.
Let $f : \scV \to \mathbb{R}$ be any function.
Then the Dirichlet form associated to $P,\pi,f$ is:
\begin{align}
\label{eqn:dirichlet}
    \mathcal{E}_{P,\pi}(f) = \frac{1}{2} \sum_{x,y \in \scV} \big(f(x)-f(y)\big)^2 Q(x,y)
\end{align}
\end{definition}
\noindent The Dirichlet form characterises the spectral gap as follows:
\begin{lemma}[see~\cite{Levin&2009}, Lemma 13.12]
\label{lem:gamma_extremal}
The spectral gap satisfies:
\begin{align}
	\gamma = \min_{\substack{f \in \R^{V} \\ \var_{\pi}(f) \ne 0}} \frac{\mathcal{E}_{P,\pi}(f)}{\var_{\pi}(f)}
	\label{eqn:gamma_extremal}
\end{align}
\end{lemma}
Next, we recall some results relating the spectral gaps of different chains.
\subsubsection{Direct comparison}
\begin{lemma}[\cite{Levin&2009}, Lemma 13.18]
\label{lem:dirichlet}
Let P and $\Tilde{P}$ be reversible transition matrices with stationary distributions $\pi$ and $\tilde{\pi}$, respectively.
If $\mathcal{E}_{\Tilde{P},\Tilde{\pi}}(f) \le \alpha\, \mathcal{E}_{P,\pi}(f)$ for all functions $f$, then
\begin{align}
    \tilde{\gamma} \le \left(\max_{x \in \scV} \frac{\pi(x)}{\tilde{\pi}(x)} \right) \alpha \gamma
\end{align}
\end{lemma}
\begin{lemma}[\cite{AldousFill}, Lemma 3.29]
\label{lem:direct}
Consider a graph $\scG=(\scV,\scE)$ possibly with loops. Let $w$ and $w'$ be two weightings of $\scE$ and let $\gamma$ and $\gamma'$ be the spectral gaps of the corresponding random walks.
Then:
\begin{align}
\gamma' \ge \gamma \cdot \frac{\min_{e \in \scE} (w(e) / w'(e))}{\max_{v \in \scV} (w(v)/w'(v))}
\end{align}
\end{lemma}

\subsubsection{Collapsed chains}
(See~\cite{AldousFill}, \S 2.7.3).
\begin{definition}
\label{def:collapsed}
Let $A \subset \scV$ and let $A^C = \scV \setminus A$ (note that $A^C \ne \emptyset$).
The collapsed chain $X^*$ has state space $A \cup \{a\}$ where $a$ is a new state representing $A^C$, and transition matrix given by:
\begin{alignat}{2}
P^*(u,v) &= P(u,v) && \qquad u,v \in A
\\
P^*(u,a) &= \sum_{v \in A^C} P(u,v) && \qquad u \in A
\\
P^*(a,v) &= \frac{1}{\pi(A^C)} \sum_{u \in A^C} \pi(u) P(u,v) && \qquad v \in A
\\
P^*(a,a) &= \frac{1}{\pi(A^C)} \sum_{u \in A^C}\sum_{v \in A^C} \pi(u) P(u,v)
\end{alignat}
\end{definition}
\begin{lemma}[\cite{AldousFill}, Corollary 3.27]
\label{lem:collapsed}
The collapsed chain $X^*$ satisfies $\gamma(X^*) \ge \gamma(X)$.
\end{lemma}

\subsubsection{Induced chains}
%(See also~\cite{Levin&2009}, \S 13.4 and Theorem 13.20).
\begin{definition}[\cite{Levin&2009}, \S 13.4]
\label{def:induced}
%Given a reversible chain $\{X_t\}_{t \ge 0}$ on a state space $\scV$, consider an arbitrary nonempty
Let $\emptyset \ne A \subseteq \scV$ and $\tau^+_A = \min\{t \ge 1 : X_t \in A\}$. The induced chain on $A$ is the chain with state space $A$ and transition probabilities:
\begin{align}
P_A(x,y) = P(X_{\tau_A^+} = y \,|\, X_0=x) \qquad \forall x,y \in A
\end{align}
\end{definition} 
\begin{lemma}[\cite{Levin&2009}, Theorem 13.20]
\label{lem:induced}
%Consider a reversible chain on $\scV$ with stationary distribution $\pi$ and spectral gap $\gamma$.
Let $\emptyset \ne A \subseteq \scV$, and let $\gamma_A$ be the spectral gap for the chain induced on $A$. Then $\gamma_A \ge \gamma$.
\end{lemma}

\subsection{Proof of the upper bound of Theorem~\ref{thm:tgk}}
\label{apx:tgk}
This section proves the upper bound of Theorem~\ref{thm:tgk}. First, Since $\tmix(\gk_k) \le \tau(\gk_k) \ln\frac{4}{\pi^*}$ and $\pi^* \ge \kokm n^k$, we have $\tmix(\gk_k) \le \scO(\tau(\gk_k) \lg n)$. Now consider the following inequality:
\begin{align}
\tau(\gk_k) \le \poly(k)\rho(G) \tau(\gk_{k-1})
\label{eq:tau_gk_poly}
\end{align}
Applying~\eqref{eq:tau_gk_poly} to $\tau(\gk_k), \tau(\gk_{k-1}), \ldots, \tau(\gk_2)$, and since $\gk_1=G$ and $\tau(G) = \scO(\tmix(G))$, we obtain:
\begin{align}
 \tmix(\gk_k) \le \kok \rho(G)^{k-1} \tmix(G) \lg n
\end{align}
which is precisely the upper bound of Theorem~\ref{thm:tgk}.
Thus, we only need to prove~\eqref{eq:tau_gk_poly}. The main obstacle in proving that inequality is in relating the spectral gaps of two very different walks --- one over $G$ and one over $\gk_k$. We overcome this obstacle by proving the following result:
\begin{lemma}
\label{lem:tau_Gk_L}
$\tau(\gk_k) \le \poly(k) \tau(L(\gk_{k-1}))$.
\end{lemma}
\noindent Together with Lemma~\ref{lem:Lg} applied to $\gk_{k-1}$, this result yields precisely~\eqref{eq:tau_gk_poly}. Thus, we shall prove Lemma~\ref{lem:tau_Gk_L} and Lemma~\ref{lem:Lg}, in this order.

\subsubsection{Proof of Lemma~\ref{lem:tau_Gk_L}}
\label{apx:lem:tau_Gk_L}
From $L(\gk_{k-1})$ we will construct a weighted graph $L_N$ such that $\tau(L_N) \le \tau(L(\gk_{k-1}))$, and then we will prove that $\tau(\gk_k) \le \poly(k)\tau(L_N)$. Combining these two inequalities gives the claim.

For any $g \in V(\gk_k)$ let $H(g) = \{ x_{uv} \in V(L(\gk_{k-1})) \,:\, g=u \cup v\}$. Note that $\{H(g)\}_{g \in V(\gk_k)}$ is a partition of $V(L)$ into equivalence classes. Now let $V(\gk_k)=\{g_1,\ldots,g_N\}$, and let $L_0=L(\gk_{k-1})$. For each $i=1,\ldots,N$ we define $L_i$ by taking $L_{i-1}$ and identifying $H(g_i)$. Formally, we let $L_i=(V(L_i), E(L_i), w_i)$, where $V(L_i)=V(L_{i-1}) \setminus H(g_i) \cup \{a_i\}$ with $a_i$ being a new state representing $H(g_i)$, and:
\begin{alignat}{2}
w_i(x,x') &= w_{i-1}(x,y)  && \quad x \ne a_i, x' \ne a_i
\\
w_i(x,a_i) &= \sum_{x' \in a_i} w_{i-1}(x,x')  && \quad x \ne a_i
\\
w_i(a_i,a_i) &= \sum_{x \in a_i} \sum_{x' \in a_i} w_{i-1}(x,x')
\end{alignat}
Now we prove two claims from which the thesis immediately follows.
\begin{claim}
\label{claim:tau_Gk}
$\tau(\gk_k) \le \poly(k) \tau(L_N)$.
\end{claim}
\begin{proof}
We show that the walk on $L_N$ is the lazy walk on $\gk_k$ up to a reweighting of the edges by multiplicative factors in $[1,\poly(k)]$. By Lemma~\ref{lem:direct} this implies the thesis. In particular we show that, if $\gk_k$ is taken in its lazy version (with loops accounting for half of the vertex weight), then (1) $V(L_N)=V(\gk_k)$, (2) $E(L_N)=E(\gk_k)$, (3) $1 \le \frac{w_{N}}{w_{\gk_k}} \le \poly(k)$. We denote the generic state $a_i \in L_N$ simply as $g$, meaning that $a_i$ represents $H(g)$.

\paragraph{(1) $V(L_N)=V(\gk_k)$.}
Let $g \in V(L_N)$. By construction, $g = u \cup v$ for some $\{u,v\} \in E(\gk_{k-1})$. Hence $g$ has $k$ vertices and is connected, so it is a $k$-graphlet, and $g \in V(\gk_k)$.
Conversely, let $g \in V(\gk_k)$ and let $T$ be a spanning tree of $g$ (which must exist since $g$ is connected by definition). Let $a,b$ be two distinct leaves of $T$ and let $g'=g \setminus \{a\}$ and $g''=g \setminus \{b\}$. Then $g',g''$ are connected and have $k-1$ vertices, so they are in $V(\gk_{k-1})$. Moreover $|g' \cap g''|=k-2$, so $\{g',g''\} \in E(\gk_{k-1})$. Thus $\{g',g''\} \in L(\gk_{k-1})$ and consequently $g \in V(L_N)$.
Therefore $V(L_N)=V(\gk_{k-1})$.

\paragraph{(2) $E(L_N)=E(\gk_k)$.}
First, both $L_N$ and the lazy version of $\gk_k$ have a loop at each vertex ($L_N$ inherits from $L_0$ a positive self-transition probability at each vertex). Now let $\{g',g''\} \in E(L_N)$ be a non-loop edge. By construction of $L_N$ we have $g'=u \cup v$ and $g''= u \cup z$, with $\{u,v\},\{u,z\} \in E(\gk_{k-1})$ and $u,v,z\in V(\gk_{k-1})$ distinct. This implies $g' \cap g''=u$ and so $\{g',g''\} \in E(\gk_k)$. It follows that $E(L_N) \subseteq E(\gk_k)$.
Now let $\{g',g''\} \in E(\gk_k)$ be a non-loop edge. Let $u = g' \cap g''$; note that by hypothesis $u$ is connected and $|u|=k-1$, so $u \in V(\gk_{k-1})$. Now let $\{a'\} = g' \setminus  g''$ and let $b'$ be any neighbor of $a$ in $u$. Choose any spanning tree $T'$ of $u$ rooted at $b'$, and let $c' \ne b'$ be any leaf of $T'$ (such a leaf exists since $|g| \ge 3$ and thus $|u| \ge 2$). We define $v = g' \setminus \{c'\}$. Note that by construction (1) $v$ is connected and has size $k-1$, (2) $u \cap v$ is connected and has size $k-2$, and (3) $u \cup v = g'$. Therefore $v \in V(\gk_{k-1})$ and $\{u,v\} \in E(\gk_{k-1})$. A symmetric construction using $g''$ and $u$ yields $z$ such that $z \in V(\gk_{k-1})$ and $\{u,z\} \in E(\gk_{k-1})$ and $u \cup z = g''$. Now, by construction, $\{u,z\}$ and $\{u,v\}$ give two adjacent states $x_{uv},x_{uz} \in V(L_0)$. But $u \cup v = g'$ and $u \cup z = g''$, so $x_{uv} \in H(g)$ and $x_{uz} \in H(g'')$. This implies that $\{g',g''\} \in E(L_N)$.
So $E(\gk_k) \subseteq E(L_N)$ and we conclude that $E(\gk_k) = E(L_N)$.

\paragraph{(3) $1 \le \frac{w_{N}}{w_{\gk_k}} \le \poly(k)$.}
First, let us consider non-loop edges. Let $\{a,a'\} \in E(L_N)$ with $a \ne a'$, and let $g,g'$ be the corresponding elements of $\gk_{k}$; note that $g \ne g'$. 
Observe that $w_N(a,a')=|\Cut(H(g),H(g'))|$, where the cut is taken in $L_0=L(\gk_{k-1})$.
Clearly $|\Cut(H(g),H(g'))| \ge 1$ and $w_{\gk_k}(g,g')=1$, therefore $1 \le \frac{w_{N}(a,a')}{w_{\gk_k}(g,g')}$.
For the other side, note that there are at most ${k \choose 2}$ distinct pairs of $(k-1)$-graphlets $u,v \in \gk_{k-1}$ such that $u \cup v = g$.
Thus, $H(g) \le {k \choose 2}$.
The same holds for $g'$. Therefore, $|\Cut(H(g),H(g'))| \le {k \choose 2}^2$. It follows that $\frac{w_{N}(a,a')}{w_{\gk_k}(g,g')} \le {k \choose 2}^2$.

A similar argument holds for the loops.
First, recall that $w_{\gk_k}(g) = d_g$ by the lazy weighting.
Consider then any non-loop edge $\{g, g'\} \in \gk_k$.
Note that $\{g, g'\}$ determines $u = g \cap g' \in V(\gk_{k-1})$ univocally.
Moreover, there exist some $v,z \in \gk_{k-1}$ such that $u\cup v = g$ and $u \cup z = g'$ and that $\{x_{uv},x_{uz}\}$ is an edge in $L_0$; and note that there are at most $k$ distinct $v$ and at most $k$ distinct $z$ satisfying these properties.
Therefore, every $\{g, g'\}$ can be mapped to a set of between $1$ to $k^2$ edges in $L_0$, such that every edge in the set is in the cut between $H(g)$ and $H(g')$.
Furthermore, note that different $g'$ are mapped to disjoint sets, since any edge $\{x_{uv},x_{uz}\}$ identifies univocally $g=u \cup v$ and $g' = u\cup z$.
It follows that the cut of $H(g)$ is at least $d_g$ and at most $k^2 d_g$.
Since the cut has at least one edge, and $H(g)$ has at most ${k \choose 2}^2$ internal edges, then the total weight of $H(g)$ is between $1$ and $\poly(k)$ times the cut.
This is also $w_N(a)$, the weight of the state $a$ representing $g$ in $L_N$.
The claim follows by noting that by construction $w_N(a) \le w_N(a,a) \le 2 w_N(a)$.
\end{proof}
\begin{claim}\label{claim:tau_LN}
$\tau(L_N) \le \tau(L(\gk_{k-1}))$.
\end{claim}
\begin{proof}
The walk on $L_i$ is the walk $L_{i-1}$ collapsed respect to $A^C=H(g_i)$, see Definition~\ref{def:collapsed}. Therefore by Lemma~\ref{lem:collapsed} the spectral gaps of the two walks satisfy $\gamma(L_i) \ge \gamma(L_{i-1})$, and the relaxation times satisfy $\tau(L_i)\le \tau(L_{i-1})$. Thus $\tau(L_N) \le \tau(L_0) = \tau(L(\gk_{k-1}))$. 
\end{proof}

By combining Claim~\ref{claim:tau_Gk} and Claim~\ref{claim:tau_LN}, we obtain $\tau(\gk_k) \le \poly(k) \tau(L_N) \le \poly(k) \tau(L(\gk_{k-1}))$, proving Lemma~\ref{apx:lem:tau_Gk_L}.

\subsubsection{Proof of Lemma~\ref{lem:Lg}}
\label{apx:lem:Lg}
To avoid notational ambiguity, we restate Lemma~\ref{lem:Lg} with $\gk$ in place of $G$:
\begin{replemma}{lem:Lg}
Any graph $\gk$ satisfies $\tau(L(\gk)) \le 20\, \rho(\gk)\, \tau(\gk)$, where $L(\gk)$ is the line graph of $\gk$ and $\tau(\cdot)$ denotes the relaxation time of the lazy random walk.
\end{replemma}
We build an auxiliary weighted graph $\subGw$, as follows. Let $\mathcal{S}$ be the $1$-subdivision of $\gk$ (the graph obtained by replacing each $\{u,v\} \in E(\gk)$ with the path $\{u, x_{uv}\},\{x_{uv},v\}$ where $x_{uv}$ is a new vertex representing $\{u,v\}$).
We make $\subG$ lazy by adding loops and assigning the following weights:
\begin{alignat}{2}
\label{eqn:cxy}
w_{\subG}(u,u) &= d_u && \qquad u \in V(\gk) \\
w_{\subG}(u,x_{uv}) &= 1 && \qquad \{u,v\} \in E(\gk) \\
w_{\subG}(x_{uv},x_{uv}) &= 2 && \qquad \{u,v\} \in E(\gk)
\end{alignat}
The graph $\subGw$ is the same as $\subG$ but with the following weights:
\begin{alignat}{2}
\label{eqn:cxy2}
w_{\subGw}(u,u) &= d_u^2 && \qquad u \in V(\gk) \\
w_{\subGw}(u,x_{uv}) &= d_u && \qquad \{u,v\} \in E(\gk) \\
w_{\subGw}(x_{uv},x_{uv}) &= d_u+d_v && \qquad \{u,v\} \in E(\gk)
\end{alignat}
The reader may refer to Figure~\ref{fig:walks} below.

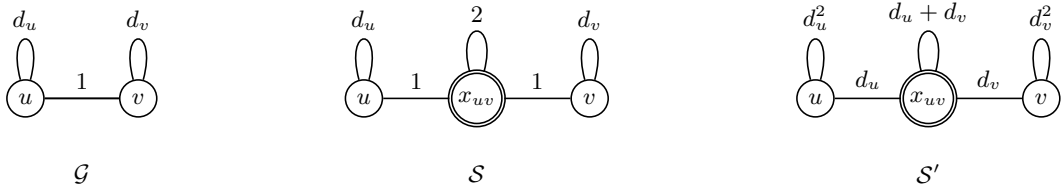
\begin{figure}[h]
\centering
\begin{tikzpicture}[
scale=1.5,semithick
]
\small
\tikzstyle{trans}=[rectangle,solid,draw=black,minimum size=14pt,inner sep=2pt]
\tikzstyle{vertex}=[circle,solid,draw=black,minimum size=14pt,inner sep=2pt]
\tikzset{every loop/.style={}}
\node[vertex] (u) at (0,0) {$u$};
\node[vertex] (v) at ($(u)+(1,0)$) {$v$};
\draw (u) edge [loop above] node[draw=none] {$d_u$} (u);
\draw (v) edge [loop above] node[draw=none] {$d_v$} (v);
\draw (u) -- (v);
\draw (u) edge node [above] {$1$} (v);
\node (txt) at ($0.5*(u)+0.5*(v)$) {};
\draw node [below of=txt] {$\gk$};

\node[vertex] (u1) at ($(v)+(2,0)$) {$u$};
\node[vertex,double] (uv) at ($(u1)+(1,0)$) {$x_{uv}$};
\node[vertex] (v1) at ($(uv)+(1,0)$) {$v$};
\draw (u1) edge [loop above] node[draw=none] {$d_u$} (u1);
\draw (v1) edge [loop above] node[draw=none] {$d_v$} (v1);
\draw (uv) edge [loop above] node[draw=none] {$2$} (uv);
\draw (u1) edge node [above] {$1$} (uv);
\draw (uv) edge node [above] {$1$} (v1);
\draw node [below of=uv] (txt1) {$\subG$};

\node[vertex] (u2) at ($(v1)+(2,0)$) {$u$};
\node[vertex,double] (uv2) at ($(u2)+(1,0)$) {$x_{uv}$};
\node[vertex] (v2) at ($(uv2)+(1,0)$) {$v$};
\draw (u2) edge [loop above] node[draw=none] {$d_u^2$} (u2);
\draw (v2) edge [loop above] node[draw=none] {$d_v^2$} (v2);
\draw (uv2) edge [loop above] node[draw=none] {$d_u+d_v$} (uv2);
\draw (u2) edge node [above] {$d_u$} (uv2);
\draw (uv2) edge node [above] {$d_v$} (v2);
\draw node [below of=uv2] (txt2) {$\subGw$};
\end{tikzpicture}
\caption{Left: a pair of $(k-1)$-graphlets $u,v$ forming an edge in $\gk_{k-1}$.
Middle: how $\{u,v\}$ appears in $\subG$, the $1$-subdivision of $G$.
Right: the reweighting given by $\subGw$.
}
\label{fig:walks}
\end{figure}

Now we prove two claims which, combined, yield the thesis.
%In the rest of the subsection we prove Lemma~\ref{lem:tmix_G_Gsub} and Lemma~\ref{lem:L_subG}.
\begin{claim}
\label{claim:tmix_G_Gsub}
$\tau(\subGw) \le 4 \rho(\gk) \tau(\gk)$.
\end{claim}
\begin{proof}%[Proof of Lemma~\ref{lem:tmix_G_Gsub}]
Let $\dmax,\dmin$ be the maximum and minimum degrees of $\gk$.
First, note that
\begin{align}
\min_{\{x,y\} \in E(\subG)} \frac{w_{\subG}(x,y)}{w_{\subGw}(x,y)} \ge \frac{1}{\dmax} \quad \text{and} \quad \max_{x \in V(\subG)} \frac{w_{\subG}(x)}{w_{\subGw}(x)} \le \frac{1}{\delta}
\end{align}
By Lemma~\ref{lem:direct} this implies that $\gamma(\subGw) \ge \rho(\gk)^{-1}\,\gamma(\subG)$, or equivalently $\tau(\subGw) \le  \rho(\gk)\, \tau(\subG)$.
Thus, we need only to show that $\tau(\subG) \le 4 \tau(\gk)$, or equivalently, $\gamma(\gk) \le 4\gamma(\subG)$.
We do so by comparing the numerators and denominators of~\eqref{eqn:gamma_extremal} in Lemma~\ref{lem:dirichlet} for $\subG$ and $\gk$.

%To simplify the notation, we write $uv$ for $x_{uv}$.
Consider the walk on $\subG$ and let $\pi_{\subG}$ be its stationary distribution.
Let $f_{\subG}$ be the choice of $f$ that attains the minimum in~\eqref{eqn:gamma_extremal} under $\pi=\pi_{\subG}$.
We will show that there exists $f_{\gk} \in \R^{V(\gk)}$ such that:
\begin{align}
\frac{\mathcal{E}_{P_{\gk},\pi_{\gk}}(f_{\gk})}{\var_{\pi_{\gk}}(f_{\gk})} \le 4\, \frac{\mathcal{E}_{P_{\subG},\pi_{\subG}}(f_{\subG})}{\var_{\pi_{\subG}}(f_{\subG})}
\label{eqn:dirich_comp}
\end{align}
By Lemma~\ref{lem:dirichlet} this implies our claim, since the left-hand side of~\eqref{eqn:dirich_comp} bounds $\gamma(\gk)$ from above and the right-hand side equals $4\gamma(\subG)$.
%Before continuing we make a few observations that one can readily check.
Now, first, note that $\pi_{\subG}(u) = \frac{2}{3}\pi_{\gk}(u)$ for all $u \in V(\gk)$ (the weight of $u$ is the same in $\gk$ and $\subG$, but the total sum of weights in $\subG$ is $\frac{3}{2}$ that of $\gk$).
Similar calculations show that for all $\{u,v\} \in E(\gk)$ we have $\pi_{\subG}(x_{uv}) = \frac{4}{3 d_u} \pi_{\gk}(u)$, where $d_u$ is the degree of $u$ in $\gk$.
Third, observe that since $f_{\subG}$ attains the minimum in~\eqref{eqn:gamma_extremal} then  $f_{\subG}(x_{uv}) = \frac{f_{\subG}(u)+f_{\subG}(v)}{2}$ for all $\{u,v\} \in E(\gk)$.
Finally, let $f_{\gk}$ be the restriction of $f_{\subG}$ to $V(\gk)$.

First, we compare the numerator of~\eqref{eqn:dirich_comp} for $\subG$ and for $\gk$. To begin, note that:
\begin{align}
\mathcal{E}_{P_{\subG},\pi_{\subG}}(f_{\subG})
&= \sum_{\{u,v\} \in E(\gk)} \!\!\!\!\left( (f_{\subG}(u)-f_{\subG}(x_{uv}))^2 \, Q_{\subG}(u,x_{uv}) + (f_{\subG}(v)-f_{\subG}(x_{uv}))^2 \, Q_{\subG}(u,x_{uv}) \right)
\end{align}
Observe that $Q_{\subG}(u,x_{uv}) = Q_{\subG}(v,x_{uv}) = \pi_{\subG}(u) \frac{1}{2d_u}$, and as noted above, $f_{\subG}(x_{uv}) = \frac{f_{\subG}(u)+f_{\subG}(v)}{2}$, thus $(f_{\subG}(u)-f_{\subG}(x_{uv}))=(f_{\subG}(v)-f_{\subG}(x_{uv}))=\frac{1}{2}(f_{\subG}(u)-f_{\subG}(v))$. Recalling that $\pi_{\subG}(u) = \frac{2}{3}\pi_{\gk}(u)$,
\begin{align}
\mathcal{E}_{P_{\subG},\pi_{\subG}}(f_{\subG}) &= \frac{1}{2} \sum_{\{u,v\} \in E(\gk)}  \left(f_{\subG}(u)-f_{\subG}(v)\right)^2 \, \pi_{\subG}(u) \frac{1}{2d_u} 
\\ &= \frac{1}{3} \sum_{\{u,v\} \in E(\gk)}  \left(f_{\subG}(u)-f_{\subG}(v)\right)^2 \, \pi_{\gk}(u) \frac{1}{2d_u}
\label{eqn:dirich_expr_1}
\end{align}
On the other hand, since by construction $f_{\gk}(u)=f_{\subG}(u)$ and since $Q_{\gk}(u,v) = \pi_{\gk}(u)\frac{1}{2d_u}$:
\begin{align}
\mathcal{E}_{P_{\gk},\pi_{\gk}}(f_{\gk})
&= \sum_{\{u,v\} \in E(\gk)} \left(f_{\gk}(u)-f_{\gk}(v)\right)^2 \, Q_{\gk}(u,v)
\\ &= \sum_{\{u,v\} \in E(\gk)} \left(f_{\subG}(u)-f_{\subG}(v)\right)^2 \, \pi_{\gk}(u)\frac{1}{2d_u}
\label{eqn:dirich_expr_2}
\end{align}
Comparing~\eqref{eqn:dirich_expr_1} and~\eqref{eqn:dirich_expr_2} shows that $\mathcal{E}_{P_{\gk},\pi_{\gk}}(f_{\gk})=3\,\mathcal{E}_{P_{\subG},\pi_{\subG}}(f_{\subG})$.

Next, we compare the denominator of~\eqref{eqn:dirich_comp} for $\subG$ and for $\gk$.
First, we have:
\begin{align}
\var_{\pi_{\subG}}(f_{\subG}) = \sum_{u \in V(\gk)} \pi_{\subG}(u)f_{\subG}(u)^2 + \sum_{\{u,v\} \in E(\gk)} \pi_{\subG}(x_{uv})f_{\subG}(x_{uv})^2
\end{align}
Since $\pi_{\subG}(u) = \frac{2}{3}\pi_{\gk}(u)$ and $f_{\subG}(u)=f_{\gk}(u)$, the first term equals $\frac{2}{3}\var_{\pi_{\gk}}(f_{\gk})$. Now we show that the second term is bounded by $\frac{2}{3}\var_{\pi_{\gk}}(f_{\gk})$. Recalling again that $f_{\subG}(x_{uv}) = \frac{f_{\subG}(u)+f_{\subG}(v)}{2}$:
\begin{align}
\sum_{\{u,v\} \in E(\gk)} \pi_{\subG}(x_{uv})f_{\subG}(x_{uv})^2
&= \sum_{\{u,v\} \in E(\gk)} \pi_{\subG}(x_{uv}) \left(\frac{f_{\subG}(u)+f_{\subG}(v)}{2}\right)^2
\\&\le \sum_{\{u,v\} \in E(\gk)} \pi_{\subG}(x_{uv}) \frac{1}{2}\left(f_{\subG}(u)^2+f_{\subG}(v)^2\right) \label{eq:explain2}
\\& = \sum_{u \in V(\gk)} \sum_{v : \{u,v\} \in E(\gk)} \frac{1}{2} \frac{4}{3 d_u} \pi_{\gk}(u) f_{\subG}(u)^2 \label{eq:explain1}
\\& = \frac{2}{3} \sum_{u \in V(\gk)} \pi_{\gk}(u) f_{\subG}(u)^2
\end{align}
where~\eqref{eq:explain2} holds by convexity, and~\eqref{eq:explain1} holds since every $u \in V(\gk)$ is charged with $\frac{1}{2}\pi_{\subG}(x_{uv})f_{\subG}(u)^2$ by every $\{u,v\} \in E(\gk)$, and since $\pi_{\subG}(x_{uv}) = \frac{4}{3 d_u} \pi_{\gk}(u)$. Therefore $\var_{\pi_{\subG}}(f_{\subG}) \le \frac{4}{3} \var_{\pi_{\gk}}(f_{\gk})$, so $\var_{\pi_{\gk}}(f_{\gk}) \ge \frac{3}{4}\var_{\pi_{\subG}}(f_{\subG})$.

By combining our two bounds, we obtain:
\begin{align}
\frac{\mathcal{E}_{P_{\gk},\pi_{\gk}}(f_{\gk})}{\var_{\pi_{\gk}}(f_{\gk})} \le
\frac{3\,\mathcal{E}_{P_{\subG},\pi_{\subG}}(f_{\subG})}{\frac{3}{4}\,\var_{\pi_{\subG}}(f_{\subG})}
=
4\, \frac{\mathcal{E}_{P_{\subG},\pi_{\subG}}(f_{\subG})}{\var_{\pi_{\subG}}(f_{\subG})}
\end{align}
which shows that $\gamma(\gk) \le 4\gamma(\subG)$, completing the proof.
\end{proof}

\begin{claim}
\label{claim:L_subG}
$\tau(L(\gk)) \le 5 \,\tau(\subGw)$.
\end{claim}
\begin{proof}%[Proof of Lemma~\ref{lem:L_subG}]
Let $X=\{X_t\}_{t \ge 0}$ be the walk on $\subGw$, and let $Y=X[A]$ be the chain induced by $X$ on the subset of states $A = \{x_{uv} \,:\, \{u,v\} \in \ek \}$ (Definition~\ref{def:induced}). Since by Lemma~\ref{lem:induced} $\tau(Y) \le \tau(\subGw)$, we need only to prove that $\tau(L(\gk)) \le 5 \,\tau(Y)$. To this end we show that $Y$ is the random walk on the graph $L'$ obtained by weighting $L(\gk)$ as follows (see Figure~\ref{fig:L_reweighted} below):
\begin{alignat}{2}
\label{eqn:cL}
w_{L'}(x_{uv},x_{uz}) &= 1 && \qquad v \ne z
\\
w_{L'}(x_{uv},x_{uv}) &= d_u+d_{v}+2 && \qquad %x = \{h,h'\} \in V(L)
\end{alignat}

\begin{figure}[h]
\centering
%\fbox{
\begin{tikzpicture}[
scale=1.5,semithick
]
\small
\tikzstyle{trans}=[rectangle,solid,draw=black,minimum size=14pt,inner sep=2pt]
\tikzstyle{vertex}=[circle,solid,draw=black,minimum size=14pt,inner sep=2pt]
\tikzset{every loop/.style={}}
\node[vertex,double] (uz) at ($(u)+(-1,-.5)$) {$x_{uz}$};
\node[vertex] (u2) at (0,0) {$u$};
\node[vertex,double] (uv2) at ($(u2)+(1,0)$) {$x_{uv}$};
\node[vertex] (v2) at ($(uv2)+(1,0)$) {$v$};
\draw (u) edge node [above] {$d_u$} (uz);
\draw (u2) edge [loop above] node[draw=none] {$d_u^2$} (u2);
\draw (v2) edge [loop above] node[draw=none] {$d_v^2$} (v2);
\draw (uv2) edge [loop above] node[draw=none] {$d_u+d_v$} (uv2);
\draw (uz) edge [loop above] node[draw=none] {$d_u+d_z$} (uz);
\draw (u2) edge node [above] {$d_u$} (uv2);
\draw (uv2) edge node [above] {$d_v$} (v2);
\draw (v2) edge node [above] {$d_v$} ($(v2)+0.7*(1,-.5)$);
\draw node [below of=uv2] (txt2) {$\subGw$};

\node (u2) at ($(5,0)$) {};
\node[vertex,double] (uv) at ($(u2)$) {$x_{uv}$};
\node[vertex,double] (uz) at ($(u2)+(2,0)$) {$x_{uz}$};
\draw (uv) edge [loop above] node[draw=none] {$d_u+d_v+2$} (uv);
\draw (uz) edge [loop above] node[draw=none] {$d_u+d_z+2$} (uz);
\draw (uz) edge node [above] {$1$} (uv);
\node (txt2) at ($0.5*(uv)+0.5*(uz)$) {};
\draw node [below of=txt2] {$L'$};
\end{tikzpicture}
%}
\caption{Left: the graph $\subGw$ described above.
Right: the reweighted line graph $L'$ obtained by weighting every loop $\{x_{uv},x_{uv}\}$ of $L(\gk)$ with $(d_u+d_v+2)$ instead of $(d_u+d_v-2)$. The random walk over $L'$ is exactly the random walk over $\subGw$ observed only on the set of states $\{ x_{uv} : \{u,v\} \in E(\gk)\}$.
}
\label{fig:L_reweighted}
\end{figure}
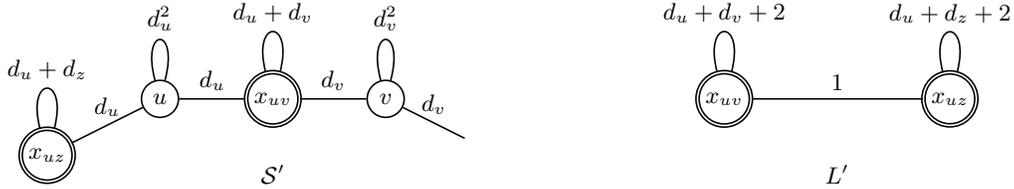

To prove the claim we compute the transition probabilities of $Y$ from $x_{uv}$. First, if $Y_t=x_{uv}$ then we can assume $X_s=x_{uv}$ for some $s=s(t)$. From $x_{uv}$, the possible transitions are $Y_{t+1}=x_{uv}$ and $Y_{t+1}=x_{uz}$ for some $z \ne v$. The transition $Y_{t+1}=x_{uv}$ happens if and only if one of these three disjoint events occurs:
\begin{enumerate}\itemsep0pt
\item $X_{s+1}=x_{uv}$
\item $X_{s+1}=\ldots=X_{s'-1}=u$ and $X_{s'}=x_{uv}$ for some $s' \ge s+2$
\item the same as (2) but with $v$ in place of $u$.
\end{enumerate}
The probability of (1) is $\frac{1}{2}$ by construction of the loop weights. The probability of (2) is the product of $\Pr(X_{s+1}=u \,|\, X_s=x_{uv}) = \frac{d_u}{2(d_u+d_v)}$ and $\Pr(X_{s'}=x_{uv} \,|\, X_{s'-1}=u)=\frac{1}{d_u}$, since $X$ leaves $u$ with probability $1$, in which case it moves to $x_{uv}$ with probability $\frac{1}{d_u}$. Thus, the probability of (2) is $\frac{1}{2(d_u+d_v)}$, and by symmetry the same is for (3). Therefore:
\begin{align}
\Pr(Y_{t+1}=x_{uv} \,|\, Y_t=x_{uv}) = \frac{1}{2} + \frac{1}{d_u+d_v} = \frac{d_u+d_v+2}{2(d_u+d_v)}
\end{align}
The transition $Y_{t+1}=x_{uz}$ is the same as event (2) above, only with $X_{s'}=x_{uz}$ instead of $X_{s'} = x_{uv}$. But conditioned on $X_{s'-1}=u$ the two events have the same probability, therefore:
\begin{align}
\Pr(Y_{t+1}=x_{uz} \,|\, Y_t=x_{uv}) = \frac{1}{2(d_u+d_v)} 
\end{align}
Thus the probabilities are proportional to $1$ and $d_u+d_v+2$, as $w_{L'}$ says.

We can now conclude the proof of the claim. If $d_u+d_v = 2$, then $|E(\gk)|=1$, so $L(\gk)$ is the singleton graph and $\tau(L(\gk))=0$, and $\tau(L(\gk)) \le 5\, \tau(\subGw)$ holds trivially. Suppose instead $d_u+d_v \ge 3$. Then $\frac{d_u+d_v+2}{d_u+d_v-2} \le \frac{3+2}{3-2} = 5$. Therefore $w_S \le w_{L'} \le 5\, w_{\subGw}$, and Lemma~\ref{lem:direct} yields $\tau(L(\gk)) \le 5\, \tau(\subGw)$. The proof is complete.
\end{proof}

By combining claims~\ref{claim:tmix_G_Gsub} and~\ref{claim:L_subG} we obtain $\tau(L(\gk)) \le 5 \,\tau(\subGw) \le 20 \rho(\gk) \tau(\gk)$, proving Lemma~\ref{lem:Lg}.

\subsection{Proof of the lower bounds of Theorem~\ref{thm:tgk}}
\label{apx:tgk_lb}
We ignore factors depending only on $k$, which are easily proven to be in $\kok$. Consider a graph $G$ formed by two disjoint cliques of order $\dmax$, connected by a ``fat path'' (the Cartesian product of a path and a clique) of length $2(k-1)$ and width $\dmin$, see Figure~\ref{fig:G}. The total number of vertices is $n = 2\dmax + 2(k-1)\dmin) = \Theta(\dmax)$, and we choose $\dmax$ and $\dmin$ so that $\rho(G) = \frac{\dmax}{\dmin} \in \Theta(\rho(n))$.
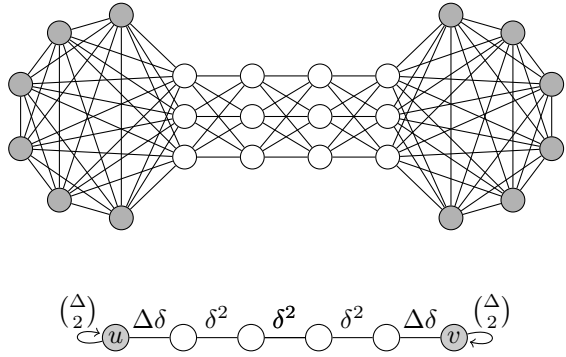
\begin{figure}[h]
\centering
\pgfmathsetmacro\maxdeg{8}
\pgfmathsetmacro\mindeg{3}
\pgfmathsetmacro\k{3}

\def\pwid{\mindeg} % fat path width
\def\xspacing{1}
\def\yspacing{.6}
\def\starSize{6}
\def\starAngle{37}
\pgfmathtruncatemacro\plen{\k-1} % fat path half-length 
\pgfmathsetmacro\starRadius{\xspacing*1.5}
\pgfmathsetmacro\voffx{\xspacing*(1+2*\plen)}
\pgfmathsetmacro\yoff{-\yspacing*(\pwid+1)/2}

\begin{tikzpicture}[scale=.9,-,draw=black, node distance=\layersep]
\tikzstyle{every node}=[circle,draw,fill=white!25,minimum size=9pt,inner sep=0pt]

	\coordinate (u) at (0,0);
	\coordinate (v) at ($(u)+(\voffx,0)$);
	
	% fat path
	\foreach \i in {1,...,\plen} {
		\foreach \z in {1,...,\pwid} {
			\node (L_\i_\z) at ($(u)+(\i*\xspacing, \yoff+\z*\yspacing)$) {};
			\node (R_\i_\z) at ($(v)+(-\i*\xspacing, \yoff+\z*\yspacing)$) {};
		}
	}

	% all the edges
    \pgfmathtruncatemacro{\im}{\plen-1}
	\foreach \i in {1,...,\im} {
		\pgfmathtruncatemacro{\h}{\i + 1}
		\foreach \z in {1,...,\pwid} {
			\foreach \w in {1,...,\pwid} { % the complete bipartite subgraphs
				\path (L_\i_\z) edge (L_\h_\w);
				\path (R_\i_\z) edge (R_\h_\w);
			}
		}
	}

	\foreach \w in {1,...,\pwid}
		\foreach \z in {1,...,\pwid} {
			\path (L_\plen_\w) edge (R_\plen_\z);
		}
    
    % the cliques
    \pgfmathsetmacro\initAngle{180-\starAngle*(\starSize+1)/2}
    \foreach \i in {1,...,\starSize} {
	    	\node[fill=gray!60] (a\i) at ($(u)+(\initAngle+\i*\starAngle:\starRadius*\xspacing)$) {};
	    	\node[fill=gray!60] (b\i) at ($(v)-(\initAngle+\i*\starAngle:\starRadius*\xspacing)$) {};
	    	\foreach \j in {1,...,\pwid} {
	    		\path (a\i) edge (L_1_\j);	
	    		\path (b\i) edge (R_1_\j);	
	    	}
    }
	% edges of the big cliques
    \foreach \i in {2,...,\starSize} {
	    	\pgfmathsetmacro\jmax{\i-1}
	    	\foreach \j in {1,...,\jmax} {
	    		\path (a\i) edge (a\j);
	    		\path (b\i) edge (b\j);
	    	}
    }

\end{tikzpicture}
\\[20pt]
\pgfmathsetmacro\maxdeg{8}
\pgfmathsetmacro\mindeg{3}
\pgfmathsetmacro\k{3}

\def\pwid{\mindeg} % fat path width
\def\xspacing{1}
\def\yspacing{.6}
\def\starSize{6}
\def\starAngle{37}
\pgfmathtruncatemacro\plen{\k-1} % fat path half-length 
\pgfmathsetmacro\starRadius{\xspacing*1.5}
\pgfmathsetmacro\voffx{\xspacing*(1+2*\plen)}
\pgfmathsetmacro\yoff{-\yspacing*(\pwid+1)/2}

\begin{tikzpicture}[scale=.9,-,draw=black, node distance=\layersep]
\tikzstyle{graph}=[circle,draw,fill=white!25,minimum size=10pt,inner sep=0pt]

	\node[graph,fill=gray!40] (u) at (0,0) {$u$};
	\node[graph,fill=gray!40] (v) at ($(u)+(\voffx,0)$) {$v$};
	
	% path
	\foreach \i in {1,...,\plen} {
		\node[graph] (u\i) at ($(u)+(\i*\xspacing, 0)$) {};
		\node[graph] (v\i) at ($(v)+(-\i*\xspacing, 0)$) {};
	}

	% all the edges
    \pgfmathtruncatemacro{\im}{\plen-1}
	\foreach \i in {1,...,\im} {
		\pgfmathtruncatemacro{\h}{\i+1}
		\path (u\i) edge node[above] {$\delta^2$} (u\h);
		\path (v\i) edge node[above] {$\delta^2$} (v\h);
	}

    \pgfmathsetmacro\initAngle{180-\starAngle*(\starSize+1)/2}
    \pgfmathtruncatemacro\l{\starSize/2}
	\node (phantomu) at ($(u)+(-2.1,0)$) {};
	\node (phantomv) at ($(v)+(2,0)$) {};

	\path (u\plen) edge node[above] {$\delta^2$} (v\plen);
	\path (u\plen) edge node[above] {$\delta^2$} (v\plen);

	\path (u) edge[loop left] node[above] {$\binom{\Delta}{2}$} (u);
	\path (v) edge[loop right] node[above] {$\binom{\Delta}{2}$} (v);

	\path (u) edge node[above] {$\Delta\delta$} (u1);
	\path (v) edge node[above] {$\Delta\delta$} (v1);

\end{tikzpicture}
\caption{Above: the graph $G$ for $k=3$, $\dmin=6$, $\dmax=3$; clique vertices in gray. Below: $G$ collapsed in a weighted path $P$.}
\label{fig:G}
\end{figure}

We start by bounding $\tmix(\gk_k)$ from below with a conductance argument. Let $C_u$ be the left clique of $G$, and for $i=1,\ldots,k-1$, let $L_i$ be the vertices of $G$ at distance $i$ from $C_u$. Let $U$ be the set of all $k$-graphlets of $G$ containing at least $\frac{k}{2}+1$ vertices from $C_u \cup L_1 \cup \ldots \cup L_{k-1}$, and $\bar{U} = \vk_k \setminus U$. Consider the cut between $U$ and $\bar U$ in $\gk_k$. Observe that $\vol(U) \le \vol(\bar U)$, which implies $\Phi(\gk_k) \le \frac{\cut(U)}{\vol(U)}$. Now, $U$ contains at least $\binom{\dmax}{k} = \Omega(\dmax^{k})$ graphlets, each of which has $\Omega(\dmax)$ neighbors in $\gk_k$. Hence $\vol(U) = \Omega(\dmax^{k+1})$. On the other hand, consider any $\{g,g'\} \in \Cut(U,\bar U)$. We claim that $g \cup g'$ is spanned by a tree on $k+1$ vertices that does not intersect the cliques of $G$. Indeed, suppose by contradiction that $g \cap C_u \ne \emptyset$. Since $g \cup g'$ has diameter at most $k$, we deduce that $g' \setminus (C_u \cup L_1 \cup \ldots \cup L_k)$ has size at most $1$. This contradicts the fact that $g' \in \bar U$, which would require $|g' \setminus (C_u \cup L_1 \cup \ldots \cup L_k)| \ge \frac{k}{2}$, which is strictly larger than $1$ since $k \ge 3$. A symmetric argument proves that $g \cup g'$ does not intersect the right clique of $G$. Hence, $g \cup g'$ is spanned by a tree on $k+1$ vertices of the fat path, and the number of such trees is $\scO(\dmin^{k+1})$. Therefore, $\cut(U) = \scO(\dmin^{k+1})$. We conclude that:
\begin{align}
\tau(\gk_k) \ge \frac{1}{2 \Phi(\gk_k)} = \Omega\left(\rho(n)^{k+1}\right) \label{eq:tau_gk_proof}
\end{align}

Now we show that $\rho(n)^2 = \Omega(\tau(G))$. Let $P$ be the weighted path graph obtained from $G$ by identifying the vertices in each clique, and the vertices in every layer of the path (see the figure again). Let $X=\{X_t\}_{t \ge 0}$ be the random walk over $G$, and for all $t \ge 0$ let $Y_t$ be the vertex of $V(P)$ corresponding to $X_t$. Note that $Y=\{Y_t\}_{t \ge 0}$ is the random walk over $P$, and that it is coupled to $X$. Now observe that, for any $t \ge 1$, if $Y_t$ is at total variation distance $\epsilon$ from the stationary distribution $\pi_Y$ of $Y$, then $X_t$ is at total variation distance $\epsilon$ from the stationary distribution $\pi_X$ of $X$. Therefore, $\scO(\tmix(G))=\scO(\tmix(P))$, which implies $\tau(G)=\scO(\tmix(P))$. In turn, $P$ is a path of constant length whose edge weights are in the range $[\dmin^2,\dmax^2]$. By Lemma~\ref{lem:direct} this implies that $\tmix(P)$ is within $\scO(\rho(n)^2)$ times the mixing time of the walk on the unweighted version of $P$, which is constant. Therefore, $\tau(G) = \scO(\rho(n)^2)$, i.e., $\rho(n)^2 = \Omega(\tau(G))$.

Combining~\eqref{eq:tau_gk_proof} with the fact that $\rho(n)^2 = \Omega(\tau(G))$ and that $\tmix(\gk_k) = \Omega(\tau(\gk_k))$ yields the lower bound of Theorem~\ref{thm:tgk}.

\subsection{Proof of Theorem~\ref{thm:walk_algo}}
First we prove two ancillary facts, and then, Theorem~\ref{thm:walk_algo}.

\iffalse
\begin{lemma}
\label{lem:edge_mix}
Let $\{X_t\}_{t \ge 0}$ be the lazy random walk over $\gk$ and let $Y_t=\{X_t,X_{t+1}\}$. Let $\pi_t$ be the distribution of $X_t$ over $V(\gk)$ and $\sigma_t$ the distribution of $Y_t$ over $E(\gk)$, and let $\pi$ be the stationary distribution of $X$. If $\pi_t$ is $\epsilon$-close to $\pi$, then $\sigma_t$ is $\epsilon$-uniform.
\end{lemma}
\begin{proof}
Let $X^*$ be a random variable with distribution $\pi$, and let $Y^*=\{X^*,Z\}$ where $Z$ is a neighbor of $X^*$ chosen uniformly at random. Clearly $Y^*$ is uniform over $E(\gk)$. Since $\tvd{\pi_t}{\pi} \le \epsilon$, we can couple $X_t$ and $X^*$ so that $\Pr(X_t \ne X^*) \le \epsilon$. If $X_t=X^*$, we can then couple $Y_t$ and $Y_{\pi}$ by letting $Y_t=Y$. This shows that $\Pr(Y^* \ne Y_t) \le \epsilon$.
\end{proof}
\fi

\begin{lemma}
\label{lem:Tg}
%Consider the graphs $\gk_k=(\vk_k,\ek_k)$ and $\gk_{k-1}=(\vk_{k-1},\ek_{k-1})$.
For every $g \in V(\gk_k)$ let $T(g) = \{ \{u,v\} \in E(\gk_{k-1}) \,:\, u \cup v = g\}$.
Then $|T(g)| \le {k \choose 2}$, and given $g$ we can compute $|T(g)|$ in time $\scO(\poly(k))$.
\end{lemma}
\begin{proof}
Every $\{u,v\} \in T(g)$ satisfies: (i) $u = g \setminus \{x\}$ and $v = g \setminus \{y\}$ for some $x,y \in g$, and (ii) $u$, $v$, and $u \cap v$ are connected. Thus given $g$ we can just enumerate all ${k \choose 2}$ pairs of vertices in $g$ and count which ones have $u$, $v$, and $u \cap v$ connected. % in which case $\{u,v\} \in \ek_{k-1}$.
This gives the bound on $|T(g)|$ too.
\end{proof}

\begin{lemma}
\label{lem:step_cost}
Any single step of the lazy walk over $\gk_k$ can be simulated in $\poly(k)$ expected time.
\end{lemma}
\begin{proof}
To decide whether to follow the loop we just toss a fair coin. Let us now see how to transition to a neighbouring graphlet uniformly at random. Let $g \in V(\gk_k)$ be the current vertex of the walk and let $N(g)$ be the set of neighbors of $g$ in $\gk_k$.
%Note that $N(g)$ can be partitioned as
%\begin{align}
%    N(g)=\dot\bigcup_{x \in V(g)} N_x(g)
%\end{align}
%where $N_x(g) = \{g' \sim g : x \notin V(g')\}$.
%For each $x \in V(g)$, if $g \setminus x$, then $N_x(g)$ is empty by definition. Otherwise, 
For every $y \in V(g)$, consider the following cut in $G$:
\begin{align}
    C(g,y) = \Cut(y, V(G) \setminus V(g))
\end{align}
Clearly, for every edge $\{y,y'\} \in C(g,y)$, the graphlet $g \cup y' \setminus x$ is adjacent to $g$ in $\gk_k$, provided that $g \setminus x$ is connected. Moreover, $|C(g,y)|$ can be computed in time $\scO(k)$, as the difference between $d_y$ and the number of neighbors of $y$ in $G$.
%\begin{align}
    %C(g,x) = \Cut(V(g) \setminus x, V(G) \setminus V(g))
%\end{align}
%It is easy to see that, for every $\{u,y\} \in C(g,x)$, the graphlet $g \cup y \setminus x$ is adjacent to $g$ in $\gk_k$. Moreover, $c(g,x)=|C(g,x)|$ can be computed in time $\scO(k)$.
%Finally, we let $\cut(g) = \sum_{x \in g} \cut(g \setminus x)$.

Now, for every $x \in V(g)$ let $c(x) = \sum_{y \in V(g) \setminus x} |C(g,y)|$ if $g \setminus x$ is connected, and $c(x) = 0$ otherwise. Finally, let $c = \sum_{x \in V(g)} c(x)$. We draw $g'$ at random as follows. First, we draw $x \in V(g)$ at random with probability $c(x)/c$. Then, we draw $y \in V(g) \setminus x$ at random with probability $|C(g,y)|/c(x)$. Finally, we select an edge $\{y,y'\}$ uniformly at random in $C(g,y)$. To do this we just sample $y'$ uniformly at random from the neighbors of $y$ in $G$ until hitting on $V(G) \setminus V(g)$. This requires at most $k$ trials in expectation, since $y$ has at most $k-1$ neighbors in $V(g)$, and has at least one neighbor in $V(G) \setminus V(g)$, otherwise $|C(g,y)|/c(x)=0$ and we wouldn't have chosen $y$.
%\begin{align}
    %\frac{c(g, x, y)}{c(g)}
%\end{align}
%Then, we select $\{u,y\} \in C(g, x, y)$ uniformly at random. To do this we just sample $y'$ uniformly at random from the neighbors of $y$ until $y' \notin V(g)$. This requires at most $k$ trials in expectation, since $y$ has at most $k-1$ neighbors in $V(g)$ and at least one neighbor in $V(G) \setminus V(g)$.

Now consider any $g' \sim g$. Note that $g'$ is identified by the pair $(x,y')$ where $\{x\} = V(g) \setminus V(g')$ and $\{y'\} = V(g') \setminus V(g)$. The probability that the random process above generates $g'$ is:
%if and only if it draws a pair $(x,y)$ and then draws $y'$ from the neighbors of $y$ in $V(G) \setminus V(g)$. The probability of this event is:
\begin{align}
\frac{c(x)}{c} \cdot \sum_{\substack{y \in V(g) \setminus x \\ y' \sim y}} \frac{|C(g,y)|}{c(x)} \cdot \frac{1}{|C(g,y)|}
=
\frac{1}{c} \left|\{y \in V(g) \setminus x : y' \sim y\} \right|,
\end{align}
that is, equal for all $g'$ up to a multiplicative factor between $1$ and $k$. However, once we have drawn $(x,y')$ we can compute $|\{y \in V(g) \setminus x : y' \sim y\} |$ in time $\poly(k)$ and apply rejection sampling to make the output distribution uniform. The expected number of rejection trials is in $\poly(k)$ as well, and so is the expected running time of the entire process.
\end{proof}

We can now prove Theorem~\ref{thm:walk_algo}. Consider $\gk_{k-1}$. By construction, $\{u,v\} \in E(\gk_{k-1})$ if and only if $g = u \cup v \in V(\gk_{k-1})$. Recall from Lemma~\ref{lem:Tg} the set $T(g)$, and that $|T(g)| \le k^2$. Hence, if we draw from a $\scO\big(\frac{\epsilon}{k^2}\big)$-uniform distribution over $E(\gk_{k-1})$, and accept the sampled edge $\{u,v\}$ with probability $\frac{1}{T(g)}$ where $g=u \cup v$, the distribution of accepted graphlets will be $\epsilon$-uniform. Let then $X=\{X_t\}_{t \ge 0}$ be the lazy random walk over $\gk_{k-1}$, and for all $t \ge 0$ let $Y_t = X_t \cup X_{t+1}$. Then, $Y_t$ is $\scO\big(\frac{\epsilon}{k^2}\big)$-uniform over $E(\gk_{k-1})$ if $X_t$ is $\scO\big(\frac{\epsilon}{k^2}\big)$-uniform distribution over $V(\gk_{k-1})$. This holds since the distributions $\pi_t$ of $X_t$ and $\sigma_t$ of $Y_t$ satisfy $\sigma_t = M \pi_t$, where $M$ is a stochastic matrix. Since for stochastic matrices $\norm{1}{M} \le 1$, we have $\norm{1}{\sigma_t - \sigma} \le \norm{1}{M(\pi_t-\pi)} \le \norm{1}{M} \norm{1}{\pi_t-\pi}$ where $\pi$ and $\sigma$ are the stationary distributions of $Y_t$ and $X_t$. 
Thus we just need to run the walk over $\gk_{k-1}$ for $\tmix_{\epsilon_k}(\gk_{k-1})$ steps where $\epsilon_k = \Theta\big(\frac{\epsilon}{k^2}\big)$. From the proof of Theorem~\ref{thm:tgk} one can immediately see that $\tmix_{\epsilon_k}(\gk_{k-1}) = \kok \scO\left(\tmix_{\epsilon}(G) \, \rho(G)^{k-2}\lg \frac{n}{\epsilon} \right)$. (The $\frac{1}{k^2}$ factor in $\epsilon_k$ is absorbed by $k^{\scO(k)}$). Finally, by Lemma~\ref{lem:step_cost}, each step takes $\scO(\poly(k))$ time in expectation. This completes the proof.

\section{Uniform graphlet sampling}
\label{sec:ugs}
This section presents our uniform graphlet sampling algorithm, \UniformAlgo. The key idea of the algorithm is to make rejection sampling efficient. To understand how, let us first describe why rejection sampling is usually \emph{not} efficient. Suppose we have a generic random process that draws graphlets from $\vk_k$. For each graphlet $g \in \vk_k$ let $p(g)$ be the probability that the process yields $g$, and let $p^* = \min_{g \in \vk_k} p(g)$. In rejection sampling, when we draw $g$, we randomly \emph{accept} it with probability $p^* p(g)^{-1}$. In this way, the probability that $g$ is returned, which equals the probability that $g$ is both sampled \emph{and} accepted, is $p(g)\, p^* p(g)^{-1} = p^*$, which is independent of $g$. This makes the distribution of returned graphlets uniform regardless of $p$. The key problem is that $p^*$ may be very small --- which happens, for instance, if the random process samples graphlets by growing a random spanning tree around a high-degree vertex of $G$. In this case we can have $p^* = \scO(\Delta^{-(k-1)})$, so we may need $\Omega(\Delta^{k-1})$ trials before accepting a graphlet. Unfortunately, all known graphlet sampling algorithms based on rejection sampling suffer from this ``curse of rejection'', and indeed they may need time $\Omega(\Delta^{k-1})$ for sampling just one uniform graphlet in the worst case.

The main idea of \UniformAlgo\ is to circumvent the obstacle by sorting $G$. By doing this, we will virtually partition $\vk_k$ into $n$ buckets $\Va(1),\ldots,\Va(n)$, one for each node of $G$, in such a way that for each $v \in V(G)$ we will know $|\Va(v)|$ with good accuracy. This will constitute our preprocessing phase. In the sampling phase, we will pick $v$ with probability proportional to our estimate of $|\Va(v)|$, and we will sample almost-uniformly from $\Va(v)$. To this end, we note that sampling from $\Va(v)$ amounts to sampling a $k$-graphlet from the subgraph $G(v)$ of $G$ induced by $v$ and all nodes after $v$ in the order. This can be done efficiently since, as we will see, for our purposes $G(v)$ behaves like a regular graph. Moreover, we will be able to compute efficiently all the probabilities involved in this process. This will allow us to reject the sampled subgraph efficiently and with the correct probability, guaranteeing a truly uniform distribution. 

\subsection{A toy example: regular graphs}
Let us build the intuition with a toy example. Suppose that $G$ is $d$-regular. For simplicity suppose that $G$ is connected, too. To begin, we let $p(v)=\frac{1}{n}$ for all $v \in V(G)$, and choose $v$ according to $p$, i.e., uniformly at random. Note that $p(v)$ is roughly proportional to the number of $k$-graphlets containing $v$, which is easily seen to be between $\kokm d^{k-1}$ and $\kok d^{k-1}$, for all $v$. Once we have chosen $v$, we sample a graphlet containing $v$, by running the following random growing process. Set $S_1=\{v\}$, and for $i=2,\ldots,k$, build $S_i$ from $S_{i-1}$ by choosing a random edge in the cut between $S_{i-1}$ and the rest of $G$, and adding to $S_{i-1}$ the other endpoint of the edge. Denote by $p_v(g)$ the probability that $g$ is obtained when the random growing process starts at $v$, and by $p(g)=\sum_{v \in g} p(v) p_v(g)$ the probability that $g$ is obtained. It is easy to show that for any $g \in \vk_k$ we have:
\begin{align}
   \frac{1}{n} \kokm d^{-(k-1)} \le p(g) \le \frac{1}{n} \kok d^{-(k-1)} 
\end{align}
Now we design the rejection step. First, observe that by setting $p^* = \frac{1}{n} k^{-C k} d^{-(k-1)}$ with $C$ large enough, for all $g \in \vk_k$ we will have $p(g) \ge p^*$ and therefore $p^* p(g)^{-1} \le 1$. Moreover, in time $\kok$ we can easily compute $p(g)$ for any given $g$ (this is shown below). In summary, once we have sampled $g$, we can efficiently compute $p_{acc}(g) = p^* p(g)^{-1} \le 1$. Then, we accept $g$ with probability $p_{acc}(g)$. The probability that $g$ is sampled \emph{and} accepted is $p(g) p_{acc}(g) = p^*$, which is independent of $g$. Therefore the distribution of the returned $k$-graphlets is uniform over $\vk_k$. Moreover, by the inequalities above we have $p_{acc}(g) \ge \kokm$, hence we will terminate after $\kok$ rejection trials in expectation. Thus, when $G$ is $d$-regular, we have an efficient uniform graphlet sampling algorithm.

\subsection{The preprocessing phase}
Let $G$ be an arbitrary graph. Our goal is to ``regularize'' $G$, in a certain sense, so that we can apply the scheme of the toy example above. Let us start by introducing some notation. Given an order $\prec$ over $V$, we denote by $G(v)=G[\{u \succeq v\}]$ the subgraph of $G$ induced by $v$ and all nodes after it in the order, and for all $u \in G(v)$, we denote by $d(u|G(v))$ the degree of $u$ in $G(v)$. Before moving to the algorithm, we introduce a definition that is central to the rest of the work.
\begin{definition}
$\prec$ is an $\alpha$-degree-dominating order ($\alpha$-DD order) of $G$ if for all $v$ and all $u \succ v$ we have $d(v|G(v)) \ge \alpha\, d(u|G(v))$.
\end{definition}

Our algorithm starts by computing a $1$-dominating order for $G$, which guarantees that $v$ has the largest degree in $G(v)$. Such an order can be easily computed in time $\scO(n+m)$ by repeatedly removing from $G$ the node of maximum degree~\cite{Matula83-Ordering}. (Later on, we will need to compute approximate $\alpha$-DD orders for $\alpha < 1$ in time roughly $\scO(n \lg n)$, which is not as easy). After computing our $1$-DD order $\prec$, in time $\scO(n+m)$ we also sort the adjacency lists of $G$ accordingly, via bucket sort. This will be used to find efficiently the edges of $G(v)$ via binary search. 

Next, we virtually partition graphlets into \emph{buckets}.
\begin{definition}
The bucket $\Va(v)$ is the set of graphlets whose smallest node according to $\prec$ is $v$.
\end{definition}
\noindent
Clearly, the buckets $\Va(v)$ form a partition of $\vk_k$. Similarly to $d^{k-1}$ in the toy example above, here $d(v|G(v))^{k-1}$ gives a rough estimate of the number of graphlets in $\Va(v)$. Indeed, if $\Va(v) \ne \emptyset$, then we can easily show that:
\begin{align}
    \kokm |\Va(v)| \le d(v|G(v))^{k-1} \le \kok |\Va(v)|
\end{align}
It is easy to see that the $d(v|G(v))$ are known after computing $\prec$. Hence, we will use $d(v|G(v))^{k-1}$ as a proxy for $|\Va(v)|$. In time $\scO(n)$ we compute:
\begin{align}
 Z &= \sum_{v \in V : \Va(v) \ne \emptyset} d(v|G(v))^{k-1}
 \\ p(v) &= \Ind{\Va(v) \ne \emptyset} \cdot \frac{d(v|G(v))^{k-1}}{Z}, \quad \forall v \in V(G)
\end{align}
This defines a distribution $p$ over the buckets that we will use in the sampling phase. Note that, to compute $Z$ and $p$, we must detect whether $\Va(v) = \emptyset$ for each $v$. To this end, we use a BFS from $v$ that explores $G(v)$ and stops as soon as $k$ nodes are found. We can show that this takes time $\scO(k^2 \lg k)$ by listing edges from the end of the adjacency lists. This makes our overall preprocessing time grow to $\scO(n k^2 \lg k +m)$, as claimed in Theorem~\ref{thm:uniform}. This concludes our preprocessing phase. See Algorithm~\ref{alg:preprocess} for the pseudocode.
\begin{algorithm}[h!]
\caption{\Preprocess($G$)} %\Comment{(see Lemma~\ref{lem:dd} for the analysis)}}
\label{alg:preprocess}
\begin{algorithmic}[1]
\State compute $\prec$ using a bucketing algorithm  \Comment{$\scO(n+m)$}
\State sort the adjacency lists of $G$ according to $\prec$  \Comment{$\scO(n+m)$} 
\State compute $\dg{v}{G(v)}$ for all $v \in V(G)$  \Comment{$\scO(n+m)$}  
\For{each $v \in V$}
\State check if $\Va(v) \ne \emptyset$ with a BFS  \Comment{$\scO(k^2 \lg k)$}
\State let $b_v \assign \Ind{\Va(v) \ne \emptyset} \cdot \dg{v}{G(v)}^{k-1}$
\EndFor
\State \Return $\prec$ and $\{b_v\}_{v \in V}$
\end{algorithmic}
\end{algorithm}

\begin{lemma}
\label{lem:dd}
\Preprocess$(G)$ can be implemented to run in time $\scO(n k^2 \lg k + m)$. The output order $\prec$ is a $1$-DD order for $G$ and thus satisfies $\dg{v}{G(v)} \ge \dg{u}{G(v)}$ for all $u \succ v$. The output estimates $b_v > 0$ satisfy $\frac{b_v}{|\Va(v)|} \in \big[k^{-\scO(k)},k^{\scO(k)}\big]$.
\end{lemma}
\begin{proof}
Consider the first two lines of \Preprocess$(G)$. Computing $\prec$ takes time $\scO(n+m)$ by a standard bucketing technique~\cite{Matula83-Ordering}. Sorting the adjacency lists of $G$ according to $\prec$ takes time $\scO(n+m)$ via bucket sort. With one final $\scO(n+m)$-time pass we compute, for each $v$, the position $i_v$ of $v$ in its own sorted adjacency list, from which we compute $\dg{v}{G(v)} = d_v-i_v$ in constant time for each $v$.

Now consider the main loop. Clearly, $\Va(v) \ne \emptyset$ if and only if $k$ nodes are reachable from $v$ in $G(v)$. Thus, we perform a BFS in $G(v)$, starting from $v$, and stopping as soon as $k$ pushes have been made on the queue (counting $v$ as well). To keep track of which nodes have been pushed we can use a dictionary; as we need to hold at most $k$ entries, every insertion and lookup will take time $\scO(\lg k)$. After popping a generic node $u$ from the queue, we proceed as follows. We take every neighbor $z$ of $u$ in reverse order (that is, according to $\succ$). If $z \succ v$ and $z$ has not been pushed, then we push it. As soon as we encounter $z \prec u$, we stop and pop a new node from the queue. Suppose that, after popping $u$, we examine $\ell$ of its neighbors. Then, at least $\ell+1$ nodes must have been pushed so far, since $u$ itself was pushed, and every neighbor examined was certainly pushed (before, or when examined). Thus, for every node $u$ we examine at most $k-1$ neighbors (since we stop the whole algorithm as soon as $k$ nodes are pushed). Since we push at most $k$ nodes in total, we also pop at most $k$ nodes in total. Therefore, we examine a total of $\scO(k^2)$ nodes. Thus, we spend a total time $\scO(k^2 \lg k)$. Summarizing, we obtain a total time bound of $\scO(n+m)+\scO(n k^2 \lg k) = \scO(n k^2 \lg k + m)$.

The claim on $b_v$ follows by Lemma~\ref{lem:count_bound}, since in $G(v)$ node $v$ has maximum degree $\dg{v}{G(v)}$.
\end{proof}

\subsection{The sampling phase}
The sampling phase starts by drawing a node $v$ from the distribution $p$. Using the alias method~\cite{Vose91-alias}, each such random draw takes time $\scO(1)$ after a $\scO(n)$-time-and-space preprocessing (which we do in the preprocessing phase). Once we have drawn $v$, we draw a graphlet from $\Va(v)$ using what we call the \emph{random growing process} at $v$. This is the same process used in the toy example above, but restricted to the subgraph $G(v)$. 
\begin{definition}
The \emph{random growing process} at $v$ is defined as follows: $S_1=\{v\}$, and for each $i=1,\ldots,k-1$, $S_{i+1}=S_i \cup \{u_i,u_i'\}$, where $\{u_i,u_i'\}$ is uniform random over $\Cut(S_i, G(v)\setminus S_i)$.
\end{definition}

Now we make two key observations. First, the random growing process at $v$ returns a roughly-uniform random graphlet of $\Va(v)$, and can be implemented efficiently thanks to the sorted adjacency lists of $G$. Second, the probability that the random growing process returns a specific graphlet $g$ can be computed efficiently, thanks again to the sorted adjacency lists. These two facts are proven below; before, however, we need a technical result about the size of the cuts in $G(v)$.
\begin{lemma}
\label{lem:cuti}
Let $V(G)$ be sorted according to a $1$-DD order. Consider any sequence of sets $S_1,\ldots,S_{k-1}$ such that $S_1=\{v\}$, that $G(v)[S_i]$ is connected for all $i$, and that $S_i = S_{i-1} \cup \{s_i\}$ for some $s_i \in G(v)$. Then for all $i=1,\ldots,k-1$:
\begin{align}
\label{eq:cuti}
    i^{-1} \le \frac{|\Cut(S_i, G(v)\setminus S_i)|}{d(v|G(v))} \le i 
\end{align}
\end{lemma}
\begin{proof}
Let $c_i=|\Cut(S_i, G(v)\setminus S_i)|$ for short. For the lower bound, note that $\cut_i \ge 1$ for all $i=1,\ldots,k-1$ since $G[S_k]$ is connected. Moreover $\cut_1=\dg{v}{G(v)}$ since $S_1=\{v\}$. Now, if $\cut_1 \le i$ then $\frac{\cut_1}{i} \le 1$ and therefore $\cut_i \ge \frac{\cut_1}{i}$. If instead $\cut_1 \ge i$, since the degree of $v$ in $S_i$ is at most $i-1$, then the cut of $S_i$ still contains at least $\cut_1-|S_i|+1 \ge \cut_1-(i-1)$ edges. Therefore:
\begin{align}
 \cut_i \ge \cut_1-(i-1) \ge \cut_1 - (i-1)\frac{\cut_1}{i} = \frac{\cut_1}{i} = \frac{1}{i}\cdot\dg{v}{G(v)}
\end{align}
For the upper bound, note that:
\begin{align}
 \cut_i \le \sum_{u \in S_i} \dg{u}{G(v)} \le \sum_{u \in S_i} \dg{v}{G(v)} = i \cdot\dg{v}{G(v)}
\end{align}
where we used the fact that $v$ is the maximum-degree node of $G(v)$.
\end{proof}

\iffalse
\begin{lemma}
\label{lem:cuti2}
Assume that $V(G)$ is sorted by an $(\alpha,\beta)$-DD order $(\prec,\ba)$ and fix any $v \in V(G)$ such that $b_v>0$.
Consider any sequence of sets $S_1,\ldots,S_{k-1}$ such that $S_1=\{v\}$, that $G(v)[S_i]$ is connected for all $i$, and that $S_i = S_{i-1} \cup \{s_i\}$ for some $s_i \in G(v)$.
Then, for all $i=1,\ldots,k-1$, we have:
\begin{align}
\label{eq:cuti2}
    \frac{1}{i} \cdot d(v|G(v)) \le |\Cut(S_i, G(v)\setminus S_i)| \le \frac{i}{\alpha} \cdot d(v|G(v))
\end{align}
\end{lemma}
\begin{proof}
The lower bound holds by the same proof of Lemma~\ref{lem:cuti}.
For the upper bound, each $u \in S_i \setminus \{v\}$ satisfies $u \succ v$.
Thus, since $b_v > 0$, by Definition~\ref{def:ab_order} we have $\dg{u}{G(v)} \le \frac{1}{\alpha} \dg{v}{G(v)}$.
This holds for $u=v$ as well since $\alpha \le 1$.
Therefore:
\begin{align}
 \cut_i \le \sum_{u \in S_i} \dg{u}{G(v)} \le \sum_{u \in S_i} \frac{1}{\alpha} \dg{v}{G(v)} \le \frac{i}{\alpha} \cdot\dg{v}{G(v)}
\end{align}
concluding the proof.
\end{proof}
\fi

Algorithms \sampleS$(G, v)$ and \computeP$(G, S)$ below gives the pseudocode of the random growing process and of the algorithm for computing the probability that the process returns a particular graphlet. Lemma~\ref{lem:subgraph_sample} shows that \sampleS$(G, v)$ can be implemented efficiently and that it returns a graphlet that is roughly uniform. Lemma~\ref{lem:prob_compute} shows that \computeP$(G, S)$ is correct and efficient.

\begin{algorithm}
\caption{\sampleS$(G, v)$}
\label{alg:sample_S}
\begin{algorithmic}[1]
\State $S_1 \assign \{v\}$
\For{$i = 1,\ldots,k-1$}
\For{$u \in S_{i}$}
\State $\cut_i(u) \assign \dg{u}{G(v)} - $(degree of $u$ in $G[S_i]$) %\Comment{cut between $u$ and $G(v) \setminus S_i$}
\EndFor
\State draw $u$ with probability $\frac{\cut_i(u)}{\sum_{z \in S_{i}}\cut_i(z)}$ \label{line:u}
\State draw $u'$ u.a.r.\ from the neighbors of $u$ in $G(v) \setminus S_i$
\label{line:u1}
\State $S_{i+1} \assign S_{i} \cup \{u'\}$ 
\EndFor
\State \Return $S_k$
\end{algorithmic}
\end{algorithm}

\begin{algorithm}[h!]
\caption{\computeP$(G, S=\{v,u_2,\ldots,u_k\})$}
\label{alg:compute_P}
\begin{algorithmic}[1]
\State $p \assign 0$
\For{each permutation $\sigma=(\sigma_2,\ldots,\sigma_{k})$ of $u_2,\ldots,u_k$} \label{line:for_sigma}
\State $p_{\sigma} \assign 1$
\For{each $i = 1,\ldots,k-1$}
\State $S_i \assign \{v,\sigma_2,\ldots,\sigma_i\}$
\State $n_i \assign $ number of neighbors of $\sigma_{i+1}$ in $S_i$
\State $\cut_i(u) \assign \dg{u}{G(v)} - $(degree of $u$ in $G[S_i]$) %\Comment{cut between $u$ and $G(v) \setminus S_i$}
\State $p_{\sigma} \assign p_{\sigma} \cdot \frac{n_i}{\cut_i}$
\EndFor
\State $p \assign p + p_{\sigma}$
\EndFor
\State \Return $p$
\end{algorithmic}
\end{algorithm}

\begin{lemma}
\label{lem:subgraph_sample}
Suppose $G$ is sorted according to a $1$-DD order and choose any $v$ such that $\Va(v) \ne \emptyset$. Then \sampleS$(G, v)$ runs in time $\scO(k^3 \lg \Delta)$. Moreover, for any $g=G[S] \in \Va(v)$, the probability $p(S)$ that \sampleS$(G, v)$ returns $S$ is between $\frac{1}{(k-1)!} \dg{v}{G(v)}^{-(k-1)}$ and $(k-1)!^3 \dg{v}{G(v)}^{-(k-1)}$.
\end{lemma}
\begin{proof}
\emph{Running time.}
Consider one iteration of the main loop. For every $u \in S_i$, computing $\cut_i(u)$ takes time $\scO(k \lg \Delta)$. Indeed, in time $\scO(\lg \Delta)$ we locate the position of $v$ in the adjacency list of $u$, which subtracted from $d_u$ yields $\dg{u}{G(v)}$. Then, we compute the number of neighbors of $u$ in $S_i$ in time $\scO(k)$ using edge queries.
Thus, the cycle over $u \in S_i$ takes $\scO(k^2 \lg \Delta)$ in total.
Drawing $u$ takes time $\scO(k)$.
Finally, drawing $u'$ takes $\scO(k)$ as well. 
To see this, note that if $u$ had no neighbors in $S_i$, then we could just draw a node uniformly at random from the last $\dg{u}{G(v)}$ elements of the adjacency list of $u$.
However, $u$ has neighbors in $S_i$.
But we still know the (at most $k$) disjoint sublists of the adjacency lists containing the neighbors in the cut.
Thus we can draw a uniform integer $j \in [\cut_i(u)]$ and select the $j$-th neighbor of $u$ in $G(v) \setminus S_i$ in time $\scO(k)$.
This proves that one iteration of the main loop of \sampleS\ takes time $\scO(k^2 \lg \Delta)$.
Thus, \sampleS\ runs in time $\scO(k^3 \lg \Delta)$.

\emph{Probability.} Consider any $S$ such that $g=G[S] \in \Va(v)$. Thus $S$ is a $k$-node subset such that $v \in S$ and that $G[S]$ is connected. We compute an upper bound and a lower bound on the probability $p(S)$ that the algorithm returns $S$.

Clearly, there are at most $(k-1)!$ sequences of nodes that \sampleS$(G,v)$ can choose to produce $S$. Fix any such sequence, $v,u_2,\ldots,u_k$, and let $S_i=\{v,\ldots,u_i\}$. Let $\cut_i(u)=|\Cut(u, G(v) \setminus S_i)|$, and let $\cut_i = \sum_{u \in S_i} \cut_i(u) = |\Cut(S_i, G(v)\setminus S_i|$. By construction, $S_{i+1}$ is obtained by adding to $S_i$ the endpoint $u_{i+1}$ of an edge chosen uniformly at random in $\Cut(S_i,G(v) \setminus S_i)$. Thus, for any $u' \in G(v) \setminus S_i$, we have:
\begin{align}
 \Pr(u_{i+1}=u') = \frac{\dg{u'}{S_i \cup u'}}{\cut_i}  \le \frac{i^2}{\dg{v}{G(v)}}
\end{align}
where in the inequality we used the facts that $\dg{u'}{S_i \cup u'} \le i$ is the number of neighbors of $u'$ in $S_i$, and that $\cut_i \ge \frac{\dg{v}{G(v)}}{i}$ by Lemma~\ref{lem:cuti}.
Thus the probability that \sampleS$(G, v)$ draws the particular sequence $v,u_2,\ldots,u_k$ is at most $\prod_{i=1}^{k-1}\frac{i^2}{\dg{v}{G(v)}} = (k-1)!^2 \dg{v}{G(v)}^{-(k-1)}$.
Since there are at most $(k-1)!$ sequences, $p(S) \le (k-1)!^3 \dg{v}{G(v)}^{-(k-1)}$.

On the other hand, since $G[S]$ is connected, then there is at least one sequence $v,u_2,\ldots,u_k$ such that $\cut_i \ge 1$ for all $i=1,\ldots,k-1$, which therefore satisfies:
\begin{align}
\label{eq:LB_prob_Si}
 \Pr(u_{i+1}=u') = \frac{\dg{u'}{S_i \cup u'}}{\cut_i}  \ge \frac{1}{i \, \dg{v}{G(v)}}
\end{align}
where we used the facts that $\dg{u'}{S_i \cup u'} \ge 1$, since $u'$ is a neighbor of some $u \in S_i$, and that $\cut_i \le i \cdot \dg{v}{G(v)}$, by Lemma~\ref{lem:cuti}. So the probability that \sampleS$(G, v)$ draws this particular sequence is at least $\prod_{i=1}^{k-1} \frac{1}{i \dg{v}{G(v)}} = \frac{1}{(k-1)!} \dg{v}{G(v)}^{-(k-1)}$, which is a lower bound on $p(S)$.
\end{proof}

\begin{lemma}
\label{lem:prob_compute}
\computeP$(G, S=\{v,u_2,\ldots,u_k\})$ runs in time $\poly(k) \scO(k! \lg \Delta)$ and outputs the probability $p(S)$ that \sampleS$(G, v)$ returns $S$.
\end{lemma}
\begin{proof}
The proof is essentially the same of Lemma~\ref{lem:subgraph_sample}.
\end{proof}

%\subsubsection{The Rejection Step}
We can now complete the sampling phase by performing a rejection step. After drawing $v$ with probability $p(v)$, we draw a random graphlet $g$ from $\Va(v)$ by invoking \sampleS$(G, v)$, and we compute $p_v(g)$ by invoking \computeP$(G,g)$. By construction, the overall probability that we have drawn $g$ is $p(g) = p(v) \cdot p_v(g)$. By the definition of $p(v)$ and by Lemma~\ref{lem:subgraph_sample}:
\begin{align}
\label{eq:check_1}
    \kokm \frac{1}{Z}  \le p(v) \cdot p_v(g) \le \kok \frac{1}{Z}
\end{align}
We therefore set the acceptance probability to:
\begin{align}
\label{eq:check_2}
p_{acc}(g) =\frac{k^{-C k}}{p(v) \cdot p_v(g) \cdot Z}    
\end{align}
This makes the probability that $g$ is sampled \emph{and} accepted equal to:
\begin{align}
    p(g) \cdot p_{acc}(g) = p(v) \cdot p_v(g) \cdot \frac{k^{-C k}}{p(v) \cdot p_v(g) \cdot Z} = \frac{k^{-Ck}}{Z}
\end{align}
which is independent of $g$ and thus constant over $\vk_k$. For $C$ large enough,~\eqref{eq:check_1} and~\eqref{eq:check_2} imply $p_{acc} \in[\kokm,1]$. Therefore, $p_{acc}(g)$ is a valid probability, and moreover, we will accept a graphlet after $\kok$ trials in expectation. As by Lemma~\ref{lem:subgraph_sample} and Lemma~\ref{lem:prob_compute} the running time of a single trial is $\poly(k) \lg \dmax$, the total expected time per sample is $\kok \lg \dmax$, as claimed in Theorem~\ref{thm:uniform}.

To wrap up, Algorithm~\ref{alg:uniform} gives the main body of \UniformAlgo.

\begin{algorithm}[h!]
\caption{\UniformAlgo$(G)$}
\label{alg:uniform}
\begin{algorithmic}[1]
\State $(\prec, \{b_v\}_{v \in V}) \assign$ \Preprocess($G$)
\State let $Z \assign \sum_{v \in V} b_v$
\State let $p(v) \assign \frac{b_v}{Z}$ for each $v \in V$
\State let $\beta_k(G) \assign \frac{1}{k! \,Z}$ 
\State
\Function{\textsc{sample}}{\,}
\While{true}
\State draw $v$ from the distribution $p$
\State $S \assign$ \sampleS$(G, v)$
\State $p(S) \assign$ \computeP$(G, S)$
\State with probability $\frac{\beta_k(G)}{p(v)\, p(S)}$ \Return $S$
\EndWhile
\EndFunction
\end{algorithmic}
\end{algorithm}

\section{Epsilon-uniform graphlet sampling}
\label{sec:epsugs}
This section describes our $\epsilon$-uniform graphlet sampling algorithm, \EpsilonAlgo. At a high level, \EpsilonAlgo\ is an adaptation of \UniformAlgo. To begin, we observe that \UniformAlgo\ relies on the following key ingredients. First, the vertices of $G$ are sorted according to a $1$-DD order $\prec$, which ensures that each subgraph $G(v)$ behaves like a regular graph for what concerns sampling (Lemma~\ref{lem:subgraph_sample}). Second, the edges of $G$ are sorted according to $\prec$ as well, which makes it possible to compute the size of the cuts $|\Cut(u, G(v)\setminus S_i)|$ in time proportional to $\lg \dmax$ (Lemma~\ref{lem:subgraph_sample} and~\ref{lem:prob_compute}). Unfortunately, both ingredients require a $\Theta(m)$-time preprocessing. To reduce the preprocessing time to $\scO(n \lg n)$, we introduce:
\begin{enumerate}[itemsep=2pt,leftmargin=2em]
    \item A preprocessing routine that computes w.h.p.\ an approximate $\alpha$-DD order, together with good bucket size estimates. By ``approximate'' we also mean that some buckets might be erroneously deemed empty, but we guarantee that those buckets contain a fraction $\le \epsilon$ of all graphlets.
    \item A sampling routine that emulates the one of \UniformAlgo, but replaces the exact cut sizes with additive approximations. These approximations are good enough that, with good probability, \EpsilonAlgo\ behaves as \UniformAlgo, including the rejection step.
\end{enumerate}
Achieving these guarantees is not just a matter of sampling and concentration bounds. For instance, to obtain an $\alpha$-DD order, we cannot just sub-sample the edges of $G$ and compute the $1$-DD order on the resulting subgraph: the sorting process would introduce correlations, destroying concentration. Similarly, we cannot just compute a multiplicative estimate of $|\Cut(u, G(v)\setminus S_i)|$: without sorted lists this would require $\Omega(\dmax)$ queries, as we might have $d_u = \dmax$ and $|\Cut(u, G(v)\setminus S_i)|=1$. Similar obstacles arise in estimating $p_v(g)$.

\subsection{Approximating a degree-dominating order}
We introduce our notion of approximate degree-dominating order. In what follows, $\ba =(b_v)_{v \in V}$ is a vector of bucket size estimates.
\begin{definition}
\label{def:ab_order}
A pair $(\prec,\ba)$ where $\ba = (b_v)_{v \in V}$ is an $(\alpha,\beta)$-DD order for $G$ if:
\begin{enumerate}[itemsep=2pt,topsep=4pt,parsep=2pt,partopsep=2pt]
\item $\sum_{v : b_v > 0} |\Va(v)| \ge (1-\beta) \sum_{v} |\Va(v)|$
\item $b_v > 0$ $\,\Longrightarrow\,$ $\kokm\beta \le \frac{b_v}{|\Va(v)|} \le \kok \frac{1}{\beta}$ 
\item $b_v > 0$ $\,\Longrightarrow\,$ $\dg{v}{G(v)} \ge \alpha \, d_v \ge \alpha \, \dg{u}{G(v)}$ for all $u \succ v$
\item $v \prec u$ $\,\Longrightarrow\,$ $d_v \ge 3k \alpha\, d_u$
\end{enumerate}%\vspace*{2pt}
\end{definition}
\noindent Let us elaborate on this. The first property says that the buckets that are deemed nonempty hold a fraction $1-\beta$ of all graphlets. The second property says that every bucket that is deemed nonempty comes with a good estimate of its size. The third property says that $\prec$ is an $\alpha$-DD order if restricted to the buckets that are deemed nonempty, and gives an additional guarantee on $d_v$. The fourth property will be used later on. The idea is that, if we look only at buckets that are deemed nonempty, we will have guarantees similar to a $1$-DD order; but bear in mind that here the edges of $G$ will not be sorted, and this will complicate things significantly.

The algorithm below, \GraphSort$(G, \beta)$, computes efficiently an ($\alpha, \beta$)-DD order with $\alpha=\beta^{\frac{1}{k-1}}\frac{1}{6k^3}$. This will be enough for our purposes.
In the remainder we prove Lemma~\ref{lem:approx_order} and Lemma~\ref{lem:approx_order_time}, which give the guarantees of \GraphSort$(G, \beta)$.
%Lemma~\ref{lem:approx_order_simple} follows immediately, as well as the running time bound on the preprocessing phase of \EpsilonAlgo\ for all $k \ge 3$.
For technical reasons, instead of $\alpha$ the proofs and the algorithm use $\eta = \alpha k = \epsilon^{\frac{1}{k-1}} \frac{1}{6 k^2}$.
The intuition of the algorithm is the following. We start at round $t=0$ with $\prec_0$ being the order of $V(G)$ by nonincreasing degree; this corresponds to the optimistic guess that $d(v|G(v))=d_v$ for all $v$. Then, we take every $v \in V(G)$ in the order of $\prec_0$, and we check if $d(v|G(v))$ is indeed close of $d_v$. To this end we sample $\frac{1}{\eta^2} \log n$ random neighbors of $v$, for some appropriate $\eta$ and check how many are after $v$ in $\prec$. If that fraction is at least $\eta$, then we let $b_v=(d_v)^{k-1}$ and set $\prec_{t+1}=\prec_t$ unchanged. Otherwise, we let $b_v = 0$ and update $\prec_{t+1}$ from $\prec_t$ by pushing $v$ to its ``correct'' position. This is enough for vertices of sufficiently high degree, but not for those of small degree. Indeed, for those vertices $B(v)$ might be empty even though $d(v|G(v))$ is close to $d_v$, just because $d(v|G(v))$ is small in an absolute sense. Hence, for vertices of small degree we check whether $B(v) \ne \emptyset$ explicitly.

\label{apx:eps_code}
\begin{algorithm}[h!]
\caption{\GraphSort($G, \beta)$}
\label{alg:sketch}
\begin{algorithmic}[1]
\State let $\eta \assign \beta^{\frac{1}{k-1}} \frac{1}{6k^2}$ and $h \assign \Theta(\eta^{-2}\log{n})$
\State init $s_v \assign d_v$ for all $v \in V$ \State init ${\prec}$ so that $u \prec v \iff (s_u > s_v) \vee ((s_u = s_v) \wedge (u > v))$ \label{line:init_prec}
\For{each $v$ in $V$ in nonincreasing order of degree} 
\State sample $h$ neighbors $x_1,\ldots, x_h$ of $v$ u.a.r.\
\State let $X \assign \sum_{j=1}^h \Ind{x_j \succ v}$
\If{$X \ge 2 \eta h$}
\State let $b_v \assign (d_v)^{k-1}$
\Else
\State let $b_v \assign 0$ and $s_v \assign 3 \eta\, d_v$
\State update $\prec$ so that $u \prec v \iff (s_u > s_v) \vee ((s_u = s_v) \wedge (u > v))$ \label{line:update_prec}
\EndIf
\EndFor
\For{each $v : d_v \le k/\eta$} 
\If{$\Va(v) \ne \emptyset$} %\Comment{check via BFS}
\State compute $\dg{v}{G(v)}$ and let $b_v \assign \dg{v}{G(v)}^{k-1}$
\Else
\State let $b_v \assign 0$
\EndIf
\EndFor
\State \Return  $\prec$ and $\{b_v\}_{v \in V}$
\end{algorithmic}
\end{algorithm}

\begin{lemma}
\label{lem:approx_order}
With high probability \GraphSort$(G, \beta)$ returns an $(\alpha,\beta)$-DD order for $G$ with $\alpha=\beta^{\frac{1}{k-1}}\frac{1}{6k^3}$.
\end{lemma}
\begin{lemma}
\label{lem:approx_order_time}
\GraphSort$(G, \beta)$ can be implemented to run in time $\scO\Big(\beta^{-\frac{2}{k-1}} k^6 \, n \lg n\Big)$.
\end{lemma}
To carry out the proofs, we need some notation and a few observations about \GraphSort$(G, \beta)$. We denote by:
\begin{itemize}\itemsep0pt
\item $t=1,\ldots,n$ the generic round of the first loop
\item $\prec_t$ the order $\prec$ at the very beginning of round $t$
\item $G_t(v)=G[\{z:z \succeq_{t} v\}]$ the subgraph induced by $v$ and the vertices after it at time $t$
\item $\dg{u}{G_t(v)}$ the degree of $u$ in $G_t(v)$; obviously $\dg{u}{G_t(v)} \le d_u$ for all $t$
\item $t_v$ the round where $v$ is processed
\item $s_t(v)$ the value of $s_v$ at the beginning of round $t$; note that $s_t(v) \ge s_{t+1}(v)$, that $s_{t_v}(v) = d_v$, and that $s_{n+1}(v) = s_{t_{v}+1}(v) \in \{d_v, 3 \eta d_v\}$
\item $X_j = \Ind{x_j \succ v}$ and $X=\sum_{j=1}^h X_j$, in a generic round
\end{itemize}
We denote the returned order by $\prec_n$ (formally it would be $\prec_{n+1}$ but clearly this equals $\prec_n$), and by $G_n(\cdot)$ the subgraphs induced in $G$ under such an order.
By $b_v$ we always mean the value of $b_v$ at return time, unless otherwise specified.

\begin{observation}
\label{ob:order}
For any $t$, if $u \succ_t v$ then $s_t(u) \le s_t(v)$.
\end{observation}
\begin{proof}
By definition $u \succ_t v$ if and only if $(s_t(u) < s_t(v)) \vee ((s_t(u) = s_t(v)) \wedge (u < v))$.
Therefore in particular $s_t(u) \le s_t(v)$.
\end{proof}

\begin{observation}
\label{obs:Gn_Gu}
If $u \prec_{t_u} v$ then $G_{n}(v) \subseteq G_{t_u}(u)$ and $\dg{u}{G_{n}(v)} \le \dg{u}{G_{t_u}(u)}$.
\end{observation}
\begin{proof}
By definition, $G_{n}(v) \subseteq G_{t_u}(u)$ means $\{z:z \succeq_n v\} \subseteq \{z:z \succeq_{t_u} u\}$, which is equivalent to $\{z:z \prec_{t_u} u\} \subseteq \{z:z \prec_n v\}$.
%; the degree inequality follows by monotonicity of $\dg{u}{\cdot}$.
Consider then any $z : z \prec_{t_u} u$. This implies $z \prec_{t_u} v$ (since $u \prec_{t_u} v$) and $t_z < t_u$ (since $t_z > t_u$ would imply $z \succ_{t_u} u$). But $z$ cannot be moved past $v$ in any round $t' > t$. Thus $z \prec_n v$. Therefore $\{z:z \prec_{t_u} u\} \subseteq \{z:z \prec_{n} v\}$, as desired. The second claim follows by the monotonicity of $\dg{u}{\cdot}$.
\end{proof}

\begin{observation}
\label{ob:degrees}
For all $v$ and all $t \ge t_v$ we have $\dg{v}{G_{t_v}(v)} \ge \dg{v}{G_{t}(v)}$, with equality if $b_v > 0$.
\end{observation}
\begin{proof}
Consider any $z:z \prec_{t_v} v$. Note that $t_z < t_v$, hence $z \notin G_{t_v}(v)$ by definition of $G_{t_v}(v)$. Moreover $z$ will never be moved past $v$ in any round $t \ge t_v$, so $z \notin G_{t}(v)$ as well. Therefore $G_{t_v}(v) \supseteq G_{t}(v)$ for all $t \ge t_v$. Now the claim follows by monotonicity of $\dg{v}{\cdot}$, and by noting that if $b_v > 0$, then $v$ is not moved at round $t_v$ and thus $G_{t}(v) = G_{t_v}(v)$ for all $t \ge t_v$.
\end{proof}

\begin{observation}
\label{ob:hoeff}
In any round, conditioned on past events, w.h.p.\ $|X - \E X \big| \le \eta h$.
\end{observation}
\begin{proof}
Consider round $t_v$.
Conditioned on past events, the $X_j$ are independent binary random variables.
Therefore by Hoeffding's inequality:
\begin{align}
\Pr(|X - \E X| > h \eta) < 2 e^{-2h\eta^2} = e^{-\Theta(\lg{n})} = n^{-\Theta(1)}
\end{align}
where $h = \Theta(\eta^{-2}\lg{n})$ and thus the $\Theta(1)$ at the exponent can be chosen arbitrarily large.
\end{proof}

\begin{observation}
\label{ob:s_u}
With high probability, $\dg{v}{G_{t}(v)} \le s_{n+1}(v)$ for every $v$ anytime in any round $t$.
\end{observation}
\begin{proof}

If $s_{n+1}(v)=d_v$ then clearly $s_{n+1}(v) \ge \dg{v}{G_{t}(v)}$. Suppose instead that $s_{n+1}(v) = 3 \eta d_v$. By Observation~\ref{ob:degrees}, $\dg{v}{G_{t}(v)} \le \dg{v}{G_{t_v}(v)}$. So, we only need to show that with high probability $\dg{v}{G_{t_v}(v)} \le 3 \eta d_v$. Consider the random variable $X=\sum_{j=1}^h X_j$ at round $t_v$, and note that $\E X_j = \frac{\dg{v}{G_{t_v}(v)}}{d_v}$ for all $j$. Therefore, if $\dg{v}{G_{t_v}(v)} > 3 \eta d_v$, then $\E X > 3 \eta h$. Now, the algorithm updates $s_v$ only if $X < 2 \eta h$. This implies the event $X < \E X - \eta h$, which by Observation~\ref{ob:hoeff} fails with high probability. Thus with high probability $\dg{v}{G_{t_v}(v)} \le 3 \eta d_v$.
\end{proof}

\begin{observation}
\label{obs:eta_dv}
If round $t_v$ of the first loop sets $b_v > 0$ then w.h.p.\ $\dg{v}{G_{t_v}(v)} \ge \eta d_v$, else w.h.p.\ $\dg{v}{G_{t_v}(v)} \le 3 \eta d_v$.
If the second loop sets $b_v > 0$ then $\dg{v}{G_{t_v}(v)} \ge \frac{\eta}{k} d_v$ deterministically.
\end{observation}
\begin{proof}
The first claim has the same proof of Observation~\ref{ob:s_u}: if $\dg{v}{G_{t_v}(v)} < \eta d_v$ then $\E X < \eta h$, so $b_v > 0$ implies $X \ge 2 \eta h$ and thus $X > \E X + \eta h$.
Similarly, if $\dg{v}{G_{t_v}(v)} > 3 \eta d_v$ then $\E X > 3 \eta h$, so letting $b_v = 0$ implies $X < 2 \eta h$ which means $X < \E X - \eta h$.
Both events fail with high probability by Observation~\ref{ob:hoeff}.
For the second claim, note that the second loop sets $b_v > 0$ only if $d_v \le \frac{k}{\eta}$, which implies $\frac{\eta}{k} d_v \le 1$, and if $\Va(v) \ne \emptyset$, which implies $\dg{v}{G_{n}(v)} \ge 1$.
Thus, $\dg{v}{G_{n}(v)} \ge  \frac{\eta}{k} d_v$.
Observation~\ref{ob:degrees} gives $\dg{v}{G_{t_v}(v)} \ge \dg{v}{G_{n}(v)}$, concluding the proof.
\end{proof}

\begin{observation}
\label{obs:domin}
With high probability, for all $v$, for all $u \succ_n v$ we have $\dg{u}{G_n(v)} \le s_{n+1}(v)$.
\end{observation}
\begin{proof}
Consider the beginning of round $t_u$. Suppose that $v \prec_{t_u} u$, which implies $t_v < t_u$. Then:
\begin{align}
    s_{n+1}(v) &= s_{t_v+1}(v)
    \\ &= s_{t_u}(v) && \text{since $t_v < t_u$}
    \\ &\ge s_{t_u}(u) && \text{Observation~\ref{ob:order}, using $u \succ_{t_u} v$}
    \\ &= d_u && \text{by construction}
    \\ &\ge \dg{u}{G_n(v)}
\end{align}
Suppose instead $u \prec_{t_u} v$. Then, with high probability:
\begin{align}
    s_{n+1}(v) &\ge s_{n+1}(u) && \text{Observation~\ref{ob:order}, using $u \succ_n v$}
    \\&\ge \dg{u}{G_{t_u}(u)} && \text{Observation~\ref{ob:s_u}, with $t=t_u$}
    \\&\ge \dg{u}{G_{n}(v)} && \text{Observation~\ref{obs:Gn_Gu}}
\end{align}
In any case, $\dg{u}{G_{n}(v)} \le s_{n+1}(v)$.
\end{proof}

\begin{proof}[Proof of Lemma~\ref{lem:approx_order}]
For technical reasons we prove the four properties of $(\prec,\ba)$, see Definition~\ref{def:ab_order}, in a different order.
Moreover, we substitute $\alpha = \frac{\eta}{k}$.
This yields the four properties: 
\begin{enumerate}\itemsep2pt
\item if $v \prec u$, then $d_v \ge 3 \eta \, d_u$
\item if $b_v > 0$, then $\dg{v}{G(v)} \ge \frac{\eta}{k} d_v \ge \frac{\eta}{k} \cdot \dg{u}{G(v)}$ for all $u \succ v$
\item if $b_v > 0$ then $\kokm \beta \le \frac{b_v}{|\Va(v)|} \le \frac{\kok}{\beta}$
\item $\sum_{v : b_v > 0} |\Va(v)| \ge (1-\beta) \sum_{v \in V} |\Va(v)|$
\end{enumerate}

\paragraph{Proof of (1)}
Simply note that $d_v \ge s_{n+1}(v) \ge s_{n+1}(u) \ge 3 \eta d_u$, where the middle inequality holds by Observation~\ref{ob:order} since $u \succ_{n} v$.

\paragraph{Proof of (2)}
Consider any $u \succ_{n} v$ with $b_v > 0$.
Then with high probability:
\begin{align}
\dg{v}{G_{n}(v)} &= \dg{v}{G_{t_v}(v)} && \text{Observation~\ref{ob:degrees}, using $b_v > 0$ and $t=n$}
\\
&\ge \frac{\eta}{k} d_v &&\text{Observation~\ref{obs:eta_dv}, using $b_v > 0$}
\\
&= \frac{\eta}{k} s_{n+1}(v) && \text{by the algorithm, since $b_v > 0$}
\\
&\ge \frac{\eta}{k} \dg{u}{G_n(v)} && \text{Observation~\ref{obs:domin}, since $u \succ_n v$}
\end{align}

\paragraph{Proof of (3)}
First, we show that if $b_v > 0$ then w.h.p.\ $|\Va(v)| \ge 1$. If $v$ is processed by the second loop, then $b_v > 0$ if and only if $|\Va(v)| \ge 1$. Otherwise, we know $d_v > \frac{k}{\eta}$, and w.h.p.:
\begin{align}
d(v|G_n(v)) &= d(v|G_{t_v}(v)) && \text{Observation~\ref{ob:degrees}, using $b_v > 0$}
\\&\ge \eta d_v && \text{Observation~\ref{obs:eta_dv}, using $b_v > 0$}
\\&> k && \text{as } d_v > \frac{k}{\eta}
\end{align}
So w.h.p.\ $d(v|G_n(v)) > k$, in which case $G_n(v)$ contains a $k$-star centered in $v$, implying $|\Va(v)|\ge 1$.

Thus we continue under the assumption $|\Va(v)| \ge 1$.
To ease the notation define $d_v^*=d(v|G_n(v))$ and $\Delta_v^*=\max_{u \in G(v)}d(u|G_n(v))$.
Lemma~\ref{lem:count_bound} applied to $G_n(v)$ yields:
\begin{align}
 k^{-\scO(k)}\, (d_v^*)^{k-1} \le |\Va(v)| \le k^{\scO(k)}\, (\Delta_v^*)^{k-1}
\end{align}
We now show that w.h.p.:
\begin{align}
\beta k^{-\scO(k)} (\Delta_v^*)^{k-1} \le b_v \le \frac{1}{\beta} k^{\scO(k)} (d_v^*)^{k-1}
\end{align}
which implies our claim.
For the upper bound, note that by construction $b_v \le (d_v)^{k-1}$ and that, by point (2) of this lemma, w.h.p.\ $d_v \le \frac{k}{\eta} d^*_v$.
Substituting $\eta$ we obtain:
\begin{align}
b_v \le (d_v)^{k-1} \le (d_v^*)^{k-1}\bigg(\frac{k}{\eta}\bigg)^{k-1} = \frac{1}{\beta} k^{\scO(k)} (d_v^*)^{k-1}
\end{align}
For the lower bound, note that since $b_v>0$ then $b_v \ge (d_v^*)^{k-1}$.
Indeed, if $b_v > 0$, then either $b_v=(d_v)^{k-1} \ge (d_v^*)^{k-1}$ from the first loop, or $b_v=(d_v^*)^{k-1}$ from the second loop (since the value $d(v|G(v))$ in the second loop equals $d(v|G_n(v))$, that is, $d_v^*$).
Now, point (2) of this lemma gives $d_v^* \ge \frac{\eta}{k} \cdot \dg{u}{G_n(v)}$ for all $u \succ_n v$.
Thus $\Delta^*_v = \max_{u \in G(v)} \dg{u}{G_n(v)} \le \frac{k}{\eta} d_v^*$.
Therefore:
\begin{align}
(\Delta_v^*)^{k-1} \le \bigg(\frac{k}{\eta} d_v^*\bigg)^{k-1} = \beta \kok (d_v^*)^{k-1} \le \beta \kok b_v
\end{align}

\paragraph{Proof of (4)}
We prove the equivalent claim:
\begin{align}
\sum_{v : b_v = 0} |\Va(v)| \le \beta \sum_{v \in V} |\Va(v)|
\end{align}
Consider any $v$ with $b_v=0$ and $|\Va(v)| > 0$. These are the only vertices contributing to the left-hand summation. First, we note that $d_v > \frac{k}{\eta}$. Indeed, if $|\Va(v)| > 0$ and $d_v \le \frac{k}{\eta}$, then the second loop of \GraphSort\ processes $v$ and sets $b_v = (\dg{v}{G_n(v)})^{k-1}$, which is positive since $|\Va(v)| > 0$ implies $\dg{v}{G(v)} > 0$. Thus, we can assume that $b_v=0$, $|\Va(v)| \ge 1$, $d_v > \frac{k}{\eta}$, and $v$ is not processed in the second loop. Since $d_v > \frac{k}{\eta} > k-1$, then $G$ contains at least ${d_v \choose k-1} \ge 1$ stars centered around $v$. Each such star contributes $1$ to $\sum_{v \in V} |\Va(v)|$. Since $\frac{(d_v)^{k-1}}{(k-1)^{k-1}} \le {d_v \choose k-1}$ whenever ${d_v \choose k-1} \ge 1$, we obtain:
\begin{align}
\label{eq:Bv_bound_1}
\sum_{v:b_v=0} \frac{(d_v)^{k-1}}{(k-1)^{k-1}} \le \sum_{v:b_v=0}{d_v \choose k-1} \le k \sum_{v : b_v=0} |\Va(v)| \le k \sum_{v} |\Va(v)|
\end{align}
where the factor $k$ arises from each star being counted up to $k$ times by the left-hand side (once for each vertex in the star).

On the other hand, by Observation~\ref{ob:s_u} and Observation~\ref{obs:domin}, w.h.p.\ $\dg{v}{G_n(v)} \le s_{n+1}(v)$ and $\dg{u}{G_n(v)} \le s_{n+1}(v)$ for all $u \succ_n v$.
Hence, the maximum degree of $G_v$ is w.h.p.\ at most $s_{n+1}(v)$.
But $s_{n+1}(v) = 3 \eta d_v$, since $b_v=0$ is set in the first loop.
Thus, the maximum degree of $G_v$ is at most $3 \eta (d_v)$.
By Lemma~\ref{lem:count_bound}, then,
\begin{align}
\label{eq:Bv_bound_2}
\sum_{v:b_v=0}|\Va(v)| \le \sum_{v:b_v=0}(k-1)! (3 \eta\,d_v)^{k-1}    
\end{align}
By coupling~\eqref{eq:Bv_bound_1} and~\eqref{eq:Bv_bound_2} and substituting $\eta=\beta^{\frac{1}{k-1}}\frac{1}{6k^2}$, we obtain:
\begin{align}
\frac{\sum_{v:b_v=0} |\Va(v)|}{\sum_{v} |\Va(v)|}
&\le 
\frac{\sum_{v:b_v=0} (k-1)! (3 \eta\,d_v)^{k-1}}
{\frac{1}{k}\sum_{v:b_v=0}  \frac{(d_v)^{k-1}}{(k-1)^{k-1}}}
&& \text{by ~\eqref{eq:Bv_bound_1} and~\eqref{eq:Bv_bound_2}}
\\&=
\frac{ (k-1)! (3 \eta)^{k-1} \sum_{v:b_v=0} (d_v)^{k-1}}
{\frac{1}{k(k-1)^{k-1}}\sum_{v:b_v=0}  (d_v)^{k-1}}
\\ &< (3\eta)^{k-1} k (k-1)^{2(k-1)} 
\\&= \left(3 \beta^{\frac{1}{k-1}}\frac{1}{6k^2}\right)^{k-1}  k (k-1)^{2(k-1)} 
\\&= \beta \frac{k (k-1)^{2(k-1)} }{2^{k-1}k^{2(k-1)}}
\end{align}
which for all $k \ge 2$ is bounded from above by $\beta$.
The proof is complete.
\end{proof}

\begin{proof}[Proof of Lemma~\ref{lem:approx_order_time}]
First of all we observe that, until return time, \GraphSort\ never needs to compute $\prec$ explicitly. Indeed, $\prec$ is used only to check whether $u \prec v$ for two generic vertices $u,v \in G$. This however boils down to evaluating $(s_u > s_v) \vee ((s_u = s_v) \wedge (u > v))$, which takes time $\scO(1)$. Therefore we only need to keep the values $s_v$ updated in an array; the updates of $\prec$ at lines~\ref{line:init_prec} and~\ref{line:update_prec} are implicit.

Now let us bound the running time.
The initialization of \GraphSort\ is dominated by sorting $V$ in order of degree, which takes time $\scO(n)$ via bucket sort. In the first loop, at each iteration we draw $\scO(h)=\scO(\eta^{-2} \lg n)$ samples, each of which takes time $\scO(1)$ via neighbor queries. Evaluating $x_j \succ v$ takes time $\scO(1)$, and computing $X$ takes time $\scO(h)=\scO(\eta^{-2} \lg n)$. Updating $s_v$ takes time $\scO(1)$. Thus, each iteration of the first loop takes time $\scO(\eta^{-2} \lg n)$.

Consider now the second loop; we claim that each iteration takes time $\scO(\eta^{-2}k^2 \lg k)$. To see this, let us describe the BFS in more detail. We start by pushing $v$ in the queue, and we maintain the invariant that the queue holds only vertices of $G(v)$. To this end, when we pop a generic vertex $u$, we examine every edge $\{u,z\} \in E(G)$, and push $z$ only if $z \succ v$ and $z$ was not pushed before. Note that checking whether $z \succ v$ takes time $\scO(1)$. Now we bound the number of neighbors $z$ of $u$ that are examined. First, this number is obviously at most $d_u$. Recall that $3 \eta d_u \le s_{n+1}(u)$ by construction of the algorithm. Moreover, since $u \in G_n(v)$, then $u \succ_{n+1} v$, which by Observation~\ref{ob:order} implies $s_{n+1}(u) \le s_{n+1}(v) \le d_v \le \frac{k}{\eta}$. Therefore, $d_u \le \frac{k}{3\eta^2}$. Hence, the number of neighbors $z$ of $u$ that are examined is at most $\frac{k}{3\eta^2}$. Since we push at most $k$ vertices before stopping, the total number of vertices/edges examined by each BFS is in $\scO(\eta^{-2} k^2)$. To store the set of pushed vertices we use a dictionary with logarithmic insertion and lookup time. Hence, each BFS will take time $\scO(\eta^{-2} k^2 \lg k)$. Finally, computing $d(v|G(v))$ also takes time $d_v \le \frac{k}{3\eta^2}$. Thus each iteration of the second loop runs in time $\scO(\eta^{-2} k^2 \lg k)$.

As each loops makes at most $n$ iterations, the total running time of \GraphSort$(G, \beta)$ is:
\begin{align}
    \scO(n \lg n) + \scO\big(n \eta^{-2} \lg n\big) + \scO\big(n \eta^{-2} k^2 \lg k\big) = \scO\big(\eta^{-2} k^2 n \lg n\big)
\end{align}
Replacing $\eta = \scO\big(\beta^{\frac{1}{k-1}} k^{-2}\big)$ shows that the running time is in $\scO\big(\beta^{-\frac{2}{k-1}} k^6 n \lg n \big)$, as claimed.
\end{proof}

We can conclude the preprocessing phase of \EpsilonAlgo. We set $\beta=\frac{\epsilon}{2}$, and run $(\prec,\ba)=\GraphSort(G, \beta)$. Then, for all $v$ we let $p(v)=\frac{b_v}{\sum_{u}b_u}$; we also set a few other variables. The running time is dominated by \GraphSort$(G, \beta)$, which by Lemma~\ref{lem:approx_order_time} takes time $\scO\!\left( \epsilon^{-\frac{2}{k-1}} k^6 n \lg n \right)$. This proves the preprocessing time bound of Theorem~\ref{thm:epsuniform} %Again by Lemma~\ref{lem:approx_order}, with high probability we obtain an $(\alpha,\frac{\epsilon}{2})$-DD order $(\prec,\ba)$ for $G$, with $\alpha=\Theta(\epsilon^{\frac{1}{k-1}} k^{-3})$.
and completes the description of the preprocessing phase.

\subsection{The sampling phase: A coupling of algorithms}
Recall that, by Lemma~\ref{lem:approx_order}, with high probability the preprocessing phase yields an $(\alpha,\frac{\epsilon}{2})$-DD order $(\prec,\ba)$ for $G$, with $\alpha=\Theta\big(\epsilon^{\frac{1}{k-1}} k^{-3}\big)$. From now on we assume this holds. Then, by Definition~\ref{def:ab_order}, $\cup_{v:b_v>0}\Va(v)$ contains a fraction $1-\frac{\epsilon}{2}$ of all graphlets. Hence, our goal becomes sampling $\frac{\epsilon}{2}$-uniformly from $\cup_{v:b_v>0}\Va(v)$. By the triangle inequality, this will give an $\epsilon$-uniform distribution over $\vk_k$. To achieve $\frac{\epsilon}{2}$-uniformity over $\cup_{v:b_v>0}\Va(v)$, we modify \UniformAlgo\ step by step.
To begin, we consider what would happen if we sorted $G$ according to $\prec$ and ran the sampling phase of \UniformAlgo\ using the bucket size estimates $\ba$. We show that, by mildly reducing the acceptance probability, we could make the output graphlet distribution uniform over $\cup_{v:b_v>0}\Va(v)$. The resulting algorithm, \ComparisonAlgo$(G, \epsilon)$, is given below.
%Note that the routine \Sample$(\,)$ is the one of \UniformAlgo.
Note that \ComparisonAlgo$(G, \epsilon)$ is just for analysis purposes; we use it as a comparison term, to establish the $\epsilon$-uniformity of our algorithm. 

\begin{algorithm}
\caption{\ComparisonAlgo$(G,\epsilon)$}
\label{alg:uniform_apx}
\begin{algorithmic}[1]
\State let $C_1 \assign$ a large enough universal constant
\State let $\beta \assign \frac{\epsilon}{2}$
\State $(\prec, \ba) \assign$ \GraphSort$(G, \alpha,\beta)$ 
\State let $Z \assign \sum_{v \in V} b_v$, and for each $v \in V$ let $p(v) \assign \frac{b_v}{Z}$ 
\State sort the adjacency lists of $G$ according to $\prec$
\State

\Function{\Sample}{\,}
\While{true}
\State draw $v$ from the distribution $p$ \label{line:comp_sample_1}
\State $S \assign$ \sampleS$(G, v)$ \label{line:comp_sample_2}
\State $p_v(S) \assign$ \computeP$(G, S)$ \label{line:comp_sample_3}
\State with probability $\frac{1}{p(v)\, p_v(S)}\frac{\beta}{Z} k^{-C_1 k}$ \Return $S$ \label{line:comp_sample_4}
\EndWhile
\EndFunction
\end{algorithmic}
\end{algorithm}

\begin{lemma}
\label{lem:comparison_algo}
In \ComparisonAlgo$(G, \epsilon)$, suppose \GraphSort$(G, \alpha,\beta)$ succeeds (Lemma~\ref{lem:approx_order}), and let $p_{\text{acc}}(v,S) = \frac{1}{p(v)\, p(S)} \frac{\beta}{Z} k^{-C_1 k}$ be the expression computed by \Sample$(\,)$ at line~\ref{line:comp_sample_4}. Then $p_{\text{acc}}(v,S) \in [\epsilon^2 \kokm, 1]$, and moreover, the distribution of the graphlets returned by \Sample$(\,)$ is uniform over $\cup_{v:b_v>0}\Va(v)$.
\end{lemma}
\begin{proof}
Rewrite:
\begin{align}
p_{\text{acc}}(v,S) = \frac{1}{p(v)\, p(S)} \frac{\beta}{Z} k^{-C_1 k} 
= \frac{1}{\frac{b_v}{Z} p(S)} \frac{\beta}{Z} k^{-C_1 k} 
= \frac{\beta k^{-C_1 k}}{b_v \, p(S)}\label{eq:pacc_vS}
\end{align}
If \GraphSort$(G, \alpha,\beta)$ succeeds then $(\prec,\ba)$ is an $(\alpha,\beta)$-order with $\alpha=\beta^{\frac{1}{k-1}}\frac{1}{6k^3}$; we will show that, if this is the case, then the last expression in~\eqref{eq:pacc_vS} is in $[\epsilon^2 \kokm, 1]$.
%We will prove that $\frac{1}{p(v)\, p(S)} \frac{\beta}{Z} k^{-C_1 k} \in [\epsilon^2 \kokm, 1]$.
This implies that $p_{\text{acc}}(v,S)$ is a well-defined probability; the uniformity of the returned graphlets then follows immediately from the fact that the sampling routine is the one of \UniformAlgo.

\paragraph{Upper bound.} We bound $b_v \, p(S)$ from below. First, since $v$ was chosen, then $p_v > 0$ and thus $b_v > 0$, in which case by construction \GraphSort$(G, \alpha,\beta)$ sets:
\begin{align}
    b_v = \min\left(d_v, d(v|G(v))\right)^{k-1} \ge d(v|G(v))^{k-1} \label{eq:lb_bv}
\end{align}
Now we adapt the lower bound on $p(S)$ of Lemma~\ref{lem:subgraph_sample} by modifying~\eqref{eq:LB_prob_Si}. By Definition~\ref{def:ab_order}, $b_v > 0$ implies $|\Va(v)|>0$, so the hypotheses of Lemma~\ref{lem:subgraph_sample} are satisfied. Since $G[S]$ is connected, then at least one sequence $v,u_2,\ldots,u_k$ exists such that $\cut_i \ge 1$ for all $i=1,\ldots,k-1$, and $p(S)$ is at least the probability that \sampleS\ follows that sequence. By Lemma~\ref{lem:approx_order}, all $u \succ v$ satisfy $\dg{u}{G(v)} \le \frac{1}{\alpha} \cdot \dg{v}{G(v)}$; this holds for $u=v$  as well, since $\alpha \le 1$. It follows that $\cut_i \le \frac{i}{\alpha}\,\dg{v}{G(v)}$ for all $i=1,\ldots,k-1$. Hence, for all $i=1,\ldots,k-1$,~\eqref{eq:LB_prob_Si} becomes:
\begin{align}
    \Pr(u_{i+1}=u') \ge \frac{\alpha}{i\, \dg{v}{G(v)}}
\end{align}
Thus the probability that the algorithm follows $v,u_2,\ldots,u_k$ is at least
\begin{align}
    p(S) \ge \prod_{i=1}^{k-1} \frac{\alpha}{i\, \dg{v}{G(v)}} = \frac{\alpha^{k-1}}{(k-1)! \, \dg{v}{G(v)}^{k-1}} \label{eq:lb_pS}
\end{align}
Combining~\eqref{eq:lb_bv} and~\eqref{eq:lb_pS}, we conclude that for some absolute constant $C_2$:
\begin{align}
    b_v \cdot p(S) \ge \dg{v}{G(v)}^{k-1} \cdot \frac{\alpha^{k-1}}{(k-1)! \, \dg{v}{G(v)}^{k-1}} \ge \frac{\alpha^{k-1}}{(k-1)!}
\end{align}
and therefore
\begin{align}
\frac{\beta k^{-C_1 k}}{b_v \, p(S)}
\le \frac{\beta k^{-C_1 k}}{ \frac{\alpha^{k-1}}{(k-1)!} } 
\le \frac{\beta k^{-(C_1-C_2)k}}{\alpha^{k-1}}
\end{align}
Since $\alpha=\beta^{\frac{1}{k-1}}\frac{1}{6k^3}$, we have:
\begin{align}
\frac{\beta k^{-(C_1-C_2)k}}{\alpha^{k-1}}
= \frac{\beta k^{-(C_1-C_2)k} (6k^3)^{k-1}}{\beta}
\le k^{-(C_1-C_3)k}
\end{align}
for some constant $C_3$. Choosing $C_1 \ge C_3$, the acceptance probability is in $[0,1]$.

\paragraph{Lower bound.} We bound $b_v \, p(S)$ from above. On the one hand, note that the upper bound on $p(S)$ of Lemma~\ref{lem:subgraph_sample} applies even for an $(\alpha,\beta)$-order. Indeed, that bound is based on the lower bound of Lemma~\ref{lem:cuti} whose proof uses only $d(v|G(v))$ but not $d(u|G(v))$ for any $u \succ v$. Thus,
\begin{align}
    p(S) \le \frac{k^{C_4 k}}{\dg{v}{G(v)}^{k-1}}
\end{align}
for some constant $C_4$. On the other hand, $b_v \le d_v^{k-1}$ by construction of \GraphSort. Moreover, since $b_v>0$, by Definition~\ref{def:ab_order} we have $d_v \le \frac{1}{\alpha} \dg{v}{G(v)}$. Therefore,
\begin{align}
    b_v \le \frac{1}{\alpha^{k-1}} \dg{v}{G(v)}^{k-1}
\end{align}
We conclude that:
\begin{align}
    b_v \cdot p(S) \le \frac{1}{\alpha^{k-1}} \dg{v}{G(v)}^{k-1} \frac{k^{C_4 k}}{\dg{v}{G(v)}^{k-1}}
    = \frac{k^{C_4 k}}{\alpha^{k-1}}
\end{align}
Since $\alpha=\beta^{\frac{1}{k-1}}\frac{1}{6k^3}$, we have:
\begin{align}
    b_v \cdot p(S) \le \frac{k^{C_4 k} (6k^3)^{k-1}}{\beta} \le \frac{k^{C_5 k}}{\beta}
\end{align}
for some constant $C_5$. Hence,
\begin{align}
    \frac{\beta k^{-C_1 k}}{b_v \, p(S)} \ge \beta k^{-C_1 k} \frac{\beta}{k^{C_5 k}} = \beta^2 c^{-(C_5+C_1) k}
\end{align}
Replacing $\beta=\frac{\epsilon}{2}$ shows that the acceptance probability is at least $\epsilon^2\kokm$, as claimed.
\end{proof}

\ComparisonAlgo\ is now our baseline. Our goal is building an algorithm whose output distribution is $\frac{\epsilon}{2}$-close to that of \ComparisonAlgo, \emph{without using the sorted adjacency lists}. To this end we will carefully re-design the routines of \UniformAlgo\ and use several coupling arguments. In what follows we assume that $V(G)$ is sorted by $\prec$, and we fix some $v \in V(G)$ with $b_v > 0$.

\subsubsection{Approximating the cuts}
First, we show how to estimate efficiently the size of the cuts encountered by the random growing process. We will use these estimates to approximate the random growing process itself, as well as the computation of $p_v(g)$. The quality of our estimates and the cost of computing them are both based on the properties of $(\alpha,\beta)$-DD orders.
\begin{algorithm}[h]
\caption{\CutEstim$(G,v,U,\alpha,\beta,\delta)$}
\begin{algorithmic}[1]
\State let $\ell \assign \frac{1}{k \delta \alpha^2}$,~~$h \assign \Theta\big(\ell^2 \lg{\frac{k}{\beta}}\big)$
\State let $\hat{\cut}(U) = 0$
\For{each $u \in U$}
\State sample $h$ neighbors $x_1,\ldots,x_h$ of $u$ i.i.d.\ u.a.r.
\State let $X \assign \sum_{j=1}^h \Ind{x_j \succ v \,\wedge\, x_j \notin U}$
\IfThenElse{$X \ge \ell$}{let $\hat{\cut}(u) \assign \frac{d_u}{h}X$ }{let $\hat{\cut}(u) \assign 0$} \label{line:assign}
\EndFor
\State \Return $\{\hat{\cut}(u)\}_{u \in S}$
\end{algorithmic}
\end{algorithm}

\begin{lemma}
\label{lem:cutestim_time}
\CutEstim$(G,v,U,\alpha,\beta)$ runs in time $\scO\left(|U|^2 \frac{1}{k \delta^2\alpha^4} \lg\frac{1}{\beta}\right)$.
\end{lemma}
\begin{proof}
%The algorithm makes $|U|$ iterations.
At each iteration \CutEstim$(G,v,U,\alpha,\beta)$ draws $h=\scO\left(\frac{1}{k^2\delta^2\alpha^4}\right) \lg \frac{k}{\beta}$ samples, which is in $\scO\left(\frac{1}{k \delta^2\alpha^4} \lg\frac{1}{\beta}\right)$ as $\lg \frac{k}{\beta}=\scO\left(k\lg\frac{1}{\beta}\right)$. For each sample, computing $\Ind{x_j \succ v \,\wedge\, x_j \notin U}$ takes time $\scO(|U|)$ via edge queries. Summing over all iterations gives a bound of $\scO(|U|) \cdot \scO\left(\frac{1}{k \delta^2\alpha^4} \lg\frac{1}{\beta}\right) \cdot \scO(|U|) = \scO\left(|U|^2\frac{1}{k \delta^2\alpha^4} \lg\frac{1}{\beta}\right)$.
%Thus, the algorithm runs in time $|U|\poly(k) \scO(\frac{1}{\alpha^4\beta^2} \lg\nicefrac{1}{\beta})$.
\end{proof}

\begin{lemma}
\label{lem:cutestim}
%Suppose that, after running \GraphSort$(G, \epsilon_0)$, the high-probability claims of Lemma~\ref{lem:approx_order} hold.
%Let $V(G)$ be sorted by an $(\alpha,\beta)$-DD order $(\prec,\ba)$.
Let $G[U]$ be a connected subgraph of $G(v)$ on $i < k$ vertices containing $v$.
%Let $U \subseteq V(G[v])$ be such that $v \in U$, $|U|=i\le k-1$, and $G[U]$ is connected.
With probability $1-\poly \frac{\beta}{k}$, the output of \CutEstim$(G,v,U,\alpha,\beta,\delta)$ satisfies:
\begin{align}
&\left|\hat{\cut}(u) - \cut(u)\right| \le \delta\, d(v|G(v))
\;\; \forall u \in U
%\\ &|\hat{\cut}(U)-\cut(U)| \le \scO\left(\frac{\beta^2\alpha^2}{k^2} \cut(U)\right)
\end{align}
where $\cut(u) = |\Cut(u, G(v) \setminus U)|$. % and $\cut(U)=\sum_{u \in U}\cut(u)=|\Cut(U, G(v)\setminus U)|$.
In this case, then $|\hat{\cut}(U) - \cut(U)| \le |U| \delta k \cut(U)$ too, where $\hat{\cut}(U)=\sum_{u \in U}\hat{\cut}(u)$ and $\cut(U)=\sum_{u \in U}\cut(u)=|\Cut(U, G(v)\setminus U)|$.
\end{lemma}
\begin{proof}
Fix any $u \in U$, and consider the iteration where the edges of $u$ are sampled. For each $j\in[h]$ let $X_{u,j}=\Ind{x_j \succ v \,\wedge\, x_j \notin U}$. Clearly, $\E[X_{u,j}]=\frac{h \cut(u)}{d_u}$. Let $X_u=\sum_{j=1}^h X_{u,j}$; this is the value of $X$ tested by the algorithm at $u$'s round. To begin, we note that:
\begin{align}
    \E\left[\frac{d_u}{h}X_u\right] = \frac{d_u}{h}\E\left[\sum_{j=1}^h X_{u,j}\right] = \cut(u)
\end{align}
%\paragraph*{Bounds for $\hat{c}(u)$.}
Define $\gamma(u)=\frac{d_u}{h}X_u$. Clearly,
\begin{align}
\Pr\Big(|\gamma(u) - \cut(u)| > \delta d(v|G(v)) \Big) &= \Pr\left(|X_u - \E X_u| > h\, \delta\, \frac{d(v|G(v))}{d_u}\right)
\end{align}
Note that the algorithm can set $\hat{\cut}(u)=\gamma(u)$ or $\hat{\cut}(u)=0$. First, we show that $\gamma(u)$ is concentrated around $\cut(u)$. Then, we deal with the value of $\hat{\cut}(u)$ set by the algorithm.

Clearly $X$ is the sum of $h$ i.i.d.\ indicator random variables.
By Hoeffding's inequality, for any $\delta>0$, we have $\Pr(|X_u- \E X_u| > t) < 2e^{-2\frac{t^2}{h}}$.
With $t=\frac{\delta\, d(v|G(v))\, h}{d_u}$, we obtain:
\begin{align}
\Pr\left(|X_u - \E X_u| >  h\, \delta\, \frac{d(v|G(v))}{d_u}\right)
&< 2 \exp\left(-2 h\, \delta^2 \left(\frac{d(v|G(v))}{d_u}\right)^2\right)
\end{align}

Now, as $b_v > 0$ and $u \succ v$, by Definition~\ref{def:ab_order} we have $d(v|G(v)) \ge \alpha d_v$ and $d_v \ge 3 k \alpha d_u$.
Hence, $d_u < \frac{1}{3k\alpha^2} d(v|G(v))$, so $\frac{d(v|G(v))}{d_u} > 3k\alpha^2$.
%Four our purposes here it is enough to use $d_u \le \alpha^{-2} d(v|G(v))$, which yields:
Therefore:
\begin{align}
    h\, \delta^2 \left(\frac{d(v|G(v))}{d_u}\right)^2 \ge 9 h\, \delta^2 k^2 \alpha^4
%    \\&= 2 \exp\left(-\Theta\left(k^2 \lg \frac{k}{\beta} \right)\right)
%    = 2 \exp\left(-2 \delta^2 \alpha^4 h\right)
    %= 2 \exp\left(-2 \delta^2 \alpha^4 \Theta\left(\frac{1}{\delta^2\alpha^4} \lg\frac{k}{\beta}\right) \right)
    %= 2 \exp\left(-\Theta\left( \lg\frac{k}{\beta}) \right)\right)
%    \\&=2 \exp\left(- \frac{\beta^2 h \eta^4}{8 k^4}\right)
\end{align}
However, note that $h=\Theta\left(\frac{1}{\delta^2k^2\alpha^4} \lg \frac{k}{\beta}\right)$.
Thus, $\Pr\big(|\gamma(u) - \cut(u)| > \delta\, d(v|G(v)) \big) \le \poly\frac{\beta}{k}$.

Now, the algorithm fails if it either sets $\hat{\cut}(u)=\gamma(u)$ and $|\gamma(u) - \cut(u)| > \delta\, d(v|G(v))$, or if it sets $\hat{\cut}(u)=0$ and $|0 - \cut(u)| > \delta\, d(v|G(v))$.
The probability of the first event is at most the probability that $|\gamma(u) - \cut(u)| > \delta\, d(v|G(v))$, which is $\poly\frac{\beta}{k}$ as shown above. So we must bound the probability that $\hat{\cut}(u)=0$ and $|0 - \cut(u)| > \delta d(v|G(v))$; this second condition is just $\cut(u) > \delta\, d(v|G(v))$. Recalling that $\E X_u = h \frac{\cut(u)}{d_u}$ and that $\frac{d(v|G(v))}{d_u} \ge 3k\alpha^2$ and $\delta=\frac{\beta}{4k^2}$, we obtain:
\begin{align}
\E X_u = h \frac{\cut(u)}{d_u} > h \frac{\delta\, d(v|G(v))}{d_u}  \ge 3k h \alpha^2 \delta %= h \frac{\alpha^2 \beta}{4k^2} = \frac{h}{4\ell}
\end{align}
Note that $h \alpha^2 \delta = \frac{h}{\ell} \in \Omega\left(\sqrt{h\lg \frac{k}{\beta}}\right)$.
Also note that $\ell = o\big(\frac{h}{\ell}\big)$.
This implies:
\begin{align}
    \E X_u - \ell
    %&\ge \frac{h}{4\ell} - \ell
    = \Omega\left(\sqrt{h\lg \frac{k}{\beta}}\right)
\end{align}
Now, $\hat{\cut}(u)=0$ by construction of the algorithm implies $X_u < \ell$, which can be rewritten as $X_u < \E X_u -t$ with $t = \E X_u - \ell = \Omega\left(\sqrt{h\lg \frac{k}{\beta}}\right)$.
Since $X_u$ is the sum of $h$ i.i.d.\ indicator random variables, Hoeffding's inequality gives:
\begin{align}
    \Pr\left(X_u - \E X_u < t \right) < e^{-\frac{2t^2}{h}} &= 
    %\exp\!\left(\!-\frac{\E X_u - \ell}{h}\right)
     \exp\left(-\frac{\Omega\big(\sqrt{h\lg \frac{k}{\beta}}\big)^2}{h}\right)
    = \exp\left(-\Omega\left(\frac{k}{\beta}\right)\right)
\end{align}
Hence this event has probability at most $\poly \frac{\beta}{k}$, too. We conclude that $\left|\hat{\cut}(u) - \cut(u)\right| \le \delta\, d(v|G(v))$ with probability at least $1-\poly \frac{\beta}{k}$. By a union bound over $u \in S$, this proves the claim for the $\hat{\cut}(u)$.

For $\hat{\cut}(U)$, note that:
\begin{align}
    |\hat{\cut}(U) - \cut(U)| &= \left|\sum_{u \in U}\hat{\cut}(u) - \sum_{u \in U}\cut(u)\right|
     \le \sum_{u \in U}\left|\hat{\cut}(u) - \cut(u)\right|
    \le |U| \,\delta\, d(v|G(v))
\end{align}
By Lemma~\ref{lem:cuti} $d(v|G(v)) \le k \cut(U)$, concluding the proof.
\end{proof}

\subsubsection{Approximating the random growing process}
Using \CutEstim\ we now run an approximate random growing process as follows. Start with $S_1=\{v\}$, and at each step $i=1,\ldots,k-1$, run \CutEstim\ with $U=S_i$. This gives estimates of $|\Cut(u, G(v) \setminus S_i)|$ for all $u \in S_i$. Using these estimates, sample a random edge near-uniformly from $\Cut(S_i, G(v) \setminus S_i)$. The result is the following routine whose output distribution is close to \sampleS.

\begin{algorithm}[h]
\caption{\EpsilonSampleS$(G,v,\alpha,\beta,\gamma)$}
\begin{algorithmic}[1]
\State $S_1 \assign \{v\}$
\For{$i = 1,\ldots,k-1$} \label{line:eps_sample:for_i}
\State $(\hat{\cut_i}(u))_{u \in S_i}$ = \CutEstim$\big(G,v,S_i,\alpha,\beta,\scO\big(\gamma k^{-4}\big)\big)$
\State $\hat{\cut}_i = \sum_{u \in S_i} \hat{\cut_i}(u)$
\State draw $u$ with probability $\frac{\hat{\cut_i}(u)}{\hat{\cut}_i}$ (if all $\hat{\cut_i}(u)=0$ then \textsc{fail}) \label{line:u_apx}
\Repeat \label{line:trial_start}
\State draw $u'$ u.a.r.\ from the adjacency list of $u$
\label{line:u1_apx}
\Until{$u' \in G(v) \setminus S_i$}\label{line:trial_end} %\Comment{check that $u' \succ v \wedge u' \notin S_{i}$} \label{line:trial_end}
\State let $S_{i+1} \assign S_{i} \cup \{u'\}$ 
\EndFor
\State \Return $S_k$
\end{algorithmic}
\end{algorithm}

\begin{lemma}%[See Appendix~\ref{apx:epsSample}]
\label{lem:apx_sample}
Let $p$ be the output distribution of \sampleS$(G, v)$ and $q$ the output distribution of \EpsilonSampleS$(G,v,\alpha,\beta,\gamma)$. Then, $\tvd{p}{q}\le \gamma +\poly\frac{\beta}{k}$. %, where the degree of $\poly \frac{\beta}{k}$ can be made arbitrarily large by adjusting the constants.
\end{lemma}
\begin{proof}
We establish a coupling between the two algorithms. For $i=1$, both algorithms set $S_1=\{v\}$. Now suppose that both algorithms agree on $S_1,\ldots,S_i$ and they are about to choose $S_{i+1}$. We show that with probability $1-\frac{\gamma}{Ck}$ they agree on the next edge drawn, and therefore on $S_{i+1}$. Here, $C$ is a constant that we can make arbitrarily large by appropriately choosing the constants used along the algorithm and in \CutEstim.

For $u \in S_i$ let $\cut_i(u) = |\Cut(u, G(v) \setminus S_i)|$, and let $\cut_i = \sum_{u \in S_i} \cut_i(u)$.
For each $u \in S_i$, let $p_i(u)= \frac{\cut_i(u)}{\cut_i}$ and $q_i(u)=\frac{\hat{\cut_i}(u)}{\hat{\cut}(S_i)}$. So $p_i(u)$ is the probability that \sampleS\ draws $u$ at line~\ref{line:u}, and $q_i(u)$ the probability that \EpsilonSampleS\ draws $u$ at line~\ref{line:u_apx}.

Now, if \EpsilonSampleS\ and \sampleS\ both choose $u$, then we can couple them so that they choose the same edge. This holds since both algorithms draw $u'$ uniformly from all neighbors of $u$ in $G(v) \setminus S_i$. So the probability that the two algorithms choose a different edge is at most the probability that they choose $u$ differently, that is, by $\tvd{q_i}{p_i} = \frac{1}{2}\|q_i-p_i\|_1$. Therefore:
\begin{align}
    \tvd{q_i}{p_i} &= \frac{1}{2} \sum_{u \in S_i} |q_i(u) - p_i(u)|
    = \sum_{u \in T} (q_i(u) - p_i(u)) 
    = \sum_{u \in T} \left(\frac{\hat{\cut_i}(u)}{\hat{\cut}(S_i)}  - \frac{{\cut_i(u)}}{\cut_i}\right)
\end{align}
where $T = \{u \in S_i \,:\, q_i(u) > p_i(u)\}$.
Now, by Lemma~\ref{lem:cutestim}, with probability $1-\poly\frac{\gamma}{k}$, we have $|\hat{\cut}(u_i)-\cut(u_i)| \le \delta \, d(c|G(v))$ for all $u \in S_i$, and $|\hat{\cut}(S_i) - \cut_i| \le k^2 \delta \cut_i$.
In this case,
\begin{align}
    \frac{\hat{\cut_i}(u)}{\hat{\cut}(S_i)}  - \frac{{\cut_i(u)}}{\cut_i} 
&\le \frac{\cut_i(u) + \delta\,d(c|G(v))}{\cut_i(1-k^2 \delta)}  - \frac{{\cut_i(u)}}{\cut_i} 
\\&= \frac{\cut_i(u) + \delta\,d(c|G(v)) - \cut_i(u)(1-k^2 \delta)}{\cut_i(1-k^2 \delta)} 
\\&= \frac{\delta\,d(c|G(v)) +k^2 \delta \cut_i(u)}{\cut_i(1-k^2 \delta)} 
\end{align}
Clearly, $\cut_i(u) \le \cut_i$, and by Lemma~\ref{lem:cuti}, $d(c|G(v)) \le k \cut_i$.
Therefore:
\begin{align}
\frac{\delta\,d(c|G(v)) +k^2 \delta \cut_i(u)}{\cut_i(1-k^2 \delta)}
& \le\frac{\delta\,k\, \cut_i + k^2 \delta \cut_i}{\cut_i(1-k^2 \delta)} 
\\&= \frac{\delta\,k\, + k^2 \delta}{1-k^2 \delta} 
\\&\le k \cdot \frac{2 \delta\,k^2}{1-k^2 \delta}
\\&= \frac{2 \delta\,k^3}{1-k^2 \delta}
\end{align}
For any $\delta = \scO\big(\frac{\gamma}{k^4}\big)$, this is in $\scO\big(\frac{\gamma}{k^2}\big)$. Taking the sum over $u \in T$, we obtain $\tvd{q_i}{p_i} = \scO\big(\frac{\gamma}{k}\big)$.

Thus, the two algorithms will disagree on $S_{i+1}$ with probability at most $\poly\frac{\beta}{k} + \scO\big(\frac{\gamma}{k}\big)$. By a union bound on all $i$, the algorithms disagree on $S_k$ with probability $\poly\frac{\beta}{k} + \scO(\gamma)$. The $\scO(\gamma)$ part can made smaller than $\gamma$ by choosing $\delta = \scO\big(\frac{\gamma}{k^4}\big)$ small enough.
\end{proof}

\begin{lemma}
\label{lem:apx_sample_cost}
\EpsilonSampleS$(G,v,\alpha,\beta,\gamma)$ has expected running time $\scO\left(\frac{k^9}{\gamma^2 \alpha^4} \lg \frac{1}{\beta}\right)$.
\end{lemma}
\begin{proof}
At each iteration, since $|S_i|<k$, by Lemma~\ref{lem:cutestim_time}, obtaining the cut estimates takes time $\scO\left(\frac{k}{\delta^2\alpha^4} \lg \frac{1}{\beta}\right)$. As $\delta = \scO\big(\frac{\gamma}{k^4}\big)$, this gives a bound of $\scO\left(\frac{k k^8}{\gamma^2\alpha^4} \lg \frac{1}{\beta} \right) = \scO\left(\frac{k^9}{\gamma^2\alpha^4} \lg \frac{1}{\beta}\right)$. We show that this dominates the expected time of the trials at lines~\ref{line:trial_start}--\ref{line:trial_end} as well.

Let $T$ be the random variable giving the number of times lines~\ref{line:trial_start}--\ref{line:trial_end} are executed. Let $\ev_u$ be the event that $u \in S_i$ is chosen at line~\ref{line:u_apx}. Clearly $\ev_u$ implies $\hat{\cut_i}(u) > 0$. Thus, $\Pr(\ev_u) \le \Pr\big(\hat{\cut_i}(u) > 0\big)$. Moreover, conditioned on $\ev_u$, the algorithm returns after $\frac{d_u}{\cut_i(u)}$ trials in expectation. Therefore:
\begin{align}
\E T &= \sum_{u \in S_i} \Pr(\ev_u) \, \E[T \,|\, \ev_u] \le \sum_{u \in S_i} \Pr\big(\hat{\cut_i}(u) > 0\big) \frac{d_u}{\cut_i(u)}
\end{align}
Now recall \CutEstim. By construction $\hat{\cut_i}(u) > 0$ implies $X_u \ge \ell$, where $\E X_u = \frac{\cut_i(u)}{d_u} h$. By Markov's inequality:
\begin{align}
\Pr\big(\hat{\cut_i}(u) > 0\big) = \Pr(X \ge \ell) \le \frac{\E X}{\ell} = \frac{\cut_i(u)}{d_u}\frac{h}{\ell} = \frac{\cut_i(u)}{d_u} \ell \lg \frac{k}{\beta}
\end{align} 
Therefore:
\begin{align}
\E T &\le \sum_{u \in S_i} \frac{\cut_i(u)}{d_u} \ell \lg \frac{k}{\beta} \cdot \frac{d_u}{\cut_i(u)} \le k \ell \lg \frac{k}{\beta} = \frac{1}{\delta \alpha^2 }\lg \frac{k}{\beta} = \frac{k}{\delta \alpha^2 }\lg \frac{1}{\beta}
\end{align}
Finally, note that each single trial takes time $\scO(k)$ via edge queries. The resulting time bound is in $\scO\big(\frac{k^2}{\delta \alpha^2 }\lg \frac{1}{\beta}\big)$, which is dominated by the sampling running time, see above.
\end{proof}

\subsubsection{Approximating the acceptance probability}
Next, we compute an acceptance probability. For any $g \in \Va(v)$ let $q_v(g)$ be the probability that \EpsilonSampleS\ returns $g$. If we could compute $q_v(g)$, we would be done. However, computing $q_v(g)$ requires computing the exact sizes of the cuts, which takes time $\Omega(\dmax)$ in the worst case. Fortunately, we can show that a good approximation of $p_v(g)$, the probability that \sampleS\ returns $g$, is sufficient. By the properties of $(\alpha,\beta)$-DD orders, we can compute such an approximation efficiently.

\begin{algorithm}
\caption{\epsProbCompute$(G, S=\{v,u_2,\ldots,u_k\}, \alpha, \beta, \rho)$}
\label{algo:apx_compute_p}
\begin{algorithmic}[1]
\State $\hat{p} \assign 0$
\For{each permutation $\sigma=(\sigma_2,\ldots,\sigma_{k})$ of $u_2,\ldots,u_k$} \label{line:for_sigma_apx}
\State $\hat{p}_{\sigma} \assign 1$
\For{each $i = 1,\ldots,k-1$}
\State $S_i \assign \{v,\sigma_2,\ldots,\sigma_i\}$
\State $n_{i} \assign |\Cut(S_i,\sigma_{i+1})|$
\State $(\hat{\cut_i}(u))_{u \in S_i}$ = \CutEstim$\big(G,v,S_i,\alpha,\frac{\beta}{\kok}, \scO\big(\frac{\rho}{k^2}\big) \big)$
\State $\hat{\cut}_i = \sum_{u \in S_i} \hat{\cut_i}(u)$
\State $\hat{p}_{\sigma} \assign \hat{p}_{\sigma} \cdot \frac{n_i}{\hat{\cut_i}}$
\EndFor
\State $\hat{p} \assign \hat{p} + \hat{p}_{\sigma}$
\EndFor
\State \Return $\hat{p}$
\end{algorithmic}
\end{algorithm}

\begin{lemma}
\label{lem:eps_prob}
For any $g = G[S] \in \Va(v)$ and $\rho > 0$, \epsProbCompute$(G, S, \alpha, \beta, \rho)$ runs in time $\kok \frac{1}{\rho^2 \alpha^4} \lg \frac{1}{\beta}$, and with probability $1 - \poly \frac{\beta}{k}$ returns a multiplicative $(1 \pm \rho)$-approx\-i\-mati\-on $\hat{p}_v(g)$ of $p_v(g)$. % where the degree of $\poly \frac{\beta}{k}$ can be made arbitrarily large by adjusting the constants.
\end{lemma}
\begin{proof}[Proof sketch.]
The running time analysis is straightforward. For the correctness, let $\Sigma$ be the set of all permutations $\sigma=(\sigma_1,\ldots,\sigma_{k})$ of $v,u_2,\ldots,u_k$ such that $\sigma_1=v$. For each $\sigma \in \Sigma$, let $S^{\sigma}_i$ be the first $i$ nodes in $S$ as given by $\sigma$. Note that \epsProbCompute\ returns:
\begin{align}
    \hat{p} = \sum_{\sigma \in \Sigma} \hat{p}_{\sigma} = \sum_{\sigma \in \Sigma} \prod_{i=1}^{k-1} \frac{n(S^{\sigma}_i)}{\hat{\cut}(S^{\sigma}_i)}
\end{align}
where $n(S^{\sigma}_i)$ is the size of the cut between $S^{\sigma}_i$ and $\sigma_{i+1}$, and $\hat{\cut}(S^{\sigma}_i)$ is the value of $\hat{\cut_i}$ used by the \epsProbCompute.
Instead, \computeP($G, S$) returns:
\begin{align}
    p = \sum_{\sigma \in \Sigma} p_{\sigma} = \sum_{\sigma \in \Sigma} \prod_{i=1}^{k-1} \frac{n(S^{\sigma}_i)}{\cut(S^{\sigma}_i)}
\end{align}
Therefore,
\begin{align}
    \frac{\hat{p}}{p} =  \frac{\sum_{\sigma \in \Sigma} \prod_{i=1}^{k-1} \cut(S^{\sigma}_i)}{\sum_{\sigma \in \Sigma} \prod_{i=1}^{k-1} \hat{\cut}(S^{\sigma}_i)}
\end{align}
Look at a single term $\sigma$.
Note that $\hat{\cut}(S^{\sigma}_i)$ is estimated as in \CutEstim$(G,v,U,\alpha,\beta,\delta)$, but with $k$ times as many samples.
Therefore, the guarantees of Lemma~\ref{lem:cutestim} apply, but the deviation probability shrinks by a factor $2^{-k}$.
Since there are at most $2^k$ different subsets $S^{\sigma}_i$, by a union bound, with probability $1-\poly\frac{\beta}{k}$ we have $\hat{\cut}(S^{\sigma}_i) \in \cut(S^{\sigma}_i) \cdot (1 \pm \delta\,k^2)\cut(S^{\sigma}_i)$ for all $S^{\sigma}_i$ simultaneously, where we used $|S|\le k$.
Thus,
\begin{align}
    \frac{\prod_{i=1}^{k-1}\cut(S^{\sigma}_i)}{\prod_{i=1}^{k-1}\hat{\cut}(S^{\sigma}_i)}
    = \prod_{i=1}^{k-1}\frac{\cut(S^{\sigma}_i)}{\hat{\cut}(S^{\sigma}_i)}
    \in \left(\frac{1}{(1 \pm \delta\,k^2)}\right)^{k-1}
\end{align}
For $\delta=\scO\left(\frac{\rho}{k^3}\right)$ small enough, the right-hand side is in $(1 \pm \rho)$. This gives $\frac{\hat{p}}{p} \in (1 \pm \rho)$, as claimed.
\end{proof}

\subsubsection{Coupling the algorithms}
We conclude the sampling phase of \EpsilonAlgo. After drawing $v$, we invoke \EpsilonSampleS$(G, v, \alpha, \beta, \gamma)$ and \epsProbCompute$(G, S, \alpha, \beta, \rho)$ with $\gamma=\rho=\epsilon^3 k^{-C_2 k}$, where $S$ is the set of vertices returned by \EpsilonSampleS. Hence we have a random graphlet $g=G[S]$ together with a probability estimate $\hat{p}_v(g)$. We then accept $g$ with probability inversely proportional to $\hat{p}_v(g)$. For reference, see the code below.

\begin{algorithm}[H]
\caption{\EpsilonAlgo$(G, \epsilon)$}
\label{alg:apx_uniform}
\begin{algorithmic}[1]
\State let $C_1, C_2 \assign$ large enough universal constants
\State let $\beta \assign \frac{\epsilon}{2}$ and let $(\prec, \ba) \assign$ \GraphSort$(G, \alpha,\beta)$ %and Lemma~\ref{lem:approx_order})
\State let $Z \assign \sum_{v \in V} b_v$, and for each $v \in V$ let $p(v) \assign \frac{b_v}{Z}$ 
\State sort $V(G)$ according to $\prec$
\State
\Function{\epsSampler}{\,}
\While{true}
\State draw $v$ from the distribution $p$ \label{line:sample_1}
\State $S \assign$ \EpsilonSampleS$(G, v, \alpha, \beta, \epsilon^3 k^{-C_2 k})$ \label{line:sample_2}
\State $\hat{p}(S) \assign$ \epsProbCompute$(G, S, \alpha, \beta, \epsilon^3 k^{-C_2 k})$ \label{line:sample_3}
\State with probability $\min\!\left(1,\frac{1}{p(v)\, \hat{p}(S)}\frac{\beta}{Z} k^{-C_1 k}\right)$ \Return $S$ \label{line:sample_4}
\EndWhile
\EndFunction
\end{algorithmic}
\end{algorithm}

% This is the same probability of Lemma~\ref{lem:1}, except that we must threshold it at $1$, since $\hat{p}_v(g)$ could be small.

%We argue that the distribution of \emph{accepted} graphlets is $\frac{\epsilon}{2}$-close to that given by the sampling phase of \ComparisonAlgo. First, by Lemma~\ref{lem:apx_sample}, $g$ comes from a distribution $\kokm \epsilon^3$-close to $p_v$. Thus, by a coupling argument, we can assume that $g$ coincides with the graphlet sampled by \sampleS\ with probability $1-\kokm \epsilon^3$. Furthermore, by Lemma~\ref{lem:eps_prob}, with probability $1 - \poly \frac{\beta}{k}$ the estimate $\hat{p}_v(g)$ is a multiplicative $(1 \pm \kokm \epsilon^{3})$-approximation of $p_v(g)$. Thus, by another coupling argument, we can assume that \EpsilonAlgo\ and \ComparisonAlgo\ take the same decision (accept/reject $g$) with probability $1-\kokm \epsilon^{3}$. Moreover, if $\hat{p}_v(g)$ satisfies the guarantees, then we also know that the acceptance probability is at least $\kokm \epsilon^2$. By taking the probability that the two algorithms disagree on the sampled graphlets \emph{conditional on acceptance}, we obtain:
%\begin{align}
%    \Pr(\text{the two algorithms' output differ}) \le \frac{k^{-C_2 k} \epsilon^3}{k^{-C_1 k} \epsilon^2}
%\end{align}
%which we can make smaller than $\frac{\epsilon}{2}$ by adjusting $C_1$ and $C_2$.

%Clearly, this argument above is not very formal. The next two lemmas give a formal argument of the

The next two lemmas show that \EpsilonAlgo\ satisfies the claims of Theorem~\ref{thm:epsuniform}.

\begin{lemma}
\label{lem:wrap_sample_dist}
Suppose that the preprocessing of \EpsilonAlgo$(G,\epsilon)$ succeeds (Lemma~\ref{lem:approx_order}). Then, each invocation of \epsSampler$(\,)$ returns a graphlet independently and $\epsilon$-uniformly at random from $G$.
\end{lemma}
\begin{proof}
By Lemma~\ref{lem:approx_order}, \GraphSort$(G, \alpha,\beta)$ with high probability returns an $(\alpha,\beta)$-DD order for $G$.
The rest of the proof is conditioned on this event. We will couple the sampling phases of \EpsilonAlgo$(G, \epsilon)$ and \ComparisonAlgo$(G, \epsilon)$.
Note that the the preprocessing phases of the two algorithms are identical (save for the fact that \ComparisonAlgo\ also sorts the adjacency lists).
In particular, they use the same order $\prec$ over $V(G)$, which induces the same bucketing $\{\Va(v)\}_{v \in V}$, as well as the same bucket size estimates $\{b_v\}_{v \in V}$, and therefore also the same distribution $p$ over $V$.

Let $H=\{v \in V \,:\,b_v > 0\}$. Let $\scU_V$ and $\scU_H$ be the uniform distributions respectively over $\cup_{v \in V}\Va(v)$ and $\cup_{v \in H} \Va(v)$, and let $q$ be the output distribution of \EpsilonAlgo::\epsSampler. Our claim is that $\tvd{q}{\scU_V} \le \epsilon$. By the triangle inequality, $\tvd{q}{\scU_V} \le \tvd{\scU_H}{\scU_V} + \tvd{q}{\scU_{H}}$, and by Definition~\ref{def:ab_order}, the buckets indexed by $H$ hold a fraction at least $1-\beta=1-\frac{\epsilon}{2}$ of all graphlets, hence $\tvd{\scU_H}{\scU_V} \le \frac{\epsilon}{2}$. Therefore, to prove that $\tvd{q}{\scU_V} \le \epsilon$ we need only to prove that $\tvd{q}{\scU_{H}} \le \frac{\epsilon}{2}$, which we do in the remainder.

First, by Lemma~\ref{lem:comparison_algo}, $\scU_H$ is precisely the output distribution of \ComparisonAlgo::\Sample. Thus, we will couple \ComparisonAlgo::\Sample\ and \EpsilonAlgo::\epsSampler; under this coupling they will return the same graphlet with probability at least $1-\frac{\epsilon}{2}$, establishing that $\tvd{q}{\scU_{H}} \le \frac{\epsilon}{2}$.

To begin, since the \ComparisonAlgo::\Sample\ and \EpsilonAlgo::\epsSampler\ use the same distribution $p$ over the buckets, we can couple them so that they choose the same bucket $\Va(v)$. Now let $S_P$ denote the random set of nodes drawn by \ComparisonAlgo::\Sample\ at line~\ref{line:comp_sample_2}, and by $S_Q$ the one drawn by \EpsilonAlgo::\epsSampler\ at line~\ref{line:sample_2}. As the two algorithms invoke respectively \sampleS$(G, v)$ and \EpsilonSampleS$(G, v, \alpha, \frac{\epsilon}{2}, \epsilon^3 k^{-C_2 k})$, Lemma~\ref{lem:apx_sample} yields:
\begin{align}
    \tvd{S_P}{S_Q}\le \epsilon^3 k^{-C_2 k} + \poly\frac{\epsilon}{k}
\end{align}
Hence, we can couple the two algorithms so that $\Pr(S_Q \ne S_P) \le \epsilon^3 k^{-C_2 k}$.

Now let $X_P$ be the indicator random variable of the event that \ComparisonAlgo::\Sample\ accepts $S_P$ (line~\ref{line:comp_sample_4} of \ComparisonAlgo), and $X_Q$ the indicator random variable of the event that \EpsilonAlgo::\epsSampler\ accepts $S_Q$ (line~\ref{line:sample_4} of \EpsilonAlgo). The outcome of \ComparisonAlgo::\Sample\ is the pair $(S_P,X_P)$, and that of \EpsilonAlgo::\epsSampler\ is the pair $(S_Q,X_Q)$. Let $\scD_P$ and $\scD_Q$ be the distributions of respectively $(S_P,X_P)$ and $(S_Q,X_Q)$. Note that $\scD_P(\cdot | X_P=1)$ and $\scD_Q(\cdot | X_Q=1)$ are the distributions of the graphlets returned by respectively \ComparisonAlgo::\Sample\ and \EpsilonAlgo::\epsSampler. Thus, our goal is to show:
%We want to show that the two distributions are $\nicefrac{\epsilon}{2}$-close conditioned on the graphlets being accepted, that is, we want:
\begin{align}
\tvd{\scD_P(\cdot | X_P=1)}{\scD_Q(\cdot | X_Q=1)} \le \frac{\epsilon}{2}
\end{align}
%Ignoring for a moment the $k^{\scO(-k)}$ factors, the strategy is the following.
%The probability that the algorithms accept the sampled graphlet is $\Omega(\epsilon^2)$.
%Therefore, to make the conditional distributions $\scO(\epsilon)$-close, we can make the unconditional distributions $\scO(\epsilon^3)$-close.
%Obviously we have to take care of the fact that the two algorithms can disagree on both the sampled graphlet and its acceptance.

Let $X_{\vee} = \max(X_P,X_Q)$ be the indicator random variable of the event that at least one algorithm accepts its graphlet.
Clearly $\Pr(X_{\vee}=1) \ge \Pr(X_P=1)$, and by Lemma~\ref{lem:comparison_algo}, $P(X_P=1) \ge \epsilon^2 \kokm$.
By the triangle inequality:
\begin{align}
\tvd{\scD_P(\cdot | X_{P}=1)}{\scD_Q(\cdot | X_{Q}=1)} \label{eqn:tvd_P_Q_2}
\le
\tvd{\scD_P(\cdot | X_{P}=1)}{\scD_P(\cdot | X_{\vee}=1)}
\\+
\tvd{\scD_P(\cdot | X_{\vee}=1)}{\scD_Q(\cdot | X_{\vee}=1)}\nonumber
\\+
\tvd{\scD_Q(\cdot | X_{Q}=1)}{\scD_Q(\cdot | X_{\vee}=1)}\nonumber
\end{align}
Let us start by bounding the middle term.
We have:
\begin{align}
\tvd{\scD_P(\cdot | X_{\vee}=1)}{\scD_Q(\cdot | X_{\vee}=1)}
& \le \Pr(S_Q \ne S_P \,|\, X_{\vee}=1) && \text{by the coupling}
\\ &\le \frac{\Pr(S_Q \ne S_P)}{\Pr(X_{\vee}=1)}
\\ &\le \frac{\Pr(S_Q \ne S_P)}{\Pr(X_P=1)} && X_{\vee}=\max(X_P,X_Q)
\\ &\le \frac{\Pr(S_Q \ne S_P)}{\epsilon^2 \kokm} && \text{Lemma~\ref{lem:comparison_algo}}
\\ &\le \frac{\epsilon^3 k^{-C_2 k}}{\epsilon^2 \kokm} && \text{see above}
\\ &= \epsilon k^{-(C_2+C_3) k} 
\label{eqn:p_SQ_ne_SP}
\end{align}
for some $C_3$ independent of $C_2$.  %% CHECK!

We bound similarly the sum of the other two terms. 
For the first term note that:
\begin{align}
\tvd{\scD_P(\cdot | X_{P}=1)}{\scD_P(\cdot | X_{\vee}=1)}
\le
\Pr(X_P = 0 \,|\, X_{\vee}=1)
\end{align}
This is true since $\scD_P(\cdot | X_{P}=1)$ is just $\scD_P(\cdot | X_{\vee}=1)$ conditioned on $X_P=1$, an event which has probability $1-\Pr(X_P = 0 \,|\, X_{\vee}=1)$.
Symmetrically, for the last term 
\begin{align}
    \tvd{\scD_Q(\cdot | X_{Q}=1)}{\scD_Q(\cdot | X_{\vee}=1)} \le \Pr(X_Q = 0 \,|\, X_{\vee}=1)
\end{align}
Thus:
\begin{align}
    & \tvd{\scD_P(\cdot | X_{P}=1)}{\scD_P(\cdot | X_{\vee}=1)} + \tvd{\scD_Q(\cdot | X_{Q}=1)}{\scD_Q(\cdot | X_{\vee}=1)}
    \\ & \quad \le \Pr(X_P = 0 \,|\, X_{\vee}=1) + \Pr(X_Q = 0 \,|\, X_{\vee}=1)
\end{align}
Now,
\begin{align}
\Pr(X_P = 0 \,|\, X_{\vee}=1) + \Pr(X_Q = 0 \,|\, X_{\vee}=1) 
&= \Pr(X_P \ne X_Q \,|\, X_{\vee}=1) 
\\&\le \frac{\Pr(X_P \ne X_Q )}{\Pr(X_{\vee}=1)}
\\&\le \frac{\Pr(X_P \ne X_Q )}{\epsilon^2 k^{-\scO(k)}} && \text{see above} \label{eq:PrXPXQ}
\end{align}
For the numerator,
\begin{align}
\label{eqn:p_S_X}
\Pr(X_Q \ne X_P) &\le \Pr(S_Q \ne S_P) + \Pr(X_Q \ne X_P \,|\, S_Q = S_P)
\\& \le \epsilon^3 k^{-C_2 k} + \Pr(X_Q \ne X_P \,|\, S_Q = S_P) && \text{see above}
\end{align}
As said, $\Pr(S_Q \ne S_P) \le \epsilon_1$.
As $X_Q$ and $X_P$ are binary, our coupling yields:
\begin{align}
\Pr(X_Q \ne X_P \,|\, S_Q = S_P) &= \big|\Pr(X_P=1\,|\, S_Q = S_P) - \Pr(X_Q=1\,|\, S_Q = S_P) \big|
\\
&\le \left|\frac{\Pr(X_P=1\,|\, S_Q = S_P)}{\Pr(X_Q=1\,|\, S_Q = S_P)} - 1\right|
\end{align}
Now, let $S$ be any realization of $S_Q$ and $S_P$. 
By construction of the algorithms, $\Pr(X_P=1\,|\, S_P = S) = \nu$, and $\Pr(X_Q=1\,|\, S_Q = S)=\min(1,\hat{\nu})$, where:
\begin{align}
\nu &= \frac{1}{p(v)\, p(S)}\frac{\beta}{Z} k^{-C_1 k}, \qquad \hat{\nu} = \frac{1}{p(v)\, \hat{p}(S)}\frac{\beta}{Z} k^{-C_1 k}
\end{align}
Therefore:
\begin{align}
\label{eq:XQXP}
\Pr(X_Q \ne X_P \,|\, S_Q = S_P)
&\le \left|\frac{\nu}{\min(1,\hat{\nu})} - 1 \right|
\le \left|\frac{\nu}{\hat{\nu}} - 1 \right|
= \left|\frac{\hat{p}(S)}{p(S)} - 1 \right|
\end{align}
where the second inequality holds since, if $\hat{\nu} > 1$, then
\begin{align}
    \left|\frac{\nu}{\min(1,\hat{\nu})} - 1 \right|
    =
    \left|\frac{\nu}{1} - 1 \right|
    < 
    \left|\frac{\nu}{\hat{\nu}} - 1 \right|
\end{align}

Now, by Lemma~\ref{lem:eps_prob}, with probability $1-\poly\frac{\epsilon}{k}$ we have $\hat{p}(S) \in (1 \pm \epsilon^3 k^{-C_2 k})p(S)$. So, if this event holds, we have $\Pr(X_Q \ne X_P \,|\, S_Q = S_P) \le \epsilon^3 k^{-C_2 k}$. If if fails, we still have the trivial bound $\Pr(X_Q \ne X_P \,|\, S_Q = S_P) \le 1$. By the law of total probability,
\begin{align}
\Pr(X_Q \ne X_P \,|\, S_Q = S_P) &\le \left(1-\poly\frac{\epsilon}{k}\right) \epsilon^3 k^{-C_2 k} + \poly\frac{\epsilon}{k} = \scO(\epsilon^3 k^{-C_2k})
\end{align}
Applying these two bounds to the right-hand side of~\eqref{eqn:p_S_X}, we obtain:
\begin{align}
\label{eqn:pSQ}
\Pr(X_Q \ne X_P) &= \scO(\epsilon^3 k^{-C_2 k})
\end{align}
Going back to~\eqref{eq:PrXPXQ}, we obtain:
\begin{align}
    \Pr(X_P = 0 \,|\, X_{\vee}=1) + \Pr(X_Q = 0 \,|\, X_{\vee}=1) = \scO\left(\frac{\epsilon^3 k^{-C_2 k}}{\epsilon^2 k^{-\scO(k)}}\right) = \scO(\epsilon)
\end{align}
By taking this bound together with~\eqref{eqn:p_SQ_ne_SP}, we conclude that:
\begin{align}
    \tvd{\scD_P(\cdot | X_{P}=1)}{\scD_Q(\cdot | X_{Q}=1)} = \scO(\epsilon)
\end{align}
which we can bring below $\frac{\epsilon}{2}$ by adjusting the constants.
This concludes the proof.
\end{proof}

\begin{lemma}
\label{lem:wrap_sample_time}
Suppose that the preprocessing of \EpsilonAlgo$(G,\epsilon)$ succeeds (Lemma~\ref{lem:approx_order}). Then, each invocation of \epsSampler$(\,)$
%Suppose after running \GraphSort$(G, \nicefrac{\epsilon}{2})$ the properties of Lemma~\ref{lem:approx_order} hold.
has expected running time $k^{\scO(k)} \epsilon^{-8-\frac{4}{k-1}} \lg \frac{1}{\epsilon}$.
\end{lemma}
\begin{proof}
First, we bound the expected number of rounds of \EpsilonAlgo::\EpsilonSampleS. Recall $X_P$ and $X_Q$ from the proof of Lemma~\ref{lem:wrap_sample_dist}. Note that $\Pr(X_Q = 1) \ge \Pr(X_P = 1) - \Pr(X_Q \ne X_P)$. Moreover, the proof of Lemma~\ref{lem:wrap_sample_dist} showed $\Pr(X_Q \ne X_P) = \scO(\epsilon^3 k^{-C_2 k})$. Therefore, $\Pr(X_Q = 1) \ge \Pr(X_P = 1) - \scO(\epsilon^3 k^{-C_2 k})$. However, by Lemma~\ref{lem:comparison_algo}, $\Pr(X_P=1) \ge k^{-C_3} \epsilon^2$ for some constant $C_3$. Therefore, $\Pr(X_Q = 1) \ge \epsilon^2 k^{-C_3 k} - \epsilon^3 k^{-C_2 k} = \kokm \epsilon^2$. So the expected number of round performed by \EpsilonAlgo::\EpsilonSampleS\ is bounded by $\kok \epsilon^{-2}$.

Now we bound the expected time spent in each round. By Lemma~\ref{lem:apx_sample_cost}, and as $\gamma=\epsilon^3 k^{-C_2 k}$ and $\alpha=\beta^{\frac{1}{k-1}}\frac{1}{6k^3}$ and $\beta=\frac{\epsilon}{2}$, \EpsilonSampleS$(G, v, \alpha, \beta, \gamma)$ has expected running time at most:
\begin{align}
    \scO\left(\frac{k^9}{\gamma^2 \alpha^4} \lg \frac{1}{\gamma}\right)
    &=
    \kok \epsilon^{-6-\frac{4}{k-1}} \lg \frac{1}{\epsilon}
\end{align}
Note that the bound holds at each round, regardless of past events. Using Lemma~\ref{lem:eps_prob}, one can show the same bound holds for the running time \epsProbCompute$(G, S,  \alpha, \beta, \rho)$, where $\rho=\epsilon^3 k^{-C_2 k}$.

Therefore, the total expected running time satisfies: 
\begin{align}
\E[T] &\le \kok \epsilon^{-2} \cdot \kok \epsilon^{-6-\frac{4}{k-1}} \lg \frac{1}{\epsilon}
= \kok \epsilon^{-8-\frac{4}{k-1}} \lg \frac{1}{\epsilon}
\end{align}
which concludes the proof.
\end{proof}

\section{Conclusions}
We have shown that, starting from just sorting a graph in linear time, one can overcome the usual inefficiency of rejection sampling of graphlets. This idea yields the first efficient uniform and $\epsilon$-uniform graphlet sampling algorithms, with preprocessing times $\scO(|G|)$ and $\scO(|V(G)| \lg |V(G)|)$. These are the first algorithms with strong theoretical guarantees for these problems in a long line of research that spans the last decade.
Due to their simplicity, we believe that our algorithms are amenable to being ported in parallel, distributed, or dynamic settings; these are all directions for future research.
We also leave open the problem of determining whether $\Omega(|G|)$ operations are necessary for uniform graphlet sampling when $|G| = \omega(|V(G)|)$; a positive answer would imply the optimality of our uniform sampling algorithm.

\section*{Acknowledgements}
Part of this work was done while the author was at the Sapienza University of Rome. The author was partially supported by Google under the Focused Award ``Algorithms and Learning for AI'' (ALL4AI), by the Bertinoro International Center for Informatics (BICI), by the European Research Council under the Starting Grant DMAP 680153, and by the Department of Computer Science of the Sapienza University of Rome under the grant Dipartimenti di Eccellenza 2018-2022''.

\nocite{Bressan19IPEC,Bressan21Algo} 
\bibliographystyle{abbrv}
\bibliography{biblio.bib}

\appendix
\appendix

\section{Ancillary results}
\begin{lemma}
\label{lem:count_bound}
Let $G=(V,E)$ be any graph, and for any $v \in V$ let $N_v$ be the number of $k$-graphlets of $G$ containing $v$.
If $N_v > 0$, then:
\begin{align}
N_v \ge \frac{(d_v)^{k-1}}{(k-1)^{k-1}} = (d_v)^{k-1} k^{-\scO(k)}
\end{align}
Moreover, if $d_u \le \Delta$ for all $u \in G$, then:
\begin{align}
N_v \le (k-1)!(\Delta)^{k-1} = (\Delta)^{k-1} k^{\scO(k)}
\end{align}
\end{lemma}
\begin{proof}
For the lower bound, if $d_v \le k-1$ then $\frac{(d_v)^{k-1}}{(k-1)^{k-1}} \le 1$, so if $N_v \ge 1$ then $N_v \ge \frac{(d_v)^{k-1}}{(k-1)^{k-1}}$.
If instead $d_v > k-1$, then $N_v \ge {d_v \choose k-1}$ since any set of nodes formed by $v$ and $k-1$ of its neighbors is connected.
However $N_v \ge {d_v \choose k-1} \ge \frac{(d_v)^{k-1}}{(k-1)^{k-1}}$ since ${a \choose b} \ge \frac{a^b}{b^b}$ for all $a \ge 1$ and all $b \in [a]$.

For the upper bound, note that we can construct a connected subgraph on $k$ nodes containing $v$ by starting with $S_1=\{v\}$ and at every step $i=1,\ldots,k-1$ choosing a neighbor of $S_i$ in $G \setminus S_i$.
Since each $u \in G$ has degree at most $\Delta$, then $S_i$ has at most $i \Delta$ neighbors.
Thus the total number of choices is at most $\prod_{i=1}^{k-1}i \Delta = (k-1)! (\Delta)^{k-1} $.
\end{proof}

\section{Proof of Theorem~\ref{thm:counting_1}}
\label{apx:counting}
We start by running the preprocessing phase of \UniformAlgo. Let $N_k = \sum_{v \in V} |B(v)|$ be the total number of $k$-graphlet occurrences in $G$. We compute an estimate $\hat N_k$ of $N_k$ such that $|\hat{N}-N_k| \le \epsilon_0 N$ with probability at least $1-\frac{\delta}{2}$. To this end, for each $v \in V$ such that $\Va(v) \ne \emptyset$, we estimate $|\Va(v)|$ up to a multiplicative error $(1 \pm \epsilon_0)$ with probability $1-\frac{\delta}{2n}$, as detailed below. By a union bound, setting $\hat N_k$ to the sum of all those estimates will satisfy the bound above.

To estimate $|\Va(v)|$, we run the sampling routine of \UniformAlgo\ over bucket $\Va(v)$. However, after $S$ is sampled, instead of rejecting it randomly, we return the probability $p(S)$ computed by \computeP(G,S). By Lemma~\ref{lem:prob_compute}, $p(S)$ is exactly the probability that $S$ is sampled. Thus, if $X$ is the random variable denoting the output value of this modified routine, we have:
\begin{align}
    \E[X] = \sum_{S \in \Va(v)} p(S) \cdot \frac{1}{p(S)} = |\Va(v)|
\end{align}
It remains to apply concentration bounds.
To this end, note that $X \le \kok \E[X]$ by Lemma~\ref{lem:subgraph_sample}.
Thus, $X \in [0, \kok \E[X]]$.
Therefore, by averaging over $\ell$ independent samples of $X$, we obtain:
\begin{align}
\Pr\left(\left|\frac{1}{\ell}\sum_{i=1}^{\ell} X_i \,-\, \E[X]\right| > \epsilon_0 \E[X] \right)
< 2\exp\left(-\frac{(\epsilon_0 \E[X])^2 \ell}{(\kok \E[X])^2}\right)
= 2\exp\left(-\frac{\epsilon_0^2 \ell}{\kok}\right)
\end{align}
Therefore, our guarantees are achieved by setting $\ell = \kok \epsilon_0^{-2} \ln \frac{2n}{\delta}$.
Since we have at most $n$ nonempty buckets, to estimate $\hat N_k$ we use a total of $\kok n \epsilon_0^{-2} \ln \frac{2n}{\delta}$ samples.

Next, we estimate the graphlet frequencies via the sampling routine of \UniformAlgo.
For every distinct (up to isomorphism) $k$-node simple connected graph $H$, let $N_H$ be the number of distinct $k$-graphlet occurrences of $H$ in $G$.
Clearly, $\sum_{H}N_H = N_k$.
Let $f_H=\frac{N_H}{N_k}$ be the relative frequency of $H$.
Now, we take $\poly(k)\, \frac{4}{\epsilon_1^2} \ln \frac{1}{\delta}$ independent uniform samples.
By standard concentration bounds, we obtain an estimate $\hat{f}_H$ of $f_H$ such that $|\hat{f}_H-f_H| \le \frac{\epsilon_1}{2}$ with probability at least $1-\delta_1^{-\poly(k)}$.
Since there are $2^{\poly(k)}$ distinct $k$-node (connected) graphs, by a union bound we obtain such an estimate $\hat{f}_H$ for all $H$ simultaneously with probability $1-\frac{\delta}{2}$.

Now, for all $H$, we set $\hat{N}_H = \hat{N_k} \hat{f}_H$.
By a union bound, with probability at least $1-\delta$ we have simultaneously for all $H$:
\begin{align}
    \hat{N}_H - N_H &= \hat{N_k} \hat{f}_H - N_H
    \\ &\le N_k(1+\epsilon_0) \, \left(f_H+\frac{\epsilon_1}{2}\right) - N_H
    \\ &= (1+\epsilon_0) N_H + N_k (1+\epsilon_0)\frac{\epsilon_1}{2} - N_H
    \\ &\le \epsilon_0 N_H + \epsilon_1 N_k
\end{align}
on the one hand, and similarly, $\hat{N}_H - N_H \ge -\epsilon_0 N_H - \epsilon_1 N_k$ on the other hand.
Therefore $|\hat{N}_H - N_k| \le \epsilon_0 N_H + \epsilon_1 N_k$ with probability at least $1-\delta$ for all $H$ simultaneously, as desired.
 
The running time is given by (i) the preprocessing phase, which takes time $\scO(k^2+m)$; (ii) $\kok n \epsilon_0^{-2} \ln \frac{2n}{\delta} + \poly(k)\, \frac{4}{\epsilon_1^2} \ln \frac{1}{\delta}$ samples, each one taking time $\kok \log \Delta$ as per Theorem~\ref{thm:uniform}.
This gives a total running time of:
\begin{align}
    \scO(k^2n+m) + \left(\kok n \epsilon_0^{-2} \ln \frac{2n}{\delta} + \poly(k)\, \frac{4}{\epsilon_1^2} \ln \frac{1}{\delta}\right) \kok \log \Delta
\end{align}
which is in $\scO(m) + \kok \left( n \epsilon_0^{-2} \ln \frac{n}{\delta} + \epsilon_1^{-2} \ln \frac{1}{\delta} \right) \log \Delta$. The proof is complete.
\section{Epsilon-uniform sampling via color coding}
\label{apx:cc}
We show how to use the color coding algorithm of~\cite{Bressan&2017} in a black-box fashion to perform $\epsilon$-uniform sampling from $G$.
The overhead in the running time and space is $2^{\scO(k)}\scO\big(\lg\frac{1}{\epsilon}\big)$, and the overhead in the sampling time is $2^{\scO(k)}\scO\big(\big(\lg\frac{1}{\epsilon}\big)^2\big)$.

First, we perform $\ell=\scO\big(e^k\lg\frac{1}{\epsilon}\big)$ independent runs of the preprocessing phase of the algorithm of~\cite{Bressan&2017}, storing all their output count tables. This gives a time-and-space $\scO\big(e^k\lg\frac{1}{\epsilon}\big)$ overhead with respect to~\cite{Bressan&2017}. In each run, any graphlet $g$ has probability $\frac{k^k}{k!} \ge e^{-k}$ of becoming colorful. Thus, with $\scO\big(e^k\lg\frac{1}{\epsilon}\big)$ independent runs, $g$ is colorful with probability $1-\poly \epsilon$ in at least one run, and appears in the respective count table. As shown in~\cite{Bressan&2017}, for each run $i=1,\ldots,\ell$ one can estimate, within a multiplicative $(1\pm\epsilon)$ factor, the number of colorful graphlets $N_i$, using $\scO\big(\frac{\kok}{\epsilon^2}\big)$ samples. %(This requires estimating the average number of spanning trees per-graphlet in that run, and multiply by the number of detected trees, which is known).
In time $\scO\big(\frac{\kok}{\epsilon^2}\lg \frac{1}{\epsilon}\big)$, we can do so for all runs with probability $1-\poly \epsilon$. This concludes the preprocessing phase.

For sampling, we choose a random run $i \in [\ell]$ with probability proportional to the estimate of $N_i$. Then, we draw a graphlet from that run uniformly at random using the sampling phase of~\cite{Bressan&2017}. This yields a graphlet uniformly at random from the union of all colorful graphlets in all runs. Thus, the probability that a specific graphlet $g$ is sampled is proportional to the number of runs $\ell(g)$ where $g$ is colorful, which we can compute by looking at the colors assigned to $g$ by every run in time $\ell=\scO\big(e^k\lg\frac{1}{\epsilon}\big)$. Then, we accept $g$ with probability $\frac{1}{\ell(g)} \ge \frac{1}{\ell}$. Therefore we need at most $\ell = \scO\big(e^k\lg\frac{1}{\epsilon}\big)$ trials in expectation before a graphlet is accepted. This gives an overhead of $\scO\big(e^k\lg\frac{1}{\epsilon}\big)^2$ in the sampling phase. This construction can be derandomized using an $(n,k)$-family of perfect hash functions %--- informally, a family of colorings $C_1,\ldots,C_{\ell}$ over $V$ such that every $k$-node subset in $V$ is colorful under some $C_i$.
of size $\ell = 2^{\scO(k)} \lg n$, see \cite{Alon&1995}.
% and the $\lg n$ factor is unavoidable by a pigeonhole principle.
This derandomization would increase the time and space of the preprocessing by a factor $\lg n$, but we would still need to estimate the number of graphlets in each run, so the final distribution would still be non-uniform.

\end{document}